\documentclass[12pt]{article}
\usepackage[margin=1.5in]{geometry}
\usepackage[usenames]{color}
\usepackage{array}
\usepackage{cite}
\usepackage{amssymb}
\usepackage[ruled]{algorithm2e}
\usepackage{algpseudocode}
\usepackage{graphicx}
\usepackage{soul}
\usepackage{amsmath}
\usepackage{amssymb}
\usepackage{amsthm}
\usepackage{mathrsfs}
\usepackage{cases}
\usepackage{setspace}
\usepackage{enumerate}
\usepackage{subfigure}
\usepackage{float}
\usepackage{listings}
\usepackage{xcolor}
\usepackage{appendix}
\usepackage{pdfpages}
\usepackage{abstract}
\usepackage{longtable,booktabs}
\usepackage{rotating}
\usepackage{lscape}
\usepackage{multicol}
\usepackage{makecell}
\usepackage{natbib}
\usepackage{mathtools}
\usepackage{multirow}
\usepackage{threeparttable}
\usepackage{fancyhdr}
\usepackage{xparse}
\usepackage[bookmarks=false,hypertexnames=false]{hyperref}\hypersetup{colorlinks=true, linkcolor=blue, filecolor=gray, urlcolor=blue, citecolor=black}
\usepackage{caption}
\usepackage{authblk}

\newtheorem{Corollary}{Corollary}
\newtheorem{Proposition}{Proposition}
\newtheorem{Assumption}{Assumption}
\newtheorem{Lemma}{Lemma}
\newtheorem{Definition}{Definition}
\newtheorem{Theorem}{Theorem}

\newtheorem{manualassumptioninner}{Assumption}
\newenvironment{manualassumption}[1]{%
  \renewcommand\themanualassumptioninner{#1}%
  \manualassumptioninner
}{\endmanualassumptioninner}

\theoremstyle{remark}
\newtheorem{Remark}{Remark}

\newcommand{\convd}{\stackrel{d}{\to}}
\newcommand{\convp}{\stackrel{p}{\to}}
\newcommand{\sumn}{\sum_{i=1}^n}
\newcommand{\supn}{\sup_{i\in[n]}}

\newcommand{\ind}{\text{ind}}
\newcommand{\ot}{{\otimes2}}
\newcommand{\ic}{{i\cdot}}

\newcommand{\R}{\mathbb{R}}
\newcommand{\E}{{\mathbb{E}}}

\newcommand{\PP}{{\mathbb{P}}}

\newcommand{\hatqp}{\widehat{q}^\prime}

\newcommand{\tildegp}{\widetilde{g}^\prime}

\newcommand{\hatgp}{\widehat{g}^\prime}
\newcommand{\qp}{q^\prime}
\newcommand{\qpp}{q^{\prime\prime}}

\newcommand{\gp}{g^\prime}

\newcommand{\hatsigma}{\widehat{\sigma}}
\newcommand{\hatbeta}{\widehat{\beta}}
\newcommand{\hateta}{\widehat{\eta}}
\newcommand{\hatkappa}{\widehat{\kappa}}
\newcommand{\hatomega}{\widehat{\omega}}
\newcommand{\hattheta}{\widehat{\theta}}

\newcommand{\hatvarphi}{\widehat{\varphi}}

\newcommand{\hats}{\widehat{s}}
\newcommand{\hatH}{\widehat{H}}

\newcommand{\hatm}{\widehat{m}}
\newcommand{\hatW}{\widehat{W}}
\newcommand{\hatv}{\widehat{v}}
\newcommand{\hatU}{\widehat{U}}
\newcommand{\hatF}{\widehat{F}}

\newcommand{\lep}{\lesssim_p}
\newcommand{\gep}{\gtrsim_p}

\newcommand{\logpM}{\log (pM)}

\renewcommand{\hat}{\widehat}
\renewcommand{\tilde}{\widetilde}

\begin{document}
\title{Inference for Nonlinear Endogenous Treatment Effects Accounting for High-Dimensional Covariate Complexity\thanks{\noindent The authors are in alphabetical order. We thank Xiaohong Chen, Whitney Newey, and seminar/workshop/conference participants at various places for helpful comments. Emails: michaelqfan@gmail.com (Q. Fan), zijguo@stat.rutgers.edu (Z. Guo), zwmei@link.cuhk.edu.hk (Z. Mei), czhang@stat.rutgers.edu (C.-H. Zhang). }
	}
	\author[a]{Qingliang Fan}
	\author[b]{Zijian Guo}
	\author[a]{Ziwei Mei}
	\author[b]{Cun-Hui Zhang}	
	
	\affil[a]{Department of Economics, The Chinese University of Hong Kong}
	\affil[b]{Department of Statistics,  Rutgers University}

\date{\today}
\maketitle

\begin{abstract} 
Nonlinearity and endogeneity are prevalent challenges in causal analysis using observational data. This paper proposes an inference procedure for a nonlinear and endogenous \emph{marginal effect function}, defined as the derivative of the nonparametric treatment function, with a primary focus on an additive model that includes high-dimensional covariates. Using the control function approach for identification, we implement a regularized nonparametric estimation to obtain an initial estimator of the model. Such an initial estimator suffers from two biases: the bias in estimating the control function and the regularization bias for the high-dimensional outcome model. Our key innovation is to devise the \emph{double bias correction} procedure that corrects these two biases simultaneously. Building on this debiased estimator, we further provide a confidence band of the marginal effect function. Simulations and an empirical study of air pollution and migration demonstrate the validity of our procedures.
\end{abstract}
\clearpage
\onehalfspacing
\section{Introduction} 
Nonlinear treatment effects with endogeneity are prevalent in empirical economic studies. The econometric model could be high-dimensional with the increasing availability of rich datasets that include many potential covariates. For example, in \citet{gopalan2021home}, the effect of home equity on labor income is found to be nonlinear, while as a measure of home equity, the loan-to-value (LTV) ratio is endogenous. Their study addressed the issue of endogeneity using the synthetic Loan-To-Value (LTV) ratio as an instrumental variable (IV). This ratio is constructed using synthetic loans, the original LTV level, and changes in the house price index. Other control variables might include the original loan amount, purchase price, loan balance, job tenure, and homeowner's age, among many other individual characteristics. Another example, which is used as our later empirical study, is the potential nonlinear effect of environmental pollution on migration \citep{chen2022}. Air pollution is endogenous due to unobserved common economic factors affecting both migration and pollution. Thermal inversion is a popular IV in related environmental economic studies. We consider high-dimensional controls, including county-level income, expenditure, investment in health and education, etc. A third example is the nature of the job ladder in terms of working time and career advancement. \cite{Gicheva2013} found a positive but nonlinear relationship between weekly working hours and hourly wage growth using data\footnote{Data were from the 1979 cohort of National Longitudinal Survey of Youth, U.S. Bureau of Labor Statistics, and GMAT Registrant Survey.}  from the high-end labor market. Potential covariates include gender, race, age, the number of children under 18, marital status, the mother's education, the college GPA, the GMAT score, working experience, and whether a person is enrolled in school. The marginal effects seem small at low levels of working hours and increase at high levels like 50 hours or above. However, some missing variables, including parents' involvement in learning and human capital formation, are likely to be correlated with working hours, which causes bias in the results.

Motivated by the real applications above, we propose new estimation and inference procedures for nonlinear treatment functions with high-dimensional covariates. For observations indexed by $i=1,2,\dots,n$ with a scalar outcome $Y_i$, a scalar endogenous treatment $D_i$, high-dimensional covariates $X_\ic$ and finite-dimensional instruments $Z_\ic$, we consider the model satisfying the conditional moment restriction $\mathbb{E}(Y_i-g(D_i)-X_\ic^{\top}\theta|X_\ic,Z_\ic)=0$,  where the function $g$ and high-dimensional coefficients $\theta$ are unknown. Granted that our methodology applies to the original nonlinear treatment effect function $g(\cdot)$, we focus on the inference for the \emph{marginal effect function}, defined as the derivative function  $g'(\cdot)$. It is interpreted as the ``marginal effect"  at given levels of the treatment variable, analogous to the slope coefficient in a linear causal model. The valid inference for the nonlinear marginal effect function, taking into account high-dimensional covariates, is essential for empirical researchers and policy-makers.

\par This paper uses the control function approach \citep{newey1999} for model identification. It is well known in the literature (e.g., \citealp{neweypowell89,Florens03,aichen2003})  that a fully nonparametric IV model encounters an ill-posed inverse problem and the ``curse of dimensionality". We thus consider an additive model specified in Section \ref{sec21}, which is already a challenging problem in the presence of high-dimensional covariates.

%%%%%%%%%%%%%%%%%%%%%%%%

\subsection{Main Results and Contributions}
% \par The proposed methodology has the following features:
% \begin{enumerate}
% \item It jointly addresses endogeneity, nonlinearity, and high dimensionality, which are commonly encountered issues in empirical studies. The high-dimensional covariates bring significant challenges that are absent from the classical nonparametric instrumental variable regression.
% \item It specifies two sources of biases arising from the penalized estimators of the outcome model and control function. Our proposal provides an innovative double bias correction procedure for this issue.
% \item Granted that our methodology applies to the nonlinear treatment function itself, it focuses on the estimation and inference of its derivative function, which represents the marginal lik of treatment on the outcome.
% \end{enumerate}
\par Under the control function setup, our procedure mainly includes the following steps: First, estimate the control function using the LASSO regularization \citep{tibshirani1996regression}; Second, construct an initial estimator of the marginal effect function using LASSO again; Third, correct the regularization bias in the initial estimator for inference. Our constructed initial estimator suffers from two sources of regularization bias from estimating both the control function in the first step and the nonlinear treatment in the second step. The former causes new obstacles to the inference problem, which is absent from the high-dimensional inference in additive models without endogeneity.  
\par To address this difficulty, we propose a novel debiasing procedure for valid inference of the marginal effect function in the presence of endogeneity, nonlinearity, and high dimensionality. We call it the \emph{double bias correction} because it corrects two sources of biases. Furthermore, we construct a uniform confidence band for the nonlinear marginal effect in the manner of \citet{lu2020kernel}, adopting the techniques of multiplier bootstrap \citep{chernozhukov2014anti,chernozhukov2014gaussian}. 
%Theories on the consistency of the initial LASSO estimator and the honesty of the confidence band are provided in Section \ref{sec: inference theory NL}. 
\par The simulation study supports the validity of our estimator and confidence band of the endogenous marginal effect in the presence of high-dimensional covariates. In our empirical analysis of migration and air pollution, we find new results in addition to those of \citet{chen2022}, who consider only the linear effect. Specifically, using thermal inversion as IV, the effect of pollution on migration is found to be insignificant when the pollution level is very low or moderately high. The effect is significant when pollution is worse than the ``good air quality'' level but below the medium level or at high levels. The magnitude of the effect increases at very high levels of pollution. In Section \ref{sec: emp}, we provide an explanation of this result based on the  \emph{reference-dependent} principle in behavioral economics. 

The main contributions of the paper are summarized as follows. \begin{enumerate}
    \item In a high-dimensional and endogenous setting, we propose the double bias correction procedure to correct for regularization biases from both the outcome regression and the nonlinear control function estimation. %The double bias correction is relevant to, but more challenging than the debiased LASSO in the literature \citep{van2014asymptotically,zhang2014confidence,javanmard2014confidence}; See Section \ref{sec: inference g NL} for details.  
    \item We provide a uniform confidence band for the nonlinear marginal effect function. The development of a uniform confidence band for the nonlinear endogenous effects under high dimensions appears to be novel in the literature. 
\end{enumerate}%As we mentioned above, it handles two sources of regularization bias in high-dimensional nonlinear models, which is absent from the regression models without endogeneity. 

\subsection{Literature Review}
%\cite{Matzkin94} provided a survey of the identification of nonlinear models with endogeneity. \cite{newey1999} presented a two-step nonparametric estimator. They noted that the control function method and the conditional independence assumption \citep{neweypowell89} cannot subsume each other.  \cite{blundell03} discussed endogeneity in nonparametric and semiparametric models.\cite{blundell07} proposed nonparametric estimation for Engle curves. \cite{Horowitz2011} described the control function approach, which can be used in many empirical studies. \citet{chen2018optimal} developed optimal sup-norm rates and uniform confidence bands for nonlinear functionals of nonparametric IV regression. \cite{chen2021adaptive} further proposed data-driven methods for the choice of Sieve dimension and uniform inference of the nonlinear treatment function and its derivatives to improve the efficiency of confidence bands. 

Our research connects to the literature on nonparametric estimation, causal inference, and high-dimensional models. Endogeneity in general nonparametric models was considered in \cite{neweypowell89, Matzkin94, blundell03, Hall05, darolles2011nonparametric}. \cite{Newey1990} studied the efficiency of a nonparametric reduced form function for the instrumental variables. The control function method for nonparametric IV models was widely studied in \cite{newey1999,Horowitz2011,Wooldridge05, Guo16,su2008local,ozabaci2014additive,lee2007endogeneity,aghion2013innovation}. 
%In the more recent development of nonparametric IV models, 
Recently, \citet{chen2018optimal}, \cite{chen2021adaptive} developed optimal sup-norm rates and uniform confidence band for the nonlinear treatment function and its derivatives. \cite{Breunig2022adaptive} considered minimax adaptive estimation of quadratic functionals in nonparametric IV models. \cite{babii2020honest} developed honest inference for a wide class of ill-posed regression models, including the nonparametric IV model. Distinguished from the above literature focusing on low-dimensional models, we specialize in the valid inference of nonlinear marginal effect that addresses the complexity from high-dimensional covariates.  \cite{chernozhukov2022locally,chernozhukov2022automatic} established automatic debiasing estimators of functionals like average treatment effect for general nonlinear models. Our method differs in that we use the control function approach for identification, and develop a uniform confidence band for the marginal effect function. %rather than a point-wise confidence interval of a functional.  
%Granted these nonseparable models are more general than ours, we specialize in the valid inference of nonlinear effects in the high-dimensional regime. \Ziwei{Please check if we need the following sentence.} \Zijian{I do not feel this is a proper way of saying our novelty.} None of the mentioned literature has simultaneously addressed the double sources of regularization biases in a high-dimensional two-stage NPIV regression. \cite{chetverikov2017} derived a non-asymptotic error bound for nonparametric IV model under monotonicity. \cite{chen2015} provided a unified theory about sieve Wald and quasi-likelihood ratio inferences on functionals of semi/nonparametric conditional moment restrictions. \cite{Angrist2022} discussed machine learning in applied (labor) economics with many instruments and called for caution regarding interpretability in case-specific studies.
\par A surge of machine learning methods provides solutions to the estimation and inference of linear treatment effects with rich observational features. The popular double machine learning (DML) by \citet{chernozhukov2018double} corrects the regularization bias from estimating the linear treatment effect with high-dimensional covariates. \cite{belloni2014} proposed a post-selection inference procedure for high-dimensional linear treatment effects models. \cite{fan2020endogenous} considered a model with a mixture of controls and instruments. \citet{fan2022testing} proposed an overidentification test for high-dimensional linear IV models. Unlike the above studies, we focus on the nonlinear effect.
\par Another strand of literature focuses on the inference for high-dimensional nonlinear models without endogeneity. The estimation of high-dimensional partially linear models was considered by \cite{wang2010estimation,muller2015partial,yu2016minimax}. The estimation of high-dimensional additive models has also been frequently discussed \citep{meier2009high,huang2010variable,koltchinskii2010sparsity,suzuki2013fast,yuan2016minimax,tan2019doubly}. \citet{lu2020kernel} extended the procedure of \citet{javanmard2014confidence} to the high-dimensional additive model. \citet{ kozbur2021inference} and \citet{gregory2021statistical} proposed inferential procedures for additive models through post-selection or debiased estimators. \citet{su2019non} proposed bootstrapping inference for high-dimensional nonseparable models. \citet{guo2022decorrelated} used a decorrelated local linear estimator for inference of the first-order derivative of the target function. \citet{ning2023} proposed estimation and inference for high-dimensional partially linear models with estimated outcomes. The complexity of endogeneity distinguishes our problem from the abovementioned literature. 
\par \noindent \textbf{Notations.} We use ``$\convp$'' and ``$\convd$'' to denote convergence in probability and distribution, respectively. The phrase ``with probability approaching one as $n\to\infty$" is abbreviated as ``w.p.a.1". We use ``$\E(\cdot)$" to denote the expectation and ``$\E_{\mathcal{F}}(\cdot)$" to denote the conditional expectation $\E(\cdot|\mathcal{F})$ for any $\sigma$-field $\mathcal{F}$. For any positive sequences $a_n$ and $b_n$, ``$a_n\lesssim b_n$'' means there exists some constant $C$ such that $a_n\leq Cb_n$, ``$a_n\gtrsim b_n$'' means $b_n\lesssim a_n$, and ``$a_n\asymp b_n$'' indicates $a_n\lesssim b_n$ and $b_n\lesssim a_n$.  We use $[n]$ for some $n\in\mathbb{N}$ to denote the integer set $\{1,2,\cdots,n\}$. The floor and ceiling functions are $\lfloor\cdot\rfloor$ and $\lceil\cdot\rceil$, respectively. For a $p$-dimensional vector $x=(x_{1,}x_{2},\cdots,x_{p})^{\top}$, the number of nonzero entries is $\|x\|_0$, the $L_{2}$ norm is $\left\Vert x\right\Vert _{2}=\sqrt{\sum_{j=1}^{p}x_{j}^{2}}$, the $L_{1}$ norm is $\left\Vert x\right\Vert_{1}=\sum_{j=1}^{n}\left|x_{j}\right|$, and its maximum norm is $\|x\|_{\infty}=\max_{j\in[p]}|x_{j}|$. For a $p\times r$ matrix $A=(A_{ij})_{i\in[p],j\in[r]}$, we define the maximum norm $\|A\|_\infty = \max_{i,j}|A_{i,j}|$, $L_2$ norm $\|A\|_2 = \sqrt{\lambda_{\max}(A^\top A)}$ and the $L_1$ norm $\|A\|_1 = \max_{j\in[r]}\sum_{i\in[p]}|A_{ij}|$. For a square-integrable function $f(\cdot)$, define its $L_2$ norm as $\|f\|_2 = (\int f^2(x) dx)^{1/2}$. We use $0_p$ to denote the $p\times 1$ null vector. The indicator function is ${\textbf 1}(\cdot)$.   For any function $g(\cdot)$, its first-order derivative function is $\gp(\cdot)$. 
\par \noindent \textbf{Paper Organization.} The remainder of the paper is organized as follows. In Section \ref{sec: Model}, we introduce the model and the main methodology. Section \ref{sec: inference theory NL} provides the main theoretical justification. In Section \ref{sec: sim}, we demonstrate the finite-sample performance of our nonlinear effect estimator and uniform confidence band. Section \ref{sec: emp} provides an empirical example. Section \ref{sec: conc} concludes the paper. Technical proofs and additional simulation results are provided in the Appendix.

\section{Model and Methodology}
\label{sec: Model}
\subsection{Model}\label{sec21}
We consider the following additive model:
\begin{align}
 Y_i&=g(D_i)+X_{i\cdot}^{\top}\theta+u_i,
 \label{eq: y model}\\
 D_i&= \sum_{\ell=1}^{p_z} \psi_\ell(Z_{i\ell}) + X_{i\cdot}^{\top}\varphi  + v_i,
 \label{eq: D model}
\end{align}
where $Y_i\in\R$ is the outcome variable, $D_i\in\R$ is the univariate treatment variable, $g(\cdot)$ is the unknown \emph{treatment function} of interest,  $X_{i\cdot}\in\R^{p}$ represents the high-dimensional covariates\footnote{Without loss of generality, $X_\ic$ may include nonlinear terms such as the polynomials of some elementary covariates. The extension of \eqref{eq: y model} and \eqref{eq: D model} to the nonparametric additive model is possible.}, $Z_{i\cdot}\in\R^{p_z}$ denotes the instrumental variables, and $(u_i,v_i)^\top$ are unmeasured errors. Due to the existence of unmeasured confounders, $u_i$ might be correlated with $v_i$, leading to the endogeneity of $D_i$. In the above model \eqref{eq: D model}, we consider the additive nonlinear relation between $D_i$ and $Z_i$, where the $\psi_\ell(\cdot)$'s, for $\ell=1,2,\cdots,p_z$, are unknown functions. It accommodates the common linear reduced form with  $\psi_\ell(Z_{i\ell})=Z_{i\ell}\psi_\ell$. The additive model in \eqref{eq: D model} has been widely considered in the literature  \citep{belloni2012sparse,fan2018nonparametric,ozabaci2014additive}. We assume that the dataset $\mathscr{D}_n := \{Y_i,D_i,X_\ic,Z_\ic\}_{i\in[n]}$ is independently and identically distributed (i.i.d.). We allow the dimension $p$ of covariates $X_{i\cdot}$ to be larger than $n$ and assume that the high-dimensional nuisance parameters $\theta,\varphi\in{\mathbb{R}}^{p}$ are sparse with the sparsity level $s=\max\{\|\theta\|_0,\|\varphi\|_0\}$. We assume that $p_z$ is a fixed number ($p_z=1$ in many applications).  
\par The main focus of this paper is the statistical inference for the derivative function $\gp(\cdot)$. % representing the marginal effect of $D$ on $Y$ whose functional form is unknown.
 Formally, we call $\gp(\cdot)$ the \emph{marginal effect function} or marginal effect for short\footnote{The marginal effect function that we define here (see also \citealt{chen2021adaptive}) should be distinguished from the term ``marginal treatment effect'' as in \cite{heckman05}. The latter is defined as the average causal effect of $D$ on $Y$ for individuals with some observed $X = x$ and unobserved $U = u$, where $U$ is a continuously distributed random variable satisfying the monotonicity condition of \cite{imbens94} with a binary treatment $D$. See the excellent survey by \cite{mogstad2018}.}. This includes $g(D_i)=D_i\beta$ as a special case, where the derivative $g^\prime(D_i)=\beta$ is the homogeneous treatment effect in linear causal models.
%\par  %In a special case where $\psi_\ell(Z_{i\ell})=Z_{i\ell}\psi_\ell$, we have the common linear reduced form equation.%Our setting accommodates a linear reduced form in \eqref{eq: D model}, whereas the current model allows more flexibility in the functional form of the effects of instruments $Z_\ic$ on $D_i$.%

\par The identification of $g(D_i)$ relies on the control function approach. Specifically, we impose the condition that 
\begin{equation}\label{eq: control function}
    \mathbb{E}(u_i|v_i,X_{i\cdot},Z_\ic)=\mathbb{E}(u_i|v_i), 
\end{equation}
which is widely used in the literature \citep{florens2008identification,newey1999,imbens2009identification,su2008local,Horowitz2011,Wooldridge05}. A simple sufficient condition for (\ref{eq: control function}) is that $(u_i,v_i)$ is independent of the exogeneous covariates $X_\ic$ and the instruments $Z_\ic$.
\par Define $q(v_i) := \mathbb{E}(u_i|v_i)$. The function $q(v_i)$ is called the control function in the literature; see \cite{blundell03} for a review. Then we have the following decomposition:
\begin{equation}\label{eq: decom control func}
    u_i = q(v_i) + \varepsilon_i, \quad \text{with}\quad \mathbb{E}(\varepsilon_i|Z_\ic,X_{i\cdot},v_i)=\mathbb{E}(\varepsilon_i|v_i)=0.
\end{equation}
For simplicity, we assume that $\mathbb{E}(g(D_i))$, $\mathbb{E}(q(v_i))$, and $\mathbb{E}(\psi_\ell(Z_{i\ell}))$ for all $\ell\in[p_z]$ are all zero. In practice, we allow nonzero expectations by adding an intercept term to the model and handle the intercept by demeaning.
\subsection{Estimation}
\label{sec: est}
\par Throughout the paper, we use the B-spline basis functions defined in Section \ref{subsec: def} following \citet{chen2018optimal} to approximate the nonlinear functions in the model specified by (\ref{eq: y model})-(\ref{eq: decom control func}). We use the uniformly based knots \citep{schumaker2007spline} for B-spline functions such that the distances between every two adjacent knots are equal. Specifically, we use $B$, with the $(i,j)$-th element $B_{ij}=B_{j}(D_i)$, to denote an $n\times M_D$ matrix of $g(D_i)$'s basis functions; we use $H$, with the $(i,j)$-th element $H_{ij}=H_{j}(v_i)$, to denote an $n\times M_v$ matrix of $q(v_i)$'s basis functions; we use $K_\ell$, with the $(i,j)$-th element $(K_{\ell})_{ij}=K_{j\ell}(Z_{i\ell})$, to denote an $n\times M_\ell$ matrix of $\psi_\ell(Z_{i\ell})$'s basis functions, and $K=(K_1,K_2,\cdots,K_{p_z})$. We define $B(\cdot) = (B_1(\cdot),B_2(\cdot),\dots,B_{M_D}(\cdot))^\top$ as a functional vector that collects the spline functions of $g(\cdot)$, and similarly define $H(\cdot)$ and $K_\ell(\cdot)$. More spline functions result in smaller approximation errors of the nonlinear functions but larger variances in the estimators. The choice of numbers of spline functions $M_D$, $M_v$, and $M_\ell$ are discussed in the third paragraph of Section \ref{subsec: Setup} for simulation studies. 
\par Let $\mathcal{D}$ denote the support of $D_i$. By Proposition \ref{prop: spline approx power}, when the function $g$ belongs to the $\gamma$-th H\"{o}lder Class   specified in Definition \ref{def: holder}, there exists a $\beta\in\mathbb{R}^{M_D}$ such that the following approximation error 
\begin{equation}
r_{g}(d)=g(d)-B(d)^\top\beta
\end{equation}
satisfies 
\begin{equation}\label{eq: error approx}
    \sup_{d\in\mathcal{D}}|r_{g}(d)| = O(M_D^{-\gamma}),\ {\rm and }\ \sup_{d\in\mathcal{D}}|r^\prime_{g}(d)| = O(M_D^{-\gamma+1}).
\end{equation} 
Therefore, the coefficient $\beta$ guarantees accurate spline approximations of both the original treatment function $g(\cdot)$ and its derivative $g^\prime(\cdot)$. Similarly, we use the bases $H(\cdot)$ to approximate $q(\cdot)$, and $K_\ell(\cdot)$ to approximate $\psi_\ell(\cdot)$ for $\ell=1,2,\dots,p_z$. There exist $\eta\in\mathbb{R}^{M_v}$ and  $\kappa_\ell\in\mathbb{R}^{M_\ell}$ for $\ell=1,2,\dots,p_z$ such that the approximation errors 
\begin{equation}
r_{q}(v)=q(v)-H(v)^\top\eta 
\end{equation}
and 
\begin{equation}
r_{\ell}(z)=\psi_\ell(z)- K_\ell(z)^\top\kappa_\ell\text{ for }\ell=1,2,\dots,p_z
\end{equation} 
satisfy similar error bounds as \eqref{eq: error approx}. 
We write $B_\ic = (B_1(D_i),\cdots.B_{M_D}(D_i))^\top$ as the spline functions for the $i$-th individual; $H_\ic$ and $K_\ic$ are defined similarly. Furthermore, define $\kappa = (\kappa_1^\top, \kappa_2^\top,\cdots,\kappa_{p_z}^\top)^\top$.  With the spline approximation, the original models (\ref{eq: y model})-(\ref{eq: decom control func}) are rewritten as follows:
\begin{align}\label{eq: y model spline}
Y_i&= B_\ic^\top\beta +   H_\ic^\top \eta + X_{i\cdot}^{\top}\theta+r_g(D_i)+r_q(v_i)+\varepsilon_i,\\\label{eq: d model spline}
D_i&= K_\ic^\top \kappa + X_{i\cdot}^{\top}\varphi + r_{\psi i} + v_i, 
\end{align}
where $r_{\psi i} := \sum_{\ell=1}^{p_z}r_{\ell}(Z_{i\ell})$. In order to carry out the regression for the aforementioned models, it is necessary to estimate $H_\ic$ that includes an unobservable random error $v_i$. To address the high dimensionality in (\ref{eq: d model spline}), we implement the following partial $L_1$ penalized regression  
\begin{equation}
\left\{\widehat{\kappa},\widehat{\varphi}\right\}=\arg\min_{\kappa,\varphi}\frac{1}{n}\sumn \left(D_i- K_\ic^\top\kappa -X_{i\cdot}^{\top}\varphi\right)^2+\lambda_D \|\varphi\|_1,
\label{eq: partial penalization for D}
\end{equation}
where the tuning parameter $\lambda_D>0$ is chosen by cross validation. The $L_1$ penalization in \eqref{eq: partial penalization for D} is only applied to $\varphi$ to address the high-dimensionality of $X_{i\cdot}$. We compute the residuals $\hatv_i=D_i- K_\ic^\top\hat\kappa - X_{i\cdot}^{\top}\hatvarphi$ and use it to estimate $v_i$.
\par Recall that $H_j(\cdot)$ for $j=1,2,\dots,M_v$ are basis functions of $q(v_i)$. Define $\widehat{H}_{ij} = H_{j}(\hatv_i)$, $\hat H_\ic = (\widehat{H}_{i1},\widehat{H}_{i2},\dots,\widehat{H}_{iM_v})^\top$, and    $\widehat{H} = (\hat H_{1\cdot},\hat H_{2\cdot},\dots,\hat H_{n\cdot})^\top$.  The model (\ref{eq: y model spline}) can be rewritten as follows:
\begin{equation}\label{eq: y model spline v hat}
Y_i= B_\ic^\top \beta + \hat H_\ic^\top\eta +X_{i\cdot}^{\top}\theta+r_i+ \varepsilon_i,
\end{equation}
where the approximation error $r_i$ admits the following form 
\begin{equation}
\label{eq: def  ri}
  r_i =  r_g(D_i)+r_q(\hatv_i)+q(v_i)-q(\hatv_i).
\end{equation}
The above error $r_i$ includes the spline approximation error $r_g(D_i)+r_q(\hatv_i)$ and the control function estimation error $q(v_i)-q(\hatv_i)$. 

To estimate $\beta$, $\eta$, and $\theta$, we use a similar partially $L_1$ penalized regression for the outcome model (\ref{eq: y model spline v hat}), %based on (\ref{eq: y model spline v hat}). The estimators are given as follows:
\begin{equation}
\left\{\widehat{\beta},\widehat{\eta},\widehat{\theta}\right\}=\arg\min_{\beta,\eta,\theta}\frac{1}{n}\sumn \left(Y_i-B_\ic^\top\beta -  \hat H_\ic ^\top\eta-X_{i\cdot}^{\top}\theta\right)^2+\lambda_Y \|\theta\|_1,
\label{eq: partial penalization}
\end{equation}
where $\lambda_Y>0$ is the tuning parameter chosen by cross-validation. Similar to \eqref{eq: partial penalization for D}, the penalty is imposed on the high-dimensional vector $\theta$ only. 
\par Define $B^\prime(d) = (B_j^\prime(d))_{j\in[M_D]}$ as the derivative of spline functions.  The plug-in estimator for $g^\prime(d)$ using $\widehat{\beta}$ is 
\begin{equation} \widehat{g}^\prime(d)=B^\prime(d)^\top\hat\beta. 
\label{eq: plug-in estimator}
\end{equation}
The above plugin estimator admits the following error,
\begin{equation}
\widehat{g}^\prime(d)-g^\prime(d)= B^\prime(d)^\top \left(\widehat{\beta} -\beta \right) - r_{g}^\prime(d),
\label{eq: error decomposition}
\end{equation}
where $r_{g}^\prime(d) = g^\prime(d) - B^\prime(d)^\top\beta $ is a higher-order term from the spline approximation error of $g^\prime(d)$. Even though the penalization in \eqref{eq: partial penalization} is not directly applied to the parameter $\beta$, the plug-in estimator $\widehat{g}^\prime(d)$ using $\hat\beta$ still inherits a certain level of bias from \eqref{eq: partial penalization}. This is due to the correlation between the treatment variable $D_i$ and the high-dimensional covariates $X_{i\cdot}$. We propose a new debiased estimator for $g^\prime(d)$ in Section \ref{sec: inference g NL}, and construct a uniform confidence band for $g^\prime(d)$ based on the bias-corrected estimator in Section \ref{sec: ucb}. 

\subsection{Double Bias Correction} 
\label{sec: inference g NL}
\par We aim at a debiased estimator of the marginal effect function $g^\prime(d)$, approximated by $B^\prime(d)^\top \beta$. The primary task is thus to find a debiased estimator of $\beta$. % Recall we obtain the initial estimator for $\beta$ with regularization bias from the LASSO problem \eqref{eq: partial penalization}. 
Earlier literature \citep{zhang2014confidence,van2014asymptotically,javanmard2014confidence} developed debiased LASSO estimators for inference of high-dimensional models without endogeneity. In this section, we develop a novel debiasing procedure to overcome the unique challenge in our model, which arises from high dimensionality, nonlinearity, and endogeneity. 
\par To fix ideas, let the $p_F$-dimensional vector $\hat F_\ic$ collect some ``regressors''. The bias-corrected estimator of $\beta$ has the following form: 
\begin{equation}\label{eq: db beta}
    \tilde\beta = \hat\beta + \dfrac{1}{n}\hat\Omega_B\sumn \hat F_\ic\hat\varepsilon_i, 
\end{equation}
where $\hat\varepsilon_i$ is the residual, and $\hat\Omega_B$ is an estimator of a submatrix composed of the first $M_D$ rows of the inverse Gram matrix of $\hat F_\ic$. We will specify $\hat F_\ic$ and $\hat\Omega_B$ in \eqref{eq: def F hat i} and \eqref{eq: projection direction 1}, respectively.
\par In view of (\ref{eq: db beta}), we examine the LASSO residual $\hat\varepsilon_i$ of the model (\ref{eq: y model spline v hat}) and (\ref{eq: def  ri}). Define $\hatW_\ic := (B_\ic,\hatH_\ic)$, $\omega := (\beta^\top,\eta^\top )^\top$, and $\hat\omega = (\hat\beta^\top,\hat\eta^\top )^\top$ as the initial LASSO estimator of $\omega$. The fitted value is defined as $\hat Y_i = \hatW_\ic^\top\hat\omega + X_\ic^\top\hat\theta$, and the residual is decomposed as
\begin{equation}\label{eq: residual begin}
\begin{aligned}
   \hat\varepsilon_i = Y_i - \hat Y_i =  \underbrace{\hatW_\ic^\top(\omega-\hat\omega) + X_\ic^\top(\theta-\hat\theta)}_{\varDelta_i^Y}  + \underbrace{q(v_i) - q(\hat v_i)}_{\varDelta_i^v}  + r_g(D_i) + r_q(\hat v_i) + \varepsilon_i.
\end{aligned}
\end{equation}
We see the necessity of double bias correction from the above error decomposition \eqref{eq: residual begin}. The first source of bias $\varDelta_i^Y$ comes from the regularization in the outcome regression (\ref{eq: partial penalization}). It can be addressed by existing bias correction methods. The second source of bias $\varDelta_i^v$ comes from the control function approximation using an $L_1$ regularized regression (\ref{eq: partial penalization for D}). The additional bias $\varDelta_i^v$ brings a unique challenge to our high-dimensional nonparametric endogenous effect model with the control function. Consequently, we need a solution to simultaneously correct $\varDelta_i^Y$ from regularization (\ref{eq: partial penalization}) and $\varDelta_i^v$ from the regularization (\ref{eq: partial penalization for D}). We call our solution \emph{double bias correction} to highlight the necessity of removing two biases. 
\par Using Taylor expansion, we have the following approximation of $q(v_i)-q(\hat v_i),$ 
\[q(v_i)-q(\hat v_i) \approx q^\prime(\hat v_i)( v_i - \hat v_i) \approx q^\prime(\hat v_i) X_\ic^\top (\hat\varphi - \varphi) + q^\prime(\hat v_i) K_\ic^\top (\hat\kappa - \kappa), \]
where $\hat\varphi$ and $\hat\kappa$ are from (\ref{eq: partial penalization for D}), and the higher-order biases and spline approximation errors are omitted for simplicity of presentation. This approximation is linear in the LASSO estimators $\hat\varphi$ and $\hat\kappa$. Thus, the residual \eqref{eq: residual begin} is approximated by 
\begin{equation}\label{eq: residual continue}
\begin{aligned}
    \hat\varepsilon_i &\approx \hatW_\ic^\top(\omega-\hat\omega) + X_\ic^\top(\theta-\hat\theta) +   q^\prime(\hat v_i) X_\ic^\top (\hat\varphi - \varphi) +  q^\prime(\hat v_i) K_\ic^\top (\hat\kappa - \kappa) + \varepsilon_i.
\end{aligned}
\end{equation}
The linearization \eqref{eq: residual begin} is not operational for bias correction since the function  $q^\prime(\cdot)$ is infeasible. We thus estimate it by 
\begin{equation}
\label{eq: qphat def}
\hatqp(\hat v_i) := \hatH_{i\cdot}^{\prime\top} \hateta 
\end{equation}
where $\hat\eta$ comes from the LASSO regression (\ref{eq: partial penalization}), and $\hatH_\ic^\prime$ is an $M_v\times 1$ vector with the $j$-th coordinate being the derivative of the spline  function $\hat H_{ij}^\prime = H^\prime_j(\hat v_i)$. Now the regressors in (\ref{eq: residual continue}) are approximated by 
\begin{equation}\label{eq: def F hat i}
    \hatF_\ic=(\hat W_{i\cdot}^\top, X_\ic^\top, \hatqp_i(\hat v_i), K_{i\cdot}^\top,\hatqp(\hat v_i) X_{i\cdot}^\top)^\top\in\mathbb{R}^{p_F},
\end{equation}
where $p_F$ is the length of $\hat F_\ic$. Notice that the valid instruments $Z_\ic$ have no direct effects on $Y_i$ in model (\ref{eq: y model}). Nevertheless, the regularization bias $\varDelta_i^v$ from estimating the control function relates to the splines $K_\ic$ of the instruments $Z_\ic$. Thus, though the instruments $Z_\ic$ do not explicitly appear in model (\ref{eq: y model}), their spline functions $K_\ic$ are included in $\hat F_\ic$ and used for bias correction. 
\par We then construct a debiased LASSO estimator based on the residual $\hat\varepsilon_i$ and the regressors $\hat F_\ic$ defined in \eqref{eq: def F hat i}. Define $\hatF :=(\hatF_\ic^\top)_{i\in[n]}^\top$. Then  $\hat\Omega_B^\top = (\hat\Omega_1,\cdots,\hat\Omega_{M_D})$ is a $p_F\times M_D$ matrix whose $j$-th column is defined as follows:
\begin{equation}
\begin{aligned}
\hat\Omega_{j} =\arg\min_\Omega\  &\Omega^{\top} \hat\Sigma_F \Omega\\
\text{subject to \ }
&{ {\|\hat\Sigma_{F} \Omega - \textbf{i}_j\|_{\infty} \leq \mu_j} } \\
&\|n^{-1/2}\hat F \Omega\|_\infty \leq \mu_j \\
\end{aligned}
\label{eq: projection direction 1}
\end{equation}
where $\textbf{i}_j$ is the $j$-th standard basis and $\mu_j$ is a tuning parameter. The constraint $\|\hat\Sigma_{F} \Omega - \textbf{i}_j\|_{\infty} \leq  \mu_j$ is imposed to control the bias term, and $\|n^{-1/2}\hat F \Omega \|_\infty\leq \mu_j$ bounds the higher-order conditional moments and hence guarantees the asymptotic normality for non-Gaussian errors $\varepsilon_i$. By the bias-corrected estimator $\tilde\beta$ defined in (\ref{eq: db beta}), we finally construct a debiased estimator for the marginal effect function:
\begin{equation}\label{eq: local correction 1}
    \tilde g^\prime(d) := B^\prime(d)^\top \tilde\beta = B^\prime(d)^\top \hat\beta + \hat m(d)^\top\dfrac{\sumn \hat F_\ic\hat\varepsilon_i}{n},
\end{equation}
where $\hat m(d)$ = $\hat\Omega_B^\top B^\prime(d)$. The tuning parameter selection is specified in the last two paragraphs of Section \ref{subsec: Setup}. 
\par We use the debiased estimator (\ref{eq: local correction 1}) for inference of the marginal effect  $g^\prime(d)$. Define \begin{equation}\label{eq: s.e. def}
  \hats(d)=\sqrt{\hatm(d)^\top\hat\Sigma_F\hatm(d)}, \quad \text{and} \quad \  \hat\sigma_\varepsilon^2 = n^{-1}\sumn \hat\varepsilon_i^2.
\end{equation}  For any fixed $d$, a point-wise $100(1-\alpha)\%$ confidence interval is given by  
\begin{equation}
\begin{aligned}\label{eq: def confidence interval}   
   \left[\tildegp(d)-z_{1-\alpha/2}\cdot\hats(d)\cdot\sqrt{\dfrac{\hatsigma_\varepsilon^2}{n}},\ \tildegp(d)+ z_{1-\alpha/2}\cdot\hats(d)\cdot\sqrt{\dfrac{\hatsigma_\varepsilon^2}{n}}\right],
\end{aligned}
\end{equation}
where $z_{1-\alpha/2}$ is the $(1-\alpha/2)$-th quantile of the standard normal distribution. To establish the uniform confidence band of the whole function $g^\prime(d)$ for all $d$ over its domain $\mathcal{D}$, we will use a critical value from multiplier bootstrap in place of the $z_{1-\alpha/2}$ used in (\ref{eq: def confidence interval}). 
\subsection{Uniform Confidence Band}
\label{sec: ucb} 
We define the following quantity:
\begin{equation}
    \label{eq: def H n}\mathbb{H}_n(d) := \dfrac{\sqrt{n}(\tildegp(d) - \gp(d))} {\hats(d)\hatsigma_\varepsilon},
\end{equation} 
where $\hats(d)$ and $\hatsigma_\varepsilon$ are defined in (\ref{eq: s.e. def}). Under the regularity conditions stated in Section \ref{sec: inference theory NL}, we show that  $\mathbb{H}_n(d)$ converges in distribution to a standard normal variable for a given treatment level $d$. We proceed to a confidence band for all $d$ over its domain $\mathcal{D}$ by considering the distribution of the empirical process $ \mathbb{H}_n(d) $. Following the techniques of \citet{chernozhukov2014anti,chernozhukov2014gaussian}, we approximate $\mathbb{H}_n(d)$ by the following Gaussian multiplier process:

\begin{equation}\label{eq: H hat}
    \hat{\mathbb{H}}_n(d) = \dfrac{1}{\sqrt{n}}\sum_{i\in[n]} e_i\cdot\dfrac{ \hat F_\ic^\top \hat m(d)}{\hat s(d)},
\end{equation}
where $e_i$ for $1\leq i\leq n$ are i.i.d.\ standard normal variables. Let $\hat c_n(\alpha)$ be the $(1-\alpha)$-th quantile of $ \sup_{d\in\mathcal{D}} \left|\hat{\mathbb{H}}_n(d)\right| $. Then, we construct the following confidence band at level $100(1-\alpha)\%$, 
\begin{equation}
\begin{aligned}\label{eq: def confidence band}    \mathcal{C}_{n,\alpha}^{\mathcal{D}} &:= \{\mathcal{C}_{n,\alpha}(d):d\in\mathcal{D}\}, \text{ with }\\ 
    \mathcal{C}_{n,\alpha}(d) &:= \left[\tildegp(d)-\hat c_n(\alpha)\cdot\hats(d)\cdot\sqrt{\dfrac{\hatsigma_\varepsilon^2}{n}},\tildegp(d)+\hat c_n(\alpha)\cdot\hats(d)\cdot\sqrt{\dfrac{\hatsigma_\varepsilon^2}{n}}\right]. 
\end{aligned}
\end{equation}
The critical value $\hat c_n(\alpha)$ approximates the $(1-\alpha)$-th quantile of $\sup_{d\in\mathcal{D}}\left| \mathbb{H}_n(d)\right|$, which is the supreme absolute value of an asymptotically Gaussian empirical process. The theoretical justifications of Gaussian approximation by \citet{chernozhukov2014anti,chernozhukov2014gaussian} suggest that $\hat c_n(\alpha)$ is a good approximation such that the confidence band (\ref{eq: def confidence band}) is asymptotically honest.
\par All methods described thus far have utilized the full sample. Unlike the regression models without endogeneity, our regressors $\hat F_\ic$ used for bias correction include some initial LASSO estimators that substantially complicate the theoretical analysis for our double bias correction. These estimators include $\hat\varphi$ and  $\hat\kappa$ that construct the residual $\hat v_i$ of the LASSO regression (\ref{eq: partial penalization for D}), and $\hat\eta$ used to estimate the derivative  $q^\prime$ as in (\ref{eq: qphat def}). Thus, sample-splitting (see details in the following Algorithm \ref{algorithm1}) is required for these initial LASSO estimators in theoretical derivations. We highlight that the sample-splitting is merely a technical restriction. In Section \ref{sec: sim} for numerical studies, we demonstrate that the confidence band using a full sample has a reasonable coverage rate and is more informative than the split-sample confidence band due to a larger effective sample size, and thus we recommend using the full-sample inference for practitioners. 

\par To keep consistent with the theoretical justifications, we summarize the split-sample procedure for a uniform confidence band of $g^\prime(d)$ in Algorithm \ref{algorithm1} below\footnote{The code for implementation is available at \url{https://github.com/ZiweiMEI/HDNPIV}.}. The algorithm for the full-sample inference is summarized in Appendix \ref{subsec: algo full sample}. We remark that our procedure also works for the original function $g(d)$, with $B^\prime(d)$ in (\ref{eq: local correction 1}) replaced by the original splines $B(d)$ and (\ref{eq: H hat})-(\ref{eq: def confidence band}) revised accordingly. \\
\begin{algorithm}[H]
\footnotesize
\caption{\label{algorithm1}Split-Sample Uniform Confidence Band for $g^\prime(d)$}
\hspace*{0.01in}
\begin{algorithmic}[1]
\State Data preparations: Prepare two independent datasets $\mathscr{D}_n, \mathscr{D}^\prime_{n^\prime}$.
\State Initial estimators: Obtain the full sample LASSO estimators $\hat\kappa,\hat\varphi$ by (\ref{eq: partial penalization for D}) and $\hat\beta,\hat\eta,\hat\theta$ by (\ref{eq: partial penalization}) with $\mathscr{D}_n\cup \mathscr{D}^\prime_{n^\prime}$. Furthermore,  compute the LASSO estimators $\hat\kappa^\ind$, $\hat\varphi^\ind$, and $\hat\eta^\ind$ using $\mathscr{D}^\prime_{n^\prime}$ only. 
\State Control function approximation: For $i\in \mathscr{D}_n$, compute the residuals $\hat v_i = D_i - K_\ic^\top \hat\kappa^\ind - X_\ic^\top\hat\varphi^\ind$ and estimate $q^\prime(v_i)$ by $\hatqp(\hat v_i) = H^\prime(\hat v_i)^\top\hat\eta^{\ind}$.
\State Double bias correction: For $i\in \mathscr{D}_n$, construct $\hat F_\ic$ as in (\ref{eq: def F hat i}). Solve (\ref{eq: projection direction 1}) and save the solution $\hat\Omega_B$.
\State Multiplier bootstrap: Compute the $(1-\alpha)$-th quantile of $\sup_{d\in\mathcal{D}}\left| \hat{\mathbb{H}}_n(d)\right|$ with $\Hat{\mathbb{H}}_n(d)$ defined as (\ref{eq: H hat}). Save the quantile as $\hat c_n(\alpha)$.
\State Uniform confidence band: Construct the confidence band of $g^\prime(d)$ as in (\ref{eq: def confidence band}).
\end{algorithmic}
\end{algorithm}
\section{Asymptotic Theory} \label{sec: inference theory NL}
Throughout the paper, we consider the number of covariates $p$, the sparsity index $s$, and the number of bases functions $M_D$, $M_v$, and $M_\ell$ for $\ell=1,2,\dots,p_z$ as deterministic functions of the sample size $n$. In asymptotic statements, we only explicitly let $n\to\infty$, thus allowing $p$, $M_D$, $M_v$, and $M_\ell$ all go to infinity, whereas $s$ is either fixed or divergent.  
\par We impose the following theoretical assumptions. 
\begin{Assumption}[Control Function] \label{assu: control func}Assume that the dataset $\mathscr{D}_n := \{Y_i,D_i,X_\ic,Z_\ic\}_{i\in[n]}$ is independently and identically distributed. Furthermore, suppose that (\ref{eq: control function}) holds, and $\mathbb{E}(v_i|Z_{i\cdot},X_{i\cdot}) = 0$. 
We assume that $\sigma_v^2 = \mathbb{E}(v_i^2|Z_\ic,X_{i\cdot})$ and $\sigma_\varepsilon^2 = \mathbb{E}(\varepsilon_i^2|v_i,Z_\ic,X_{i\cdot})$ are positive constants where $\varepsilon_i$ is defined in (\ref{eq: decom control func}). Furthermore, $\E(\varepsilon_i^4|v_i,Z_\ic,X_{i\cdot})$ is bounded by some positive constant independent of $n$ and $p$.
\end{Assumption}
% \Ziwei{I remove the sub-Gaussianity of $\varepsilon_i$.}
\par As mentioned earlier in Section \ref{sec21}, we impose the condition (\ref{eq: control function}) for identification using the control function approach in  Assumption \ref{assu: control func}. We also impose the conditional homoskedasticity of $v_i$ and $\varepsilon_i$. Even though the consistency of the initial estimators does not require such homoskedasticity assumption, it is unclear how to relax this assumption for establishing asymptotic normality of our double-bias-correction estimators unless we impose additional sparsity conditions on precision matrix of $\hat F_\ic$ (the inverse of $\Sigma_{F|\mathcal{L}}$ defined before Proposition \ref{prop: feasible}). Since $\hat F_\ic$ includes LASSO estimators and spline functions, the sparsity of this precision matrix is difficult to deduce from low-level assumptions. For the sake of simplicity and to maintain focus on our methodological contributions, we will concentrate on homoskedasticity, thereby avoiding any diversions. We further assume the bounded third order moment of the random shock $\varepsilon_i$ to control for the error bounds.

\begin{Assumption}[Bounded Supports and Densities] \label{assu: D v}
Suppose that $D_i\in\mathcal{D}$, $v_i\in\mathcal{V}$, $X_{ij}\in\mathcal{X}$ $\forall$ $j\in[p]$, and $Z_{i\ell}\in\mathcal{Z}$ for all $\ell\in[p_z]$ are continuous, where $\mathcal{D}=[a_D,b_D]$, $\mathcal{V}=[a_v,b_v]$, $\mathcal{X}=[a_x,b_x]$, and $\mathcal{Z}=[a_z,b_z]$ are compact intervals in $\mathbb{R}$. Additionally, the marginal density functions of $D_i$ and $v_i$, as well as the joint density functions of $(Z_{i1},\cdots,Z_{ip_z})$, are absolutely continuous with density functions bounded below by $c_f>0$ and above by $C_f>0$.
\end{Assumption}

We assume boundedness of $v_i$ and $D_i$ to guarantee small errors of B-spline approximation. This technical assumption has been used in the control function literature  \citep{newey1999, imbens2009identification,su2008local,ozabaci2014additive}. In practice, we can normalize the data to a compact interval. As pointed out by \citet{ozabaci2014additive}, boundedness is possible to remove with arguments in \citet{su2012sieve}, inducing additional complications. Finally, we focus on the continuously valued treatment, instruments, and covariates.

We further impose regular conditions on the nonlinear functions $g(\cdot)$, $q(\cdot)$ and $\psi_\ell(\cdot)$ in models \eqref{eq: y model} and \eqref{eq: D model}. We define the functions of the H\"{o}lder Class as follows.
\begin{Definition}\label{def: holder}
The $\gamma$-th H\"{o}lder Class $\mathcal{H}_{\mathcal{X}}(\gamma,L)$ is the set of $\gamma$-times differentiable functions $f:\mathcal{X}\to\mathbb{R}$ such that its derivative $f^{(\ell)}$ with $\ell=\lfloor\gamma\rfloor$ satisfies the following:
\begin{equation}
    |f^{(\ell)}(x)-f^{(\ell)}(y)|\leq L|x-y|^{\gamma-\ell},\text{ for any } x,y\in\mathcal{X}.
\end{equation}
\end{Definition}
When $\mathcal{X}$ is compact in the real line, $f\in\mathcal{H}_{\mathcal{X}}(\gamma,L)$ implies that all $j$-th derivatives of $f$ are bounded in $\mathcal{X}$ for any $j\in\{1,2,\cdots,\ell=\lfloor\gamma\rfloor\}$. To address the error between $\hat v_i$ and $v_i$, we assume that the domain of $q(\cdot)$ is $\mathcal{V}_q=[a_v-\epsilon_v,b_v+\epsilon_v]$ for some constant $\epsilon_v > 0$, which extends $\mathcal{V}=[a_v,b_v]$, the support of $v_i$, , as specified in Assumption \ref{assu: D v}.
\begin{Assumption}[H\"{o}lder Class] Suppose that \label{assu: function}
$g(\cdot)\in\mathcal{H}_{\mathcal{D}}(\gamma,L)$, $q(\cdot)\in\mathcal{H}_{\mathcal{V}_q}(\gamma,L)$ and $\psi_\ell(\cdot)\in\mathcal{H}_{\mathcal{Z}}(\gamma,L)$ for some $\gamma\geq 2$ and positive constant $L$. Further assume $\sum_{\ell=1}^{p_z}\psi_\ell^\prime(Z_{i\ell}) \neq 0$ with probability one. 
\end{Assumption}
 Assumption \ref{assu: D v} about the compact supports and bounded densities for $D_i$ and $v_i$, together with Assumption \ref{assu: function}, guarantees a small spline function approximation error. The condition of nonzero derivative in Assumption \ref{assu: function} relates to identification, as explained in the following Remark \ref{rem: id}.

\begin{Remark}[Identifiability]\label{rem: id}
    \par According to \citet[Theorem 2.3]{newey1999}, the models (\ref{eq: y model})-(\ref{eq: decom control func}) are identifiable if (a) the boundary of the support of $(Z_\ic^\top, X_\ic^\top,v_i)$ shares zero probability; (b) all functions are differentiable; and (c) $\sum_{\ell=1}^{p_z}\psi_\ell^\prime(Z_{i\ell})$ is almost surely nonzero. These conditions are reflected in our assumptions. Specifically, condition (a) is implied by the bounded support and density condition by Assumption \ref{assu: D v}. Condition (b) is implied by Assumption \ref{assu: function} that all functions are in the H\"{o}lder class. The last condition (c), imposed by Assumption \ref{assu: function}, means the IVs are relevant to the endogenous variable $D$.
\end{Remark}

In the following, for any random vector $\zeta = (\zeta_{j})_{j\geq 1}$, we use $\tilde \zeta = (\tilde\zeta_{j})_{j\geq 1}$ to denote the standardized vector with $\tilde\zeta_{ij} :=(\E(\zeta_{ij}^2))^{-1/2}\zeta_{ij}$. 
\begin{Assumption}[Eigenvalues] \label{assu: eigen} Suppose that $\E(X_{ij}^2)$ is bounded away from zero and above uniformly for all $j$. Furthermore, define $Q_{i\cdot}:=( K_{i\cdot}^\top, X_{i\cdot}^\top)^\top$ and $U_{i\cdot}=(B_{i\cdot}^\top,H_{i\cdot}^\top,X_{i\cdot}^\top)^\top$. We use $\tilde Q_\ic$ and $\tilde U_\ic$ to denote their standardized versions. Suppose that the eigenvalues of $\E(\tilde Q_\ic\tilde Q_\ic^\top)$ and $\E(\tilde U_\ic\tilde U_\ic^\top)$ are bounded away from zero and above.
\end{Assumption}

\begin{Remark}[No Perfect Collinearity]
In Assumption \ref{assu: eigen}, the vectors $Q_\ic$ and $U_\ic$ respectively collect the regressors, or their population version, in the LASSO regressions (\ref{eq: partial penalization for D}) and (\ref{eq: partial penalization}). Intuitively, Assumption \ref{assu: eigen} rules out perfect collinearity of the regressors in LASSO. We need this assumption for the restrictive eigenvalue conditions to derive the LASSO consistency for high-dimensional models \citep{bickel2009simultaneous}. We impose the assumption on the standardized regressors, since the B-splines $B_\ic$, $H_\ic$, and $K_\ic$ have Gram matrices with eigenvalues convergent to zero as the number of spline bases pass to infinity, and thus are of smaller scales than the covariates $X_\ic$. 
\end{Remark}

\begin{Remark}[IV Strength]
The bounded-away-from-zero eigenvalues of $\E(\tilde U_\ic\tilde U_\ic^\top)$ in Assumption \ref{assu: eigen} can be considered as a restriction on the IV strength. This condition means that the variables in $U_\ic$, including (functions of) $D_i$, $v_i$ and $X_\ic$, are not perfectly correlated. Intuitively, if the nonlinear functions of IVs $\{\psi_\ell(Z_{i\ell})\}_{\ell=1}^{p_z}$ in the model (\ref{eq: D model}) introduce sufficient variations to $D_i$, the variables $D_i$, $v_i$, and $X_\ic$ in (\ref{eq: D model}) are far away from perfect collinearity, and the eigenvalue condition for $\E(\tilde U_\ic\tilde U_\ic^\top)$ can thus be satisfied. The sufficient variations from $\{\psi_\ell(Z_{i\ell})\}_{\ell=1}^{p_z}$ requires that neither the covariates $X_\ic$ and the functions of IVs $\psi_\ell(Z_{i\ell})$ for  $\ell=1,\dots,p_z$ are perfectly collinear, nor $\sum_{\ell=1}^{p_z}\psi_\ell(Z_{i\ell})$ is  close to a constant function. The latter means the treatment $D_i$ and the IVs $Z_\ic$ are significantly relevant, which holds in many empirical applications. 
\end{Remark}

We further state the restrictions on the asymptotic regime. 
% \Ziwei{These conditions were in the Theorems. I add this assumption to summarize the conditions about convergence rates of parameters so that the Theorems look cleaner.}

\begin{Assumption}\label{assu: asymp}
    Suppose the following conditions hold:
    \begin{enumerate}[(a)]
        \item The sparsity index $s = O(n^{1/4-c_0})$ for an arbitrarily small $ c_0 \in (0, \frac{1}{4})$ and $\log p = O(n^{c_\gamma})$ with the constant $c_\gamma$ only depending on  $\gamma$ in Assumption \ref{assu: function}. 
        \item   The number of spline bases satisfy $M_D\asymp M_v \asymp M_\ell$ for $\ell=1,2,\dots,p_z$. Define $M := \max\left(M_D, M_v, M_1, \dots, M_{p_z} \right)$. Assume $M \asymp n^\nu$ with $\nu\in[\frac{1}{2\gamma+1},\frac{1}{4})$. The LASSO tuning parameters satisfy  $\lambda_Y = C_Y\sqrt{\frac{\log (pM)}{n}}$ and $\lambda_D = C_D\sqrt{\frac{\log (pM)}{n}}$ for some sufficiently large constants $C_Y,C_D>0$. 
    \end{enumerate} 
\end{Assumption}
The rates specified in Assumption \ref{assu: asymp} of the tuning parameters $\lambda_D$, $\lambda_Y$ are similar to the linear model case, which is applied to the remainder of the paper without further clarifications. These rates are only for technical proofs; in practice, we use cross-validation for data-driven choice of LASSO tuning parameters.
\par The following theorem establishes the rate of convergence for the initial estimators proposed in \eqref{eq: partial penalization}. Recall that $\hat\beta,\hat\eta$ are the estimated coefficients of spline functions, and  $\hat\theta$ stores the estimated coefficients of covariates. 
\begin{Theorem}Suppose that Assumptions \ref{assu: control func}-\ref{assu: asymp} hold. Then the initial estimators $\hat\beta,\hat\eta,\hat\theta$ in (\ref{eq: partial penalization}) and $\hat g^\prime(\cdot)$ in (\ref{eq: plug-in estimator}) have the following error bounds w.p.a.1: 
\begin{equation}
\|\widehat{\beta}-\beta\|_2^2 + \|\hat\eta - \eta\|_2^2 \lesssim M^{-2\gamma + 1} + \dfrac{(s+M)M\log (pM)}{n},
\label{eq: rate omega}
\end{equation}
\begin{equation}
\|\hatgp - \gp\|_2^2 \lesssim  M^{-2(\gamma-1) } + \dfrac{M^2(s+M)\log (pM)}{n},
\label{eq: rate derivative}
\end{equation}
\begin{equation}
\|\widehat{\theta}-{\theta}\|_1 \lesssim  M^{-2\gamma}\sqrt{\dfrac{n}{\log p}} + (s+M)\sqrt{\dfrac{\logpM}{n}}. 
\label{eq: rate theta}
\end{equation}
\label{thm: estimation bound}
\end{Theorem}
\vspace{-1em}
Though we consider estimation of the nonlinear endogenous effect with high-dimensional covariates, the error bounds in Theorem \ref{thm: estimation bound} echo the literature of spline regression without endogeneity. For instance, refer to \citet[Remark 1]{zhou2000derivative} and \citet[Theorem 1]{huang2010variable}. 
\par In terms of inference, we need additional theoretical assumptions. We first formalize the independence assumption for split-sample estimators mentioned in the discussions before Algorithm \ref{algorithm1}. Define the LASSO estimators $\hat\varphi^{\ind}$, $\hat\kappa^{\ind}$, and $\hat\eta^{\ind}$ as those in Algorithm \ref{algorithm1} using a different dataset independent of $\mathscr{D}_n$ specified in Assumption \ref{assu: control func}, where the superscript ``ind'' means independence. 
\begin{Assumption}[Split-Sample Estimators]\label{assu: indep}We assume that the LASSO estimators $\hat\varphi^{\ind}$, $\hat\kappa^{\ind}$, and $\hat\eta^{\ind}$  are constructed by another i.i.d.\ dataset independent of data $\mathscr{D}_n$ with a sample size $n^\prime \asymp n$.
\end{Assumption}
We define $F_{i\cdot} := (B_{i\cdot}^\top,H_{i\cdot}^\top,{X}_{i\cdot}^\top,\qp(v_i) K_{i\cdot}^\top, \qp(v_i){X}_{i\cdot}^\top)^\top$ as the population truth of $\hat F_{i\cdot}$ in (\ref{eq: def F hat i}). We define $F_{ij}^{\rm std} := (\E(F_{ij}^2))^{-1/2}F_{ij}$ as the standardized version of $F_{ij}$ for any $j=1,2,\dots,p_F$, and  $F_\ic^{\rm std} = (F_{i1}^{\rm std},F_{i2}^{\rm std},\dots,F_{ip_F}^{\rm std})^\top$. 
\begin{Assumption}\label{assu: more compatibility}Suppose that $v_i$ is independent of the vector $(X_{i\cdot}^\top,Z_{i\cdot}^\top)^\top$ and $X_{i\cdot}^\top$ has a bounded sub-Gaussian norm. Furthermore, assume that the eigenvalues of $\E( F_\ic^{\rm std}  (F_\ic^{\rm std})^\top)$ are bounded away from zero and above.
\end{Assumption}
% \begin{Remark}\label{rem: more compatibility}
 The sub-Gaussian norm of a random vector is defined in Definition \ref{eq:def-subG} in Appendix \ref{sec: proof}. We need the bounded sub-Gaussian norm of the whole vector $X_\ic$ (on top of the compact support in Assumption \ref{assu: D v}) to rule out strong dependence among the high dimensional covariates. Similar to Assumption \ref{assu: eigen}, Assumption \ref{assu: more compatibility} rules out perfect collinearity among the variables in $F_\ic$. This implies that $q$ is a nonzero and nonlinear function; otherwise, $q^\prime(v_i)$ is constant, and therefore, $X_\ic$ and $q^\prime(v_i) X_\ic$ are perfectly collinear. This assumption ensures the invertibility of $\Sigma_{F|\mathcal{L}}$, and thus, problem \eqref{eq: projection direction 1} is feasible with high probability. For a linear $q$, a feasible solution is also available. In Appendix \ref{sec: linear q}, we discuss the feasibility of  (\ref{eq: projection direction 1}) when $q$ is linear.
% \end{Remark}
\par We first show the feasibility of \eqref{eq: projection direction 1}. Let $\mathcal{L}$ denote the $\sigma$-field generated by the split-sample LASSO estimators  specified in Assumption \ref{assu: indep}. We define $\Sigma_{F|\mathcal{L}} := \mathbb{E}_{\mathcal{L}}\left(\hat F_{i\cdot} \hat F_{i\cdot}^\top\right)$. 
\begin{Proposition}\label{prop: feasible}Suppose that Assumptions \ref{assu: control func}-\ref{assu: more compatibility} hold. Then w.p.a.1, 
\begin{align}
    \left\|\hat\Sigma_F \Sigma_{F|\mathcal{L}}^{-1} - I \right\|_{\infty} &\lesssim M\sqrt{\dfrac{\log (pM)}{n}}, \label{eq: feasibility 1}\\
    \|\hat F^\top \Sigma_{F|\mathcal{L}}^{-1}\|_\infty &\lesssim M\sqrt{\dfrac{\log (pM)}{n}}. \label{eq: feasibility 2}
\end{align}
\end{Proposition}
Proposition \ref{prop: feasible} shows that the problem \eqref{eq: projection direction 1} is feasible w.p.a.1 when \begin{equation}\label{eq: mu j}
    \mu_j = C_j  M\sqrt{\dfrac{\log (pM)}{n}}
\end{equation}  
for $j\in[M_D]$ with some constant $C_j$ large enough.  Specifically, the first $M_D$ columns of $\Sigma_{F|\mathcal{L}}^{-1}$ fulfill the constraints in \eqref{eq: projection direction 1} w.p.a.1. Note that the feasible solution is the inverse of the conditional Gram matrix $\Sigma_{F|\mathcal{L}}$ instead of the unconditional one $\E(\hat F_\ic \hat F_\ic^\top)$, since $\{\hat F_\ic\}_{i\in[n]}$ are i.i.d.\ conditionally on $\mathcal{L}$, but not unconditionally. The convergence rate of a tuning parameter like (\ref{eq: mu j}) is a commonly used technical assumption for proofs \citep{javanmard2014confidence,gold2020inference,lu2020kernel}. For practitioners, we introduce a data-driven choice of $\mu_j$ in the last paragraph of Section \ref{subsec: Setup} for simulations.

For asymptotic normality of the debiased estimator, we need stronger restrictions on the sparsity index $s$ and the number of spline bases $M$.
\begin{manualassumption}{\ref{assu: asymp}$^\prime$} \label{assu: asym more}
       Suppose the following conditions hold:
    \begin{enumerate}[(a)]
        \item The sparsity index $s = O(n^{2/9-c_0})$ for an arbitrarily small $ c_0 \in (0, \frac{2}{9})$ and $\log p = O(n^{c_\gamma})$ with $c_\gamma$ dependent only on $\gamma$ in Assumption \ref{assu: function}. 
        \item  The number of spline bases satisfy $M_D\asymp M_v \asymp M_\ell$ for $\ell=1,2,\dots,p_z$. Recall $M = \max\left(M_D, M_v, M_1, \dots, M_{p_z} \right)$  Assume $M \asymp n^\nu$ with $\nu\in(\frac{1}{2\gamma},\frac{1}{4.5})$. The LASSO tuning parameters satisfies  $\lambda_Y = C_Y\sqrt{\frac{\log (pM)}{n}}$ and $\lambda_D = C_D\sqrt{\frac{\log (pM)}{n}}$ for some $C_Y,C_D>0$ large enough. 
        \item The tuning parameter $\mu_j$ in \eqref{eq: projection direction 1} follows \eqref{eq: mu j}.
    \end{enumerate} 
\end{manualassumption}
 Assumption \ref{assu: asym more} for asymptotic normality of the debiased estimator distinguishes from Assumption \ref{assu: asymp} for consistency of initial estimators in the following aspects. First, in Assumption \ref{assu: asym more}(a) we require $s = o(n^{2/9})$ that is slightly stronger than the rate $o(n^{1/4})$ in Assumption \ref{assu: asymp}(a). Second, a stronger condition $\nu\in (\frac{1}{2\gamma},\frac{1}{4.5})$ is needed in Assumption \ref{assu: asym more}(b) to ensure that the error caused by \eqref{eq: def  ri} is small enough to guarantee asymptotic normality. Specifically, the lower bound $\frac{1}{2\gamma}$ ensures that the spline approximation errors are small enough, and the upper bound $\frac{1}{4.5}$ controls for the variance of $\hat q^\prime(\hat v_i)$ in \eqref{eq: qphat def} to bound the estimation error of $q^\prime(v_i)$. The range of $\nu$ implies that $\gamma > 2.25$, which imposes slight additional smoothness on the second-order derivatives compared to Assumption \ref{assu: function}. Third, in Assumption \ref{assu: asym more}, we specify the rate of $\mu_j$ in \eqref{eq: projection direction 1} for bias correction. 
\par We are now ready to present the main theoretical results for the double bias correction and uniform confidence band. Recall that $\hat m(d)$ is defined below (\ref{eq: local correction 1}) and $\hat s(d)$ is defined below (\ref{eq: def H n}). 
\begin{Proposition}\label{prop: debiased estimator}
Suppose that Assumptions \ref{assu: control func}-\ref{assu: eigen}, \ref{assu: asym more}, \ref{assu: indep}, and \ref{assu: more compatibility} hold. Then for any fixed $d\in\mathcal{D}$, the debiased estimator $\tilde g^\prime(d)$ in (\ref{eq: local correction 1}) satisfies
\begin{equation}
        \sqrt{n}(\tilde g^\prime(d) - g^\prime(d)) = \mathcal{Z}(d)  +  \Delta(d), 
    \end{equation} 
where $ \Delta(d) / \hat s(d) \convp 0$ and $\mathcal{Z}(d) := \dfrac{\hat m(d)^\top \sumn \hat F_\ic \varepsilon_i}{\sqrt{n}}$ with $\E(\mathcal{Z}(d)|\hat F_\ic) = 0$ and ${\rm var}[\mathcal{Z}(d)|\hat F_\ic] = \left(\sigma_\varepsilon\hat s(d)\right)^2$. In addition, $\hat s(d) \asymp M^{1.5}$ uniformly for all $d\in\mathcal{D}$ w.p.a.1.  
\end{Proposition}
Proposition \ref{prop: debiased estimator} presents a result on the decomposition of the estimation error for the debiased estimator $\tilde g^\prime(d)$. The asymptotic distribution of $\sqrt{n}(\hat g^\prime(d) - g^\prime(d))$ is determined by $\mathcal{Z}(d)$ with a zero conditional mean and a conditional variance $(\sigma_\varepsilon\hat s(d))^2$. It can be shown that $\mathcal{Z}(d)/\sigma_\varepsilon\hat s(d)$ converges in distribution to a standard normal variable for any fixed $d$. We relegate this intermediate result and highlight the honesty of the uniform confidence band for the whole function displayed in Theorem \ref{thm: limiting distribution} below. 
% Compared to Theorem \ref{thm: estimation bound} on estimation consistency, a stronger condition $\nu\in (\frac{1}{2\gamma},\frac{1}{4.5})$ is needed here to ensure that $\tilde r_i$ in (\ref{eq: resid final}) is small enough to control the bias $\varDelta_\beta$ in (\ref{eq: decom tilde beta error}). Specifically, the lower bound $\frac{1}{2\gamma}$ ensures that the spline approximation errors are small enough, and the upper bound $\frac{1}{4.5}$ controls for the variance of $\hat q^\prime(\hat v_i)$ in (\ref{eq: qphat def}) to bound the estimation error of $q^\prime(v_i)$. The range of $\nu$ implies that $\gamma > 2.25$, which imposes slight additional smoothness on the second-order derivatives compared to Assumption \ref{assu: function}. The condition $\gamma \geq 2$ in Assumption \ref{assu: function} can be resumed by some unverifiable restriction on the conditional precision matrix $\Sigma_{F|\mathcal{L}}^{-1}$ as a feasible solution under the constraints in (\ref{eq: projection direction 1}) shown by Proposition \ref{prop: feasible}. More discussions are available in Remark \ref{rem: L1 precision} below.
\par The last result in Proposition \ref{prop: debiased estimator} implies that the standard error for the doubly bias corrected derivative function estimator is $M^{1.5}/\sqrt{n}$ w.p.a.1. Though we study the nonparametric inference of a derivative function under a high-dimensional model with endogeneity, the standard error of our debiased estimator shares the same order with the nonparametric estimator using spline regressions under the low-dimensional model without endogeneity \citep[Lemma 5.4]{zhou2000derivative}.
% \begin{Remark}\label{rem: L1 precision}We define $\tilde \Sigma_{F|\mathcal{L}} := S_F^{-1}\Sigma_{F|\mathcal{L}} $ as the scaled conditional Gram matrix, where $S_F:={\rm diag}\left\{\E(F^2_{i1}),\cdots,\E(F^2_{ip_F})\right\}$ with $F_\ic$ defined above Assumption \ref{assu: more compatibility}. The condition $\nu \in (\frac{1}{2\gamma},\frac{1}{4.5})$ in Theorem \ref{thm: limiting distribution} can be replaced by $\nu \in (\frac{1}{2\gamma+1},\frac{1}{4.5})$, $\|\tilde \Sigma_{F|\mathcal{L}}^{-1}\|_1 \lep \ell_n$ with $\ell_n\sqrt{nM^{-(2\gamma+1)}} = o(1)$ and $\ell_n \cdot M^2\sqrt{\frac{\logpM}{n}} = O(1)$. The proof is provided in Appendix \ref{sec: proof}. The $L_1$ bound of scaled inverse conditional Gram matrix $\|\tilde\Sigma_{F|\mathcal{L}}^{-1}\|_1$ allows us to control the bias $\varDelta_\beta$ in (\ref{eq: decom tilde beta error}) under a weaker condition on $\nu$; hence, $\gamma = 2$ is allowed in this new condition. This $L_1$ bound might be deduced by some restrictions on the joint density of the entries in $\hat F_\ic$ (e.g., \citet[Assumption A6]{lu2020kernel}). In our scenario, these restrictions are difficult to derive from low-level assumptions since the estimators $\hat v_i$ and $\hat q^\prime(\hat v_i)$ are included in $\hat F_\ic$, whose dependence on the high-dimensional covariates $X_\ic$ is unknown. Thus, in the main theorem, we impose a slightly stronger restriction on $\nu$ to avoid the unverifiable restriction on $\|\tilde\Sigma_{F|\mathcal{L}}^{-1}\|_1$ due to the difficulty of our problem.
% \end{Remark}
\par Next we present the asymptotic honesty of the confidence band \eqref{eq: def confidence band}. For simplicity, we only consider the honesty of our confidence band over the H\"{o}lder class of functions $g$ in the following theorem. We treat other nonrandom components in  (\ref{eq: y model}) and (\ref{eq: D model}) as fixed, including the nonlinear functions $ \psi_\ell(\cdot)$ for $\ell = 1,2,\dots,p_z$, the coefficients of covariates $\theta$ and $\varphi$, and the distributions of covariates $X_\ic$, instruments $Z_\ic$, and unmeasured errors $(u_i,v_i)$.
\begin{Theorem}
Under the conditions in Proposition \ref{prop: debiased estimator},
the $100(1-\alpha)\%$ confidence band $\mathcal{C}_{n,\alpha}^{D}=\{\mathcal{C}_{n,\alpha}(d):d\in\mathcal{D}\}$ defined as \eqref{eq: def confidence band} is asymptotically honest in the sense that
\begin{equation}
\mathop{\lim\inf}\limits_{n\to\infty}\inf_{g\in\mathcal{H}_{\mathcal{D}}(\gamma,L)}\Pr\left(g^{\prime}(d)\in\mathcal{C}_{n,\alpha}(d)\text{ for all }d\in\mathcal{D}\right) \geq 1 - \alpha,
\label{eq: honest confidence band}
\end{equation} 
where the 
H\"{o}lder functional class $\mathcal{H}_{\mathcal{D}}(\gamma,L)$ is defined in Definition \ref{def: holder}. 
\label{thm: limiting distribution}
\end{Theorem}
%\vspace{-2em}
Theorem \ref{thm: limiting distribution} formally establishes the asymptotic honesty of our confidence band (\ref{eq: def confidence band}) over the 
H\"{o}lder class of functions $g$. The result is shown by the techniques of Gaussian coupling for the suprema of empirical processes by \citet{chernozhukov2014anti,chernozhukov2014gaussian}. 
\section{Simulations}\label{sec: sim}

\subsection{Setup}\label{subsec: Setup}
In the simulation studies, we generate the outcome $Y_i$ and the treatment $D_i$  following models \eqref{eq: y model} and \eqref{eq: D model}. We generate the covariates $X_{i\cdot}$ and IVs $Z_{i\cdot}$ following \citet{lu2020kernel}. Specifically, let $(U_{ji})_{i\in[n],j\in[p+p_z+1]}$ be i.i.d.\ $U[0,1]$ random variables, and 
\[X_{ij} = \dfrac{U_{i,j}+0.3U_{i,p+p_z+1}}{1.3},\ j\in[p],\]
\[Z_{ij} = \dfrac{U_{i,p + j}+0.3U_{i,p+p_z+1}}{1.3},\ j\in[p_z].\]
The error terms are generated from $v_i\sim\text{ i.i.d.\ }\sqrt{12}U[-0.5,0.5]$, $\varepsilon_i \sim\text{ i.i.d.\ }N(0,1)$ and $u_i = q(v_i) + \varepsilon_i$ with $q(v) = v^2 - 1$. We vary $p\in\{150,400,800\}$ for high-dimensional covariates and the sample size $n\in\{500,1000,2000,3000\}$. For the sample-splitting inference, we divide the two samples into two subsets with cardinalities $n_a = \lfloor n/2 \rfloor$ and $n_b = n - n_a$.
\par We set $\theta = (1_6^\top, 0_{p-6}^\top)^\top$ and $\varphi = (1,-1,1,-1,1,-1,0_{p-6}^\top)^\top$. Additionally, we set $p_z=1$ for a just-identified IV model, which is the leading case in empirical studies. We use $\psi(z)=4(2z-1)^2$ as the nonlinear function measuring the relevance of the IV in \eqref{eq: D model}.
We consider four cases of the $g$ function in \eqref{eq: y model}:
\[g_1(d) = 0,\ g_2(d) = d,\ g_3(d)=0.05(d-3)^2,\ g_4(d)=0.02(d-3)^3. \]
The small coefficients in the nonlinear treatment functions balance the relatively large range of $D_i$ for moderate values of marginal functions. We simulate $10^5$ samples following the DGPs \eqref{eq: y model} and \eqref{eq: D model} and form a compact interval with the lower and upper bounds as the 10th and 90th percentiles of the simulated $D_i$, respectively. We take 1000 grid points in this compact interval and construct the confidence band on these grid points.
\par We use cubic B-spline functions for estimation and set the number of basis functions as $M_D = M_v = M_z = 5$ following \citet{lu2020kernel}. As a robustness check, we provide additional simulations in Table \ref{tab:simul_M} for the finite-sample performance with different numbers of $M_D$, with $M_v$ and $M_z$ fixed to 5. It turns out that a larger $M_D$ brings almost no benefits to the coverage of the confidence bands but produces larger variances due to higher variable dimensions.
\par The LASSO tuning parameters in \eqref{eq: partial penalization} and \eqref{eq: partial penalization for D} are selected based on cross-validation. Following the idea of \citet{gold2020inference}, we set the tuning parameter $\mu_j$ in \eqref{eq: projection direction 1} as
$\mu_j = a_0 \cdot \min_\Omega \left\|\hat\Sigma_F \Omega - \textbf{i}_j\right\|_{\infty}$,
where the factor $a_0$ is chosen to balance the bias control, for which a smaller $a_0$ is desirable, and the confidence band's length, for which a larger $a_0$ is desirable. We set $a_0=1.2$ following \citet{gold2020inference}.

\subsection{Numerical Results}
Table \ref{tab:simul} shows the simulation results for the linear $g$ functions. The left panel ``Full-Sample'' shows the results from full-sample inference where we always use all samples for all estimators and bias correction procedures. The right panel ``Split-Sample'' shows the split-sample results based on the procedures described in Algorithm \ref{algorithm1}.
\par We first focus on the left panel showing the full-sample results. In the ``BiasInit'' and ``BiasDB'' columns, we see that compared to the initial plug-in LASSO estimator defined as $\hat{g}^\prime(\cdot)$ in \eqref{eq: plug-in estimator}, the bias-corrected estimator $\tilde{g}^\prime(\cdot)$ in \eqref{eq: local correction 1} significantly reduces the bias from LASSO penalization. The bias of the bias-corrected estimator decreases as the sample size grows. In terms of inference, when $p=150$, the coverage probability of the full-sample confidence bands exceeds the nominal size of 0.95 under all sample sizes. The average confidence band length continues to decrease as the sample size increases. When $p=400$ or $800$, the coverage deviates from the nominal size when the sample is not large enough. As a nonparametric method, our inferential procedure requires a relatively large sample size to handle high-dimensional covariates under finite sample. Among the settings where the coverage is close to the nominal size, the average confidence length decreases as the sample size increases.
\par We then compare the full-sample and split-sample results. Under the same sample size, the split-sample confidence band, compared to the full-sample one, produces either a much worse coverage, or a far wider confidence band if its coverage is close to the nominal size. These comparative results show that even if sample splitting is required for technical proofs, the full-sample inference dominates the split-sample procedure in practice because of a larger effective sample size. We thus recommend the full-sample inference for practitioners. 
\par Table \ref{tab:simul_NL} shows the results for nonlinear $g$. In general, the performance of our procedure is robust to different settings for the $g$ functions. The simulation results show the validity of our procedure for the uniform confidence band of endogenous marginal effects.
\par As a robustness check, additional simulation results are provided in Appendix \ref{app: additional simul}. These additional results show that our procedure is robust to unbounded distributions, although the theory depends on the assumption of boundedness as most high-dimensional nonparametric methods do. In addition, as mentioned in Section \ref{subsec: Setup}, we compare various settings of $M_D$ and show that $M_D=5$ is a reasonable choice since an increase in the number of basis functions does not improve coverage but brings additional variance.
% \usepackage{multirow}
\begin{table}[H]
\begin{center} 
\caption{Simulation Results for Linear $g$ Functions.}
\label{tab:simul}
\begin{tabular}{cc|cccc|cccc}
\hline\hline 
\multirow{2}{*}{$p_x$}  & \multirow{2}{*}{$n$} & \multicolumn{4}{c}{Full-Sample}                      & \multicolumn{4}{c}{Split-Sample}                 \\
\cline{3-10}
                     &                    & BiasInit & BiasDB & Coverage & Length & BiasInit & BiasDB & Coverage & Length \\
                    \hline 
                     \multicolumn{10}{c}{$g(d) = 0$, $g^\prime(d)=0$.}  \\    
                     \hline 
\multirow{4}{*}{150} & 500                & 0.047        & 0.012          & 0.956    & 1.319     & 0.049        & 0.040          & 0.914    & 1.435     \\
                     & 1000               & 0.065        & 0.011          & 0.954    & 0.695     & 0.065        & 0.011          & 0.948    & 1.251     \\
                     & 2000               & 0.050        & 0.008          & 0.954    & 0.434     & 0.048        & 0.008          & 0.932    & 0.655     \\
                     & 3000               & 0.050        & 0.003          & 0.964    & 0.342     & 0.046        & 0.003          & 0.946    & 0.496     \\
                      \hline 
\multirow{4}{*}{400} & 500                & 0.055        & 0.030          & 0.872    & 0.847     & 0.057        & 0.058          & 0.792    & 0.873     \\
                     & 1000               & 0.069        & 0.010          & 0.946    & 1.246     & 0.066        & 0.046          & 0.896    & 0.795     \\
                     & 2000               & 0.066        & 0.009          & 0.966    & 0.515     & 0.058        & 0.006          & 0.948    & 1.181     \\
                     & 3000               & 0.058        & 0.005          & 0.948    & 0.376     & 0.057        & 0.006          & 0.946    & 0.638     \\
                      \hline 
\multirow{4}{*}{800} & 500                & 0.059        & 0.051          & 0.718    & 0.702     & 0.066        & 0.068          & 0.656    & 0.795     \\
                     & 1000               & 0.076        & 0.033          & 0.888    & 0.613     & 0.080        & 0.071          & 0.744    & 0.655     \\
                     & 2000               & 0.072        & 0.015          & 0.952    & 0.844     & 0.077        & 0.043          & 0.862    & 0.578     \\
                     & 3000               & 0.065        & 0.011          & 0.954    & 0.463     & 0.063        & 0.014          & 0.956    & 0.937     \\
                    \hline 
                     \multicolumn{10}{c}{$g(d) = d$, $g^\prime(d)=1$.}  \\    
                     \hline 
\multirow{4}{*}{150} & 500                & 0.052        & 0.014          & 0.956    & 1.337     & 0.054        & 0.040          & 0.930    & 1.445     \\
                     & 1000               & 0.056        & 0.006          & 0.956    & 0.689     & 0.067        & 0.016          & 0.938    & 1.244     \\
                     & 2000               & 0.053        & 0.012          & 0.938    & 0.435     & 0.054        & 0.006          & 0.950    & 0.657     \\
                     & 3000               & 0.044        & 0.010          & 0.948    & 0.341     & 0.048        & 0.005          & 0.936    & 0.494     \\
                      \hline 
\multirow{4}{*}{400} & 500                & 0.059        & 0.036          & 0.890    & 0.857     & 0.063        & 0.054          & 0.760    & 0.888     \\
                     & 1000               & 0.074        & 0.008          & 0.966    & 1.249     & 0.072        & 0.050          & 0.862    & 0.796     \\
                     & 2000               & 0.059        & 0.010          & 0.956    & 0.513     & 0.064        & 0.005          & 0.946    & 1.174     \\
                     & 3000               & 0.057        & 0.009          & 0.938    & 0.378     & 0.056        & 0.010          & 0.934    & 0.638     \\
                      \hline 
\multirow{4}{*}{800} & 500                & 0.055        & 0.046          & 0.732    & 0.696     & 0.053        & 0.054          & 0.726    & 0.797     \\
                     & 1000               & 0.080        & 0.031          & 0.890    & 0.613     & 0.080        & 0.069          & 0.736    & 0.650     \\
                     & 2000               & 0.071        & 0.009          & 0.948    & 0.845     & 0.074        & 0.035          & 0.878    & 0.579     \\
                     & 3000               & 0.064        & 0.010          & 0.968    & 0.463     & 0.067        & 0.013          & 0.958    & 0.947    \\
                      \hline  \hline 
\end{tabular}
\end{center}
{\footnotesize Note: ``BiasInit'' and ``BiasDB'' denote the average bias of the initial Lasso estimator $\hat g^\prime(\cdot)$ and the bias-corrected estimator $\tilde g^\prime(\cdot)$ , respectively. ``Coverage'' shows the coverage probability of the 95\% confidence band defined as \eqref{eq: def confidence band} over 500 replications. ``Length'' stands for the point-wise average length of the confidence band.}
\end{table}
 % defined as $ |\mathcal{D}_\tau|^{-1}\sum_{d_0\in\mathcal{D}_\tau}\E\left|\tilde g^\prime(\cdot) - g^\prime(d_0)\right|$ and $ |\mathcal{D}_\tau|^{-1}\sum_{d_0\in\mathcal{D}_\tau}\E\left|\hat g^\prime(\cdot) - g^\prime(d_0)\right|$,
% \usepackage{multirow}
\begin{table}[H]
\begin{center} 
\caption{Simulation Results for Nonlinear $g$ Functions.}
\label{tab:simul_NL}
\begin{tabular}{cc|cccc|cccc}
\hline\hline 
\multirow{2}{*}{$p_x$}  & \multirow{2}{*}{$n$} &   \multicolumn{4}{c}{Full-Sample}                      & \multicolumn{4}{c}{Split-Sample}                  \\
\cline{3-10}
                     &                    & BiasInit & BiasDB & Coverage & Length & BiasInit & BiasDB & Coverage & Length \\
                     \hline 
                     \multicolumn{10}{c}{$g(d) = 0.05(d-3)^2$, $g^\prime(d) = 0.1(d-3)$}   \\    
                     \hline 
\multirow{4}{*}{150} & 500                & 0.058        & 0.011          & 0.962    & 1.341     & 0.052        & 0.033          & 0.912    & 1.436     \\
                     & 1000               & 0.062        & 0.012          & 0.966    & 0.693     & 0.061        & 0.008          & 0.932    & 1.245     \\
                     & 2000               & 0.053        & 0.010          & 0.962    & 0.435     & 0.055        & 0.006          & 0.960    & 0.655     \\
                     & 3000               & 0.043        & 0.008          & 0.934    & 0.341     & 0.047        & 0.003          & 0.952    & 0.496     \\
                     \hline 
\multirow{4}{*}{400} & 500                & 0.059        & 0.036          & 0.862    & 0.848     & 0.056        & 0.059          & 0.744    & 0.881     \\
                     & 1000               & 0.071        & 0.010          & 0.948    & 1.254     & 0.071        & 0.044          & 0.868    & 0.792     \\
                     & 2000               & 0.064        & 0.010          & 0.960    & 0.515     & 0.066        & 0.007          & 0.946    & 1.180     \\
                     & 3000               & 0.056        & 0.009          & 0.950    & 0.378     & 0.059        & 0.005          & 0.936    & 0.636     \\
                     \hline 
\multirow{4}{*}{800} & 500                & 0.062        & 0.047          & 0.726    & 0.701     & 0.064        & 0.066          & 0.682    & 0.796     \\
                     & 1000               & 0.085        & 0.033          & 0.890    & 0.613     & 0.079        & 0.067          & 0.724    & 0.649     \\
                     & 2000               & 0.073        & 0.021          & 0.950    & 0.840     & 0.071        & 0.037          & 0.896    & 0.576     \\
                     & 3000               & 0.063        & 0.011          & 0.954    & 0.460     & 0.064        & 0.021          & 0.936    & 0.935     \\
                       \hline 
                     \multicolumn{10}{c}{$g(d) = 0.02(d-3)^3$, $g^\prime(d) = 0.06(d-3)^2$}   \\    
                     \hline 
\multirow{4}{*}{150} & 500                & 0.055        & 0.015          & 0.968    & 1.336     & 0.054        & 0.026          & 0.890    & 1.431     \\
                     & 1000               & 0.061        & 0.007          & 0.958    & 0.693     & 0.071        & 0.014          & 0.952    & 1.243     \\
                     & 2000               & 0.055        & 0.007          & 0.946    & 0.436     & 0.053        & 0.006          & 0.934    & 0.654     \\
                     & 3000               & 0.047        & 0.006          & 0.960    & 0.341     & 0.047        & 0.004          & 0.932    & 0.495     \\
                      \hline 
\multirow{4}{*}{400} & 500                & 0.066        & 0.037          & 0.870    & 0.848     & 0.066        & 0.067          & 0.760    & 0.881     \\
                     & 1000               & 0.072        & 0.018          & 0.962    & 1.253     & 0.073        & 0.054          & 0.868    & 0.797     \\
                     & 2000               & 0.063        & 0.012          & 0.962    & 0.516     & 0.066        & 0.006          & 0.962    & 1.165     \\
                     & 3000               & 0.060        & 0.008          & 0.958    & 0.377     & 0.060        & 0.003          & 0.946    & 0.636     \\
                      \hline 
\multirow{4}{*}{800} & 500                & 0.057        & 0.050          & 0.710    & 0.706     & 0.061        & 0.064          & 0.678    & 0.799     \\
                     & 1000               & 0.077        & 0.034          & 0.880    & 0.612     & 0.081        & 0.071          & 0.716    & 0.653     \\
                     & 2000               & 0.075        & 0.011          & 0.958    & 0.840     & 0.074        & 0.032          & 0.876    & 0.576     \\
                     & 3000               & 0.064        & 0.009          & 0.948    & 0.463     & 0.062        & 0.009          & 0.966    & 0.941     \\
                      \hline  \hline 
\end{tabular}
\end{center}
{\footnotesize Note: ``BiasInit'' and ``BiasDB'' denote the average bias of the initial Lasso estimator $\hat g^\prime(\cdot)$ and the bias-corrected estimator $\tilde g^\prime(\cdot)$ , respectively. ``Coverage'' shows the coverage probability of the 95\% confidence band defined as \eqref{eq: def confidence band} over 500 replications. ``Length'' stands for the point-wise average length of the confidence band.}
% defined as $ |\mathcal{D}_\tau|^{-1}\sum_{d_0\in\mathcal{D}_\tau}\E\left|\tilde g^\prime(\cdot) - g^\prime(d_0)\right|$ and $ |\mathcal{D}_\tau|^{-1}\sum_{d_0\in\mathcal{D}_\tau}\E\left|\hat g^\prime(\cdot) - g^\prime(d_0)\right|$,
\end{table}
\label{sec: numerical}

\section{The Nonlinear Effect of Pollution on Migration}\label{sec: emp}
In this section, we revisit the study of pollution and migration. \cite{chen2022} studied the effects of air pollution on migration in China using changes in the average strength of thermal inversions over five-year periods as a source of exogenous variation for medium-run air pollution levels. Their findings suggested that air pollution is responsible for large changes in inflows and outflows of migration in China. Specifically, they found that a 10\% increase in air pollution is capable of reducing the population through outmigration by approximately 2.8\% in a given county.

Air pollution might have nonlinear effects on net outmigration. We use the proposed nonlinear treatment effects estimator on the same dataset in \cite{chen2022} and enrich it with many other covariates. Our method has two advantages compared to those of the original study. First, we can explore the potential nonlinear effect of pollution on migration. Second, we can improve the estimation accuracy by explicitly controlling for high-dimensional characteristics in the model. Specifically, we consider the following empirical econometric model:
\begin{align} \label{eq: emp}
 \ddot{Y}_i& = g(\ddot{D}_i)+\ddot{X}_{i\cdot}^{\top}\theta+u_i,\\
\ddot{D}_i&=  \psi(\ddot{Z}_{i}) + \ddot{X}_{i\cdot}^{\top}\varphi  + v_i,
\end{align}
where for any variable $x_i$, the notation  $\ddot{x}_i$ denotes its standardized version with a zero sample mean and a unit standard deviation. 

\begin{table}[H]
\caption{Summary Statistics}
\label{tab: summary statistics}
\begin{tabular}{llrrrrr}
\hline \hline 
Variable & Descriptions                                  & Mean      & Std       & Min      & Median    & Max        \\
\hline 
    $Y$        & Migration                            & 7.96     & 10.06     & -15.13  & 8.08      & 44.71      \\
$D$       & PM2.5 in $\mu\text{g/m}^3$                                        & 51.38    & 20.35     & 3.88     & 48.37     & 112.10     \\
$Z$         & Thermal Inversions  & 0.17      & 0.17      & 0.00     & 0.10      & 0.88       \\
$X_1$       & Death Rate (per thousand)                                   & 1.28      & 0.68      & 0.29     & 1.20      & 2.65       \\
$X_2$      & Temperature                                   & 13.25     & 1.96      & 6.64     & 13.25     & 20.02      \\
$X_3$      & Precipitation (in mm)                           & 696.00    & 204.22    & 17.40    & 697.19    & 1325.41    \\
% $X_4$       & $\text{Precipitation}^2$              & 5.26\texttt{e}5 & 2.87\texttt{e}5 & 3.03\texttt{e}2   & 4.86\texttt{e}5 & 1.76\texttt{e}6 \\
$X_4$       & Sunshine Duration                             & 5.10      & 0.98      & 2.24     & 5.10      & 8.39       \\
$X_5$       & Humidity                                      & 60.18     & 7.04      & 33.84    & 60.22     & 85.01      \\
$X_6$       & Wind                                          & 4.44      & 1.63      & 1.34    & 4.43      & 7.44       \\
$X_7$       & Income                                        & 5.89\texttt{e}4  & 2.49\texttt{e}4  & 1.30\texttt{e}4 & 5.89\texttt{e}4  & 1.06\texttt{e}5 \\
$X_8$       & Expenditure                                   & 2.66\texttt{e}4  & 1.36\texttt{e}4  & 2.56\texttt{e}3   & 2.45\texttt{e}4  & 6.65\texttt{e}4   \\
$X_{9}$       & Spending in Health                            & 5.24\texttt{e}3   & 3.07\texttt{e}3   & 8.02\texttt{e}2   & 4.54\texttt{e}3   & 1.30\texttt{e}4   \\
$X_{10}$        & Investment in Education                          & 4.70\texttt{e}3   & 2.28\texttt{e}3   & 8.01\texttt{e}2   & 4.41\texttt{e}3   & 9.98\texttt{e}3    \\
$X_{11}$        & Grain Subsidy per Capita                      & 288.06    & 679.11    & 0.00     & 0.00      & 2891.57   \\
$X_{12}$        & Workforce in Agriculture                  & 0.36      & 0.10      & 0.01     & 0.36      & 0.81      \\
\hline \hline 
\end{tabular} \\
{\footnotesize Note:  $Y$ is measured by destination-based net immigration ratio  \citep{chen2022}. For thermal inversions, we use the five-year average strength from 2006 to 2010. Sample size $n=2860$. Data sources: outcome variable: China population census, pollution: WUSTL surface PM2.5 data, thermal inversion: MERRA-2, other controls: China Statistical Yearbook.}
\end{table}
\par Instead of using the fixed effect to account for unobserved factors of migration, we explicitly use the county-level controls of one cross-sectional data point, which is the five-year average of 2006-2010 that contains 2860 counties. Table \ref{tab: summary statistics} provides the summary statistics of the main variables. The outcome variable, $Y$, denotes the measure of migration in county $i$. Specifically, we use the destination-based net immigration ratio, which is the fraction of people entering a county but with their \emph{hukou} in their place of origin, net of the outflows. The treatment variable, denoted as $D$, measures the 5-year average concentration of PM2.5. The key identification strategy is the exogenous variation in thermal inversion \citep{Ransom1995,Arceo2016}. A thermal inversion refers to an abnormal temperature-altitude gradient, where the air becomes hotter instead of cooler with higher altitude and traps pollutants near the ground. Other key covariates include income, expenditure, health, education, grain subsidies, the workforce in agriculture, temperature, precipitation, sunshine, humidity, and wind. To evaluate the performance of our procedure for high-dimensional covariates, we add another 105 geo-economic variables to control for potential omitted variable bias\footnote{These variables include but are not limited to industrial, educational, and environmental variables from the China City Statistical Yearbook and China County Statistical Yearbook. We use a population- or GDP-weighted method to convert some city-level variables to county-level variables.}.
\par Following the simulation studies, we take 1000 grid points between the 10th and 90th percentiles of $D_i$. The selection of tuning parameters also follows the simulation section. Panel (A) of  Figure \ref{empfig1} illustrates the baseline results. The horizontal axis represents the standardized pollution level $D$, and the vertical axis shows the marginal effect. We see a clear nonlinear pattern of the treatment effect. When the pollution level is below the mean, the marginal effect of pollution on immigration is insignificant, showing that people can tolerate low-level pollution. When pollution approaches the mean value such that the standardized $D$ is approximately zero, air pollution shows a significant negative marginal effect on immigration; thus, pollution starts to drive out the population. As pollution becomes more severe but is at a moderately high level (within one standard deviation from the mean), the effect becomes insignificant. When the pollution level is very high, e.g., exceeding 1.3 standard deviations from the mean, the negative effect becomes significantly negative again and the magnitude continues to increase, indicating that people generally cannot tolerate high-level pollution. Quantitatively, in the significant range of treatment effects, at the low-medium level of pollution, a standard deviation increase in pollution is responsible for an approximately -0.6 standard deviation decrease in destination-based immigration. At the very high level of pollution, a standard deviation increase in pollution is responsible for up to a -1.8 standard deviation decrease in destination-based immigration.
\par The ineffectiveness of air pollution in the medium range may be explained by the reference dependence in prospect theory \citep{Kahneman1979}. People in counties with moderately high pollution can tolerate a moderate level of pollution since they treat very high pollution levels as their reference point. Reference points are found to be an important factor in human behavior, such as retirement plans \citep{Seibold2021} and labor supply \citep{Crawford2011}. Another possible explanation is adaptive behavior, such as the use of air purifiers when the pollution level increases.
\begin{figure}[H]
\begin{center}
\includegraphics[scale=0.8]{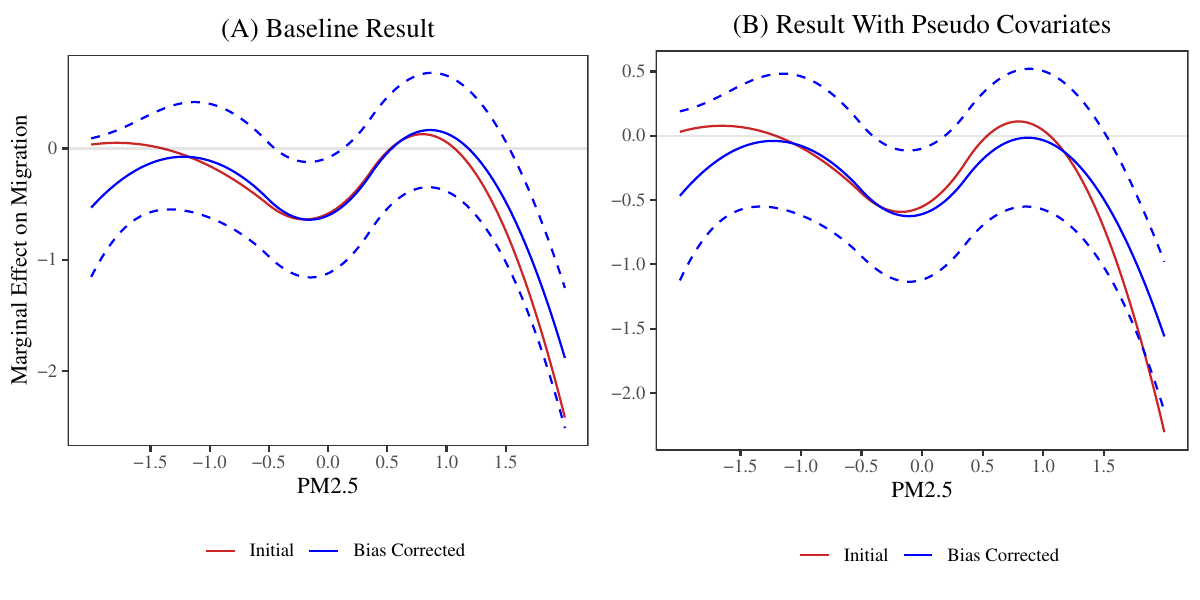}
\caption{Nonlinear Effects of Pollution on Migration} 
\label{empfig1}
\end{center}
{\footnotesize Note: The blue dashed lines show the 95\% uniform confidence bands.}
\end{figure}
\par For a further robustness check, we randomly generate 100 i.i.d.\ standard normal variables and add them to the covariate set. We expect them to be identified as pure noise by the proposed procedure such that they do not severely impact the results. The additional result shown in Panel (B) of Figure \ref{empfig1} satisfies our expectation, as it presents a very similar pattern as in Panel (A) without these irrelevant noises.

\section{Conclusion}\label{sec: conc}
We propose an inferential procedure to construct a uniform confidence band of the endogenous marginal effect function. We include high-dimensional covariates to account for potential omitted variable bias and adopt the control function approach to identify the endogenous effect function. Our main novelty is to design a novel double bias correction procedure to tackle the unique challenge in the nonparametric inference of under endogeneity and the high-dimensional covariate complexity. Our empirical application finds the nonlinear effect of air pollution on migration, which is an important complement to the recent empirical literature. To better facilitate practical studies, it is an exciting topic for future studies to extend our methodology to a model with multiple endogenous variables and heteroskedastic errors, which are also common in empirical applications.

\bibliography{NPIV}
\bibliographystyle{apalike}

\setcounter{section}{0}
\renewcommand\thesection{\Alph{section}}
\clearpage
\setcounter{page}{1}
\setcounter{footnote}{0}
{\Large \bf 
\begin{center}
Appendices to ``Inference for Nonlinear Endogenous Treatment Effects Accounting for High-Dimensional Covariate Complexity''
\end{center}
}
{
\begin{center}
{\sc Qingliang Fan}$^{\dag}$, {\sc Zijian Guo}$^{\ddag}$, {\sc Ziwei Mei}$^{\dag}$, {\sc Cun-Hui Zhang}$^{\ddag}$ \\
{$^{\dag}$Department of Economics, The Chinese University of Hong Kong}\\
{$^{\ddag}$Department of Statistics, Rutgers University}
\end{center}
}

\vspace{-1em}
\newcounter{counter}[section]
\setcounter{table}{0}
\renewcommand{\thetable}{\thesection\arabic{table}}
\setcounter{equation}{0}
\renewcommand\theequation{\thesection\arabic{equation}}
\setcounter{figure}{0}
\renewcommand\thefigure{\thesection\arabic{figure}}
\setcounter{Theorem}{0}
\setcounter{Assumption}{0}
\setcounter{Lemma}{0}
\setcounter{Remark}{0}
\setcounter{Corollary}{0}
\setcounter{Proposition}{0}
\setcounter{Definition}{0}
\renewcommand\theTheorem{\thesection\arabic{Theorem}}
\renewcommand\theLemma{\thesection\arabic{Lemma}}
\renewcommand\theRemark{\thesection\arabic{Remark}}
\renewcommand\theCorollary{\thesection\arabic{Corollary}}
\renewcommand\theProposition{\thesection\arabic{Proposition}}
\renewcommand\theAssumption{\thesection\arabic{Assumption}}
\renewcommand\theDefinition{\thesection\arabic{Definition}}
\newcommand\thealgorithm{\thesection\arabic{algorithm}} 

\setcounter{subsection}{0}
\setcounter{table}{0}
\renewcommand{\thetable}{\thesection\arabic{table}}
\setcounter{equation}{0}
\renewcommand\theequation{\thesection\arabic{equation}}
\setcounter{figure}{0}
\renewcommand\thefigure{\thesection\arabic{figure}}
\setcounter{Theorem}{0}
\setcounter{Assumption}{0}
\setcounter{Lemma}{0}
\setcounter{Remark}{0}
\setcounter{Corollary}{0}
\setcounter{Proposition}{0}

\renewcommand{\thealgocf}{\thesection\arabic{algocf}}
 \setcounter{algocf}{0}
% \section{Additional Discussions}

% \input{Discussion_debias}
% \input{Sample_splitting}
\section{Algorithm of Full-Sample Inference}\label{subsec: algo full sample}
\begin{algorithm}[H]
\footnotesize
\caption{\label{algorithmA1}Full-Sample Uniform Confidence Band for $g^\prime(d)$}
\hspace*{0.01in}
\begin{algorithmic}[1]
\State Data preparations: Prepare the dataset $\mathscr{D}_n$. 
\State Initial estimators: Obtain the full sample LASSO estimators $\hat\kappa,\hat\varphi$ by (\ref{eq: partial penalization for D}) and $\hat\beta,\hat\eta,\hat\theta$ by (\ref{eq: partial penalization}) with $\mathscr{D}_n$.
\State Control function approximation: Save the residuals $\hat v_i = D_i - K_\ic^\top \hat\kappa - X_\ic^\top\hat\varphi$. Estimate $q^\prime(v_i)$ by $\hatqp(\hat v_i) = H^\prime(\hat v_i)^\top\hat\eta$.
\State Double bias correction: Construct $\hat F_\ic$ as (\ref{eq: def F hat i}). Solve (\ref{eq: projection direction 1}) and save the solution $\hat\Omega_B$.
\State Multiplier bootstrap: Compute the $(1-\alpha)$-th quantile of $\sup_{d\in\mathcal{D}}\left| \hat{\mathbb{H}}_n(d)\right|$ with $\Hat{\mathbb{H}}_n(d)$ defined as (\ref{eq: H hat}). Save the quantile as $\hat c_n(\alpha)$.
\State Uniform confidence band: Construct the confidence band of $g^\prime(d)$ by (\ref{eq: def confidence band}).
\end{algorithmic}
\end{algorithm}
\section{The Feasibility of (\ref{eq: projection direction 1}) When  $q$ Is Linear} \label{sec: linear q}
The control function $q(\cdot)$ is usually unknown in practice. The methodology proposed in the main text accommodates the special case of a linear $q$ function. To simplify the illustration, we assume $M_D = M_v = M_\ell = M$, and that $\Sigma_{F|\mathcal{L}}$ is observable so that our goal is to find the vector $\xi_j$ such that
\begin{equation}
    \Sigma_{F|\mathcal{L}}\xi_j = \textbf{i}_j \label{eq: sol linear q}
\end{equation} 
for $j\in[M]$, where $\textbf{i}_j$ is the $j$-th standard basis.  When $q(\cdot)$ is linear, there exists a constant $q_0$ such that $q^\prime(v) \equiv q_0$ for all $v$. Define $\Psi_i = (B_\ic^\top, H_\ic^\top, q_0K_\ic^\top)^\top \in \mathbb{R}^{M_{\rm all}}$ with $M_{\rm all} := (p_z+2)M$. Recall $X_\ic\in \mathbb{R}^p$. Thus, we have 
\[\Sigma_{F|\mathcal{L}} = \left(\begin{array}{ccc}
  \E_{\mathcal{L}}[\Psi_\ic \Psi_\ic ^\top]   &  \E_{\mathcal{L}}[\Psi_\ic X_\ic^\top]  & q_0\E_{\mathcal{L}}[\Psi_\ic X_\ic^\top]\\
      \E_{\mathcal{L}}[X_\ic \Psi_\ic ^\top]   &  \E_{\mathcal{L}}[X_\ic X_\ic^\top]  & q_0\E_{\mathcal{L}}[X_\ic X_\ic^\top]\\ 
         q_0\E_{\mathcal{L}}[X_\ic \Psi_\ic ^\top]   &  q_0\E_{\mathcal{L}}[X_\ic X_\ic^\top]  & q_0^2\E_{\mathcal{L}}[X_\ic X_\ic^\top]\\ 
\end{array}\right) =: \left(\begin{array}{ccc}
   A_{11}  & A_{12} &  q_0 A_{12}\\
   A_{12}^\top  & A_{22}  & q_0 A_{22} \\
    q_0 A_{12}^\top  & q_0 A_{22}  & q_0^2 A_{22} \\
\end{array}\right)\]
where $A_{11} :=  \E_{\mathcal{L}}[\Psi_\ic \Psi_\ic ^\top]$, $A_{22} := \E_{\mathcal{L}}[X_\ic X_\ic^\top]$ and $A_{12}:= \E_{\mathcal{L}}[\Psi_\ic X_\ic^\top] $. By similar eigenvalue conditions in Assumption \ref{assu: more compatibility}, the submatrix 
$$\Sigma_{11} := \left(\begin{array}{cc}
   A_{11}  & A_{12}  \\
   A_{12}^\top   &  A_{22}
\end{array}\right) \in \mathbb{R}^{(M_{\rm all}+p)\times(M_{\rm all}+p)}$$ is invertible. Define 
\[\Omega_O := \left(\begin{array}{cc}
   \Sigma_{11}^{-1}  & O_{(M_{\rm all}+p)\times p}  \\
   O_{p\times (M_{\rm all}+p)}   &  O_{p\times p}
\end{array}\right)\] 
and $$\Sigma_{12}:= (A_{12}^\top,A_{22})^\top.$$ We show the solutions for $\{\xi_j\}_{j\in[M]}$ are the first $M$ columns of $\Omega_O$. Note that   
\[\begin{aligned}
\Sigma_{F|\mathcal{L}} \Omega_O = \left(\begin{array}{cc}
   \Sigma_{11}  &  q_0\Sigma_{12}  \\
   q_0\Sigma_{12}^\top   &  q_0^2 A_{22}
\end{array}\right) \cdot \left(\begin{array}{cc}
   \Sigma_{11}^{-1}  & O_{(M_{\rm all}+p)\times p}  \\
   O_{p\times (M_{\rm all}+p)}   &  O_{p\times p}
\end{array}\right) = \left(\begin{array}{cc}
   I_{M_{\rm all}+p}   & O_{(M_{\rm all}+p)\times p}  \\
   q_0\Sigma_{12}^\top\Sigma_{11}^{-1} &  O_{p\times p}
\end{array}\right).
\end{aligned}\]
Observe that $\Sigma_{12}^\top$ is the last $p$ rows of $\Sigma_{11}$. Thus, $\Sigma_{12}^\top\Sigma_{11}^{-1}$ is the last $p$ rows of the identity matrix $I_{M_{\rm all}+p}$. It turns out that the first $M$ columns of $q_0\Sigma_{12}^\top\Sigma_{11}^{-1}$ are all zero. Consequently, the first $M$ columns of $\Sigma_{F|\mathcal{L}} \Omega_O = \left(\begin{array}{cc}
   I_{M_{\rm all}+p}   & O_{(M_{\rm all}+p)\times p}  \\
   q_0\Sigma_{12}^\top\Sigma_{11}^{-1} &  O_{p\times p}
\end{array}\right)$ are $\{\textbf{i}_j\}_{j\in[M]}$ so that the first $M$ columns of $\Omega_O$ are the solutions to (\ref{eq: sol linear q}).  
\newcommand{\supg}{\sup_g }
\newcommand{\lepg}{\lesssim_{u.p.}}
\newcommand{\gepg}{\gtrsim_{u.p.}}
\newcommand{\asymppg}{\asymp_{u.p.}}
%%%%%%%%%%%%%%%%%%%%%%%%%%%%%%%%%%%%%%%%%%%%%%%%%%%%%%%%%%%%%%%%%%%%%%%%%%%%%%%%
\section{Proofs}\label{sec: proof}
%%%%%%%%%%%%%%%%%%%%%%%%%%%%%%%%%%%%%%%%%%%%%%%%%%%%%%
\par Throughout the proofs, we use $c$ and $C$ (sometimes with subscripts) to denote generic positive constants irrelevant to the sample size, which may vary from place to place. The notations ``$\lesssim_p$'', ``$\gtrsim_p$'', and $\asymp_p$ indicate ``$\lesssim$'', ``$\gtrsim$'', and ``$\asymp$'' hold w.p.a.1 respectively.  We assume that function $q$ and its derivatives $q^\prime(\cdot)$, $q^{\prime\prime}(\cdot)$, as well as the B-splines $H(\cdot)$ and the derivatives $H^\prime(\cdot)$, $H^{\prime\prime}(\cdot)$ are continuously extended to the whole real line. Specifically, $q(v) = q(a_v - \epsilon_v)$ for all $v < a_v - \epsilon_v $, and $q(b_v + \epsilon_v)$ for all $v > b_v + \epsilon_v$. Therefore, $\hat v_i$ will always fall into the support of $v_i$. 
\subsection{Definitions} \label{subsec: def}
\par We first define the B-spline basis following \citet[Section E.1]{chen2018optimal}. We consider a uniform B-spline basis with $m$ interior knots and support $[a,b]$. Let $a = t_{-k} =\cdots = t_0 < t_1 < \cdots < t_{m} = t_{m+1} = \cdots = t_{m+k-1} = b$ denote the extended knot sequence. Let $I_0 = [t_0,t_1),\cdots,I_m = [t_m,t_{m+1}]$. As mentioned in the beginning of Section \ref{sec: est}, we use the uniformly based knots \citep{schumaker2007spline} for B-spline functions such that the distances between every two adjacent knots are equal. All $I_m$'s thus share the same length.
\par A basis of degree 0 is constructed by 
\[N_{j,0}(x) = \begin{cases}
    1,\ \text{if } x \in I_j, \\
    0,\ \text{otherwise}
\end{cases}\]
for $j=0,\cdots,m$. Bases of degree $k>0$ is defined recursively by 
\[N_{j,k}(x) = \dfrac{x-t_j}{t_{j+k}-t_j}N_{j,k-1}(x) + \dfrac{t_{j+k+1}-x}{t_{j+k+1}-t_j}N_{j+1,k-1}(x) \]
 where $\frac{1}{0}:=0$ following \citet[Section E.1; See also Section 5 of \citet{devore1993constructive}]{chen2018optimal}. Without loss of generality, we define $N_{j,k}(x)=0$ for $j<-k$ or $j>m$. Furthermore, for any $x\in I_j$, $N_{j^\prime,k}(x)$ is nonzero only if $j^\prime = j,\cdots,j+k$.
\par Define $B_{j}^{[k^\prime]}(x) := N_{j-k^\prime-1,k^\prime}(x)$ for any $j$ and $k^\prime$. We use $M$ to denote the number of spline functions used in the estimation and inference. By the definition of $N_{j,r}$ above, $M = m + k + 1$. In addition, the B-splines satisfy partition to unity $\sum_{j\in[M]}B_{j}^{[k^\prime]}(x) = 1$ for all $k^\prime \leq k$. Similar definitions and properties also hold for $H(\cdot)$ and $K_\ell(\cdot)$ in the corresponding supports and we will not repeat the statements. 
\par In addition, we provide formal definitions for (conditionally) sub-Gaussian and sub-exponential random variables and vectors \citep{vershynin2010introduction}. 
 \begin{Definition}[Sub-Gaussian norms] 
 The sub-Gaussian norm of any random variable $x$ (conditional on the sigma-field $\mathcal{F}$) is
 \begin{equation}\label{eq:def-subG}
     \|x\|_{\psi_2|\mathcal{F}} := \sup_{q\geq 1}\dfrac{1}{\sqrt{q}}[\mathbb{E}_{\mathcal{F}}|x|^q]^{1/q}.
 \end{equation}
  For any random vector $X\in\mathbb{R}^p$, we define its (conditional) sub-Gaussian norm as  
   \begin{equation}\label{eq:def-subG-vec}
     \|X\|_{\psi_2|\mathcal{F}} := \sup_{b\in\R^p:\|b\|_2=1} \|b^\top x\|_{\psi_2|\mathcal{F}}. 
 \end{equation}
 A random variable or vector is called (conditionally) sub-Gaussian if its (conditional) sub-Gaussian norm is uniformly bounded by some absolute constant.
 \end{Definition}
 
  \begin{Definition}[Sub-Exponential norms] 
 The sub-exponential norm of any random variable $x$ (conditional on the sigma-field $\mathcal{F}$) is
 \begin{equation}\label{eq:def-sube}
     \|x\|_{\psi_1|\mathcal{F}} := \sup_{q\geq 1}\dfrac{1}{q}[\mathbb{E}_{\mathcal{F}}|x|^q]^{1/q}.
 \end{equation}
  For any random vector $X\in\mathbb{R}^p$, we define its (conditional) sub-exponential norm as  
   \begin{equation}\label{eq:def-sube-vec}
     \|X\|_{\psi_1|\mathcal{F}} := \sup_{b\in\R^p:\|b\|_2=1} \|b^\top x\|_{\psi_1|\mathcal{F}}. 
 \end{equation}
 A random variable or vector is called (conditionally) sub-exponential if its (conditional) sub-exponential norm is uniformly bounded by some absolute constant.
 \end{Definition}

%%%%%%%%%%%%%%%%%%%%%%%%%%
\subsection{Preliminary Propositions}\label{subsec: prop}
Proposition \ref{prop: DB} provides widely used sup-norm bounds of cross-products. Propositions \ref{prop: spline approx power}-\ref{prop: spline eigen}  state some properties of spline functions. Proposition \ref{prop: Lasso D} and the consequent corollaries provide the estimation error of $\hat v_i$ and its functions from the LASSO algorithm (\ref{eq: partial penalization for D}). Throughout the whole Section \ref{subsec: prop}, all probabilistic bounds are independent of the function $g$, and thus hold uniformly for any $g\in\mathcal{H}_{\mathcal{D}}(\gamma,L)$. 

\begin{Proposition}
\label{prop: DB}Under Assumptions \ref{assu: control func} and \ref{assu: D v}, there exist some absolute constants $c$ and $C$ such that with probability at least $1-c(pM)^{-4}$ 
\begin{equation}
    \left\|\dfrac{X^\top X}{n} - \mathbb{E}\left[\dfrac{X^\top X}{n}\right]\right\|_\infty + \left\|\dfrac{X^\top v}{n}  \right\|_\infty + \left\|\dfrac{X^\top \varepsilon}{n}  \right\|_\infty \leq C\sqrt{\dfrac{\log (pM) }{n}}. 
\end{equation}
for any $t > 0$. 
\end{Proposition}
\begin{proof}[Proof of Proposition \ref{prop: DB}]Since $v$ and $X$ are uniformly bounded, their cross products have bounded sub-Gaussian norm. The first two terms can be thus easily bounded using the Hoeffding-type inequality for centered sub-Gaussian variables \citep[Proposition 5.10]{vershynin2010introduction} and the union bound. The procedures in the proof of Lemma B2 in \citet{fan2022testing} apply to this proposition. 
\par The bound of the third term applies the moderate deviation inequality \citep[Lemma 5]{belloni2012sparse} that 
\begin{equation}\label{eq: moderate deviation}
    \begin{aligned}
    \Pr\left(\max_{j\in[p_x]}\dfrac{\left|\sumn X_{ij}\varepsilon_i\right|}{\sqrt{\sumn X_{ij}^2\varepsilon_i^2}} > \Phi^{-1}(1-\gamma_0/2p_x)\right) \leq \gamma_0 (1+A/\ell_n^3)
\end{aligned}
\end{equation} 
where $A$ is an absolute constant and $\Phi$ is the cumulative distribution function of a standard normal variable, provided that $\ell_n > 0$ and under the i.i.d.\ setting, 
\begin{equation}\label{eq: moderate deviation condition}
    0 \leq \Phi^{-1}(1-\gamma_0/2p) \leq \dfrac{n^{1/6}}{\ell_n} \min_{j\in[p]} \dfrac{\left( \E(X_{ij}^2\varepsilon_i^2)\right)^{1/2}}{\left(  \E|X_{ij}^3\varepsilon_i^3|\right)^{1/3}} - 1. 
\end{equation}
The following arguments verify \eqref{eq: moderate deviation condition}. By the boundness of $X_{ij}$ in Assumption \ref{assu: D v} and the bounded high-order conditional moment of $\varepsilon_i$ in Assumption \ref{assu: control func}, 
\[\E|X_{ij}^3\varepsilon_i^3| \lesssim \E|\varepsilon_i^3| \lesssim 1.\]
By the lower and upper bounds of $\E(X_{ij}^2)$ in Assumption \ref{assu: eigen}, 
\[\E(X_{ij}^2\varepsilon_i^2) = \E(X_{ij}^2\E(\varepsilon_i^2|X_{ij})) = \sigma_\varepsilon^2 \E(X_{ij}^2) \asymp 1. \]
Thus, if we take $\ell_n=A$, the right hand side of \eqref{eq: moderate deviation condition} is lower bounded by $cn^{1/6}$ for some absolute constant $c$. Then when $\gamma_0=M^{-1}$, \eqref{eq: moderate deviation condition} follows by the fact that $\Phi^{-1}(1-\gamma_0/2p) \lesssim \sqrt{\log(2p/\gamma_0)} \lesssim \sqrt{\logpM}$.
\par The validity of \eqref{eq: moderate deviation condition} implies \eqref{eq: moderate deviation} when $\gamma_0 = M^{-1}$ and $\ell_n=A$, which also implies that the right hand side of  \eqref{eq: moderate deviation} converges to zero. Thus, w.p.a.1 
\[\begin{aligned}
    \|n^{-1}X^\top\varepsilon\|_\infty &\leq n^{-1/2}\cdot \Phi^{-1}(1-\gamma_0/2p)\cdot \max_{j\in[p]}\sqrt{n^{-1}\sumn X_{ij}^2\varepsilon_i^2 } \\
    &\lesssim n^{-1/2} \cdot \sqrt{\logpM} \cdot \sqrt{n^{-1}\sumn \varepsilon_i^2 }\\
    &\lep \sqrt{\logpM / n},
\end{aligned}\]
where the second row applies the $\Phi^{-1}(1-\gamma_0/2p)\lesssim \sqrt{\logpM}$ and the boundness of $X_{ij}$, and the third row applies $n^{-1}\sumn \varepsilon_i^2 \convp \sigma_\varepsilon^2$ by the i.i.d.\ and the third moment conditions. 
\end{proof}

\begin{Proposition}\label{prop: Spline2norm}For all $1\leq i\leq n$ we have  \[\max\{\|B_{i\cdot}\|_2,   \|H_{i\cdot}\|_2,  (\|(K_\ell)_{i\cdot}\|_2)_{\ell\in[p_z]} \}\leq 1,\] 
\[\min\{\|B_{i\cdot}\|_2,   \|H_{i\cdot}\|_2,   (\|(K_\ell)_{i\cdot}\|_2)_{\ell\in[p_z]} \}\geq \dfrac{1}{\sqrt{k}}.\]
\end{Proposition}
\begin{proof}[Proof of Proposition \ref{prop: Spline2norm}]
By \citet[Theorem 4.17, Definition 4.19]{schumaker2007spline} we know that for each fixed $x$, the cardinality of the set $\{j\in[M]:B_{j}(x) \neq 0\}$ is no greater than $k$. By partition of unity \citep[Theorem 4.20]{schumaker2007spline} that  $\sum_{j=1}^{M}B_{j}(D_i)=1$, we deduce that 
$$ \|B_{i\cdot}\|_2 \leq  \sum_{j=1}^{M}B_{j}(D_i) =1,$$ 
$$ \|B_{i\cdot}\|_2 \geq  \dfrac{1}{\sqrt{k}}\sum_{j=1}^{M}B_{j}(D_i) = \dfrac{1}{\sqrt{k}}.$$ 
Likewise, we can deduce the same bounds for
$\|H_{i\cdot}\|_2$ and $\|(K_\ell)_{i\cdot}\|_2$.
\end{proof}

The following Proposition \ref{prop: SplineDerivative} is about the first-order derivative of B-Spline functions. Proposition \ref{prop: SplineDerivative} is a direct result of  \citet[p. 115]{de1978practical}. 
\begin{Proposition} 
For any $v\in[a_v-\epsilon_v,b_v+\epsilon_v]$, 
\label{prop: SplineDerivative}
$H_j^\prime(v)=M\left(H_{j}^{[k-1]}(v)-H_{j+1}^{[k-1]}(v)\right)$, where $H_{j}^{[k-1]}(v)$ is the $j$-th B-Spline function of degree $k-1$. Similarly, for any $d\in[a_D,b_D]$, 
$B_j^\prime(d) = M\left(B_{j}^{[k-1]}(d)-B_{j+1}^{[k-1]}(d)\right)$.
\end{Proposition}

\begin{Proposition} 
\label{prop: SplineDerivativeL2Norm} $\|B^\prime(d))\|_q \asymp M$ and $\|H^\prime(v)\|_q\asymp M$ for $q=1,2$ uniformly for all $d\in[a_D,b_D]$ and $v\in[a_v-\epsilon_v,b_v+\epsilon_v]$. Also, $\|B^{\prime\prime}(d)\|_q,\|H^{\prime\prime}(v)\|_q\lesssim M^2$.
\end{Proposition}
\begin{proof}[Proof of Proposition \ref{prop: SplineDerivativeL2Norm}]
% By Propositions \ref{prop: Spline2norm} and \ref{prop: SplineDerivative}, 
% \[\|B^\prime(d)\|_2 \leq 2M\left( \sqrt{\sum_{j=1}^M\left[B^{(k-1)}_{j}(d)\right]^2} + \sqrt{\sum_{j=1}^M\left[B^{(k-1)}_{j+1}(d)\rMight]^2}\right) \leq 2M.\]
% In addition, d
We only prove this proposition for $B(d)$ and the arguments for $H(v)$ are exactly the same. We first focus on $q=2$. From Proposition \ref{prop: SplineDerivative}, we have $\|B^\prime(d)\|_2 =\sqrt{\sum_{j=1}^M(B_j^\prime(d))^2} = M\sqrt{\sum_{j=1}^M [B_j^{[k-1]}(d) - B_{j+1}^{[k-1]}(d)]^2}$. Given $B_{M+1}^{[k-1]}(d) = 0$, it is easy to show that 
\[\begin{aligned}
    \sum_{j=1}^M [B_j^{[k-1]}(d) - B_{j+1}^{[k-1]}(d)]^2 \leq 2\sum_{j=1}^M[B_j^{[k-1]}(d)]^2 + 2\sum_{j=1}^M[B_{j+1}^{[k-1]}(d)]^2
    \leq 4\|B^{[k-1]}(d)\|_2^2 \lesssim 4.
\end{aligned}\]
We then show the other side of the inequality. Recall that for any $d$, there exists at most $k$ numbers $j\in[M]$ such that $B_{j}^{[k-1]}(d)\neq 0$. Suppose these $k$ numbers are $j_0,j_0+1,\cdots,j_0+k-1$. Consequently, 
\[\sum_{j=1}^M [B_j^{[k-1]}(d) - B_{j+1}^{[k-1]}(d)]^2 \geq \frac{1}{k}\left(\sum_{j=j_0}^{j_0+k-1} \left|B_j^{[k-1]}(d) - B_{j+1}^{[k-1]}(d)\right|\right).\]
It suffices to show that $\sum_{j=j_0}^{j_0+k-1} \left|B_j^{[k-1]}(d) - B_{j+1}^{[k-1]}(d)\right|$ is bounded away from zero. We will show by contradiction that a lower bound is $\frac{1}{k+1}$. Suppose that $\sum_{j=j_0}^{j_0+k-1} \left|B_j^{[k-1]}(d) - B_{j+1}^{[k-1]}(d)\right| < \frac{1}{k+1}$. Then
\[\begin{aligned}
    1 = \sum_{j=j_0}^{j_0+k-1} B_j^{[k-1]}(d) = \sum_{j=j_0}^{j_0+k-1} \sum_{i=j}^{j_0+k-1} \left[B_i^{[k-1]}(d) - B_{i+1}^{[k-1]}(d)\right] < \dfrac{k}{k+1} 
\end{aligned}\]
which is a contradiction. 
Therefore, $\|B^\prime(d)\|_2\asymp M$. As for the $L_1$ norm, note that 
\[M \lesssim \|B^\prime(d)\|_2 \leq \|B^\prime(d)\|_1 = M \sum_{j=1}^M |B_j^{[k-1]}(d) - B_{j+1}^{[k-1]}(d)| \leq 2M \]
where the last inequality applies the partition to unity property that $\sum_{j\in[M]} B_j^{[k-1]}(d) = 1$. As for the second-order derivative, note that 
\[\begin{aligned}
    \|B^{\prime\prime}(d)\|_1 &= \sum_{j\in[M]} |B_j^{\prime\prime}(d)| \leq M\sum_{j\in[M]}|(B_j^{[k-1]})^\prime(d) - (B_{j+1}^{[k-1]})^\prime(d)| \leq M\cdot 2\|(B^{[k-1]})^\prime(d)\|_1 \lesssim M^2.  
\end{aligned}\]
This completes the proof of Proposition \ref{prop: SplineDerivativeL2Norm}. 
\end{proof}

The following proposition about eigenvalues of Gram matrices for the spline bases and their derivatives. 
\begin{Proposition}
\label{prop: spline eigen}Define $B^\prime_\ic=B^\prime(D_i)$ and $H^\prime_\ic = H^\prime(v_i)$. Under Assumptions \ref{assu: D v}, there exists some absolute constants $c$ and $C$ such that for $M$ large enough, 
\begin{equation}\label{eq: spline eigen}
\begin{aligned}
    cM^{-1} \leq \lambda_{\min}(\E(B_\ic B_\ic^\top)) \leq  \lambda_{\max}(\E(B_\ic B_\ic^\top)) \leq CM^{-1} \\
        cM^{-1} \leq \lambda_{\min}(\E(H_\ic H_\ic^\top)) \leq  \lambda_{\max}(\E(H_\ic H_\ic^\top)) \leq CM^{-1} \\
            cM^{-1} \leq \lambda_{\min}(\E(K_\ic K_\ic^\top)) \leq  \lambda_{\max}(\E(K_\ic K_\ic^\top)) \leq CM^{-1}. \\
\end{aligned}
\end{equation} 
and 
\begin{equation}\label{eq: spline deriv eigen}
\begin{aligned}
    \lambda_{\max}(\E(B^\prime_\ic (B^\prime_\ic)^\top)) \leq CM \\
          \lambda_{\max}(\E(H^\prime_\ic (H^\prime_\ic)^\top)) \leq CM.
\end{aligned}
\end{equation} 
\end{Proposition}
\begin{proof}[Proof of Proposition \ref{prop: spline eigen}]
Equation \eqref{eq: spline eigen} is a direct result of \citet[Lemmas E.1 and E.2]{chen2018optimal}. For \eqref{eq: spline deriv eigen}, we only show the bounds for $H^\prime$ since the same arguments apply to $B^\prime$. Define \[\Delta H_\ic := \left(H_{j}^{[k-1]}(v_i) - H_{j+1}^{[k-1]}(v_i)\right)_{j\in[M]},\] $H^{[k-1]}_\ic := \left(H_{j}^{[k-1]}(v_i)\right)_{j\in[M]}$ and $H^{[k-1]}_{+1,\ic} := \left(H_{j+1}^{[k-1]}(v_i)\right)_{j\in[M]}$. Note that 
\[\begin{aligned}
    \|\E(H_\ic^\prime H_\ic^{\prime\top})\|_2 = M^2\left\|\E\left[\Delta H_\ic^{\ot}\right]\right\|_2 \leq M^2 \left(\left\|\E[H^{(k-1)\ot}_\ic]\right\|_2+\left\|\E[H^{(k-1)\ot}_{+1,\ic}]\right\|_2\right) \lesssim M
\end{aligned}\]
where the first equality applies Proposition \ref{prop: SplineDerivative} and the last inequality applies \eqref{eq: spline eigen}. 
\end{proof}

Below are well-known results about B-spline function approximation that are essential in the theoretical analysis; See \citet[pp. 155, Theorem (26)] {de1978practical} and \citet[Eq.\ (14)]{zhou2000derivative}. The following proposition only shows the spline approximation errors of the function $g(\cdot)$; Similar results hold for other functions $q(\cdot)$ and $\psi_\ell(\cdot)$ and are not listed here for simplicity.  
\begin{Proposition}\label{prop: spline approx power}
Suppose $\gamma\geq 2$ and $L>0$ are fixed. Then for any function $g\in\mathcal{H}(\gamma,L)$, there exists a $\beta\in\mathbb{R}^{M_D}$ such that 
\begin{equation}
\sup_{d\in{\mathcal{D}}}  \left|\sum_{j=1}^{{M}}\beta_jB_{j}(d)-g(d)\right|=O(M^{-\gamma}),
\label{eq: sup norm approximation}
\end{equation}
\begin{equation}
\sup_{d\in{\mathcal{D}}} \left|\sum_{j=1}^{{M}}\beta_jB^\prime_{j}(d)-g^\prime(d)\right|=O(M^{-\gamma + 1}),
\label{eq: sup norm prime approximation}
\end{equation}
The $O(\cdot)$ holds uniformly for all $g\in\mathcal{H}(\gamma,L)$.  
\end{Proposition}

The following Proposition bounds the errors of LASSO estimators in \eqref{eq: partial penalization for D}. 
\begin{Proposition}
\label{prop: Lasso D}Suppose that Assumptions \ref{assu: control func}-\ref{assu: eigen} hold and $\lambda_D = C_D\sqrt{\dfrac{\logpM}{n}}$ for some $C_D>0$ large enough. There exist some constants $c$ such that with probability at least $1-c(pM)^{-4}$
    \begin{equation}\label{eq: kappa hat error L2 prop p}
    \|\hat\kappa - \kappa\|_2 \lesssim \left(\sqrt{\dfrac{n}{\log (pM)}}M^{-\gamma-0.5} + 1\right)M^{-\gamma+0.5} + \sqrt{\dfrac{M^2\log (pM)}{n}} + \sqrt{\dfrac{sM\log (pM)}{n}}, 
\end{equation}
\begin{equation}\label{eq: hat varphi error L1}
\begin{aligned}
     \|\hat\varphi - \varphi\|_1 &  \leq C\left(\sqrt{\dfrac{n}{\log (pM)}}M^{-2\gamma} + (s+M)\sqrt{\dfrac{\log (pM)}{n}}  \right)
\end{aligned}
\end{equation}
and 
\begin{equation}\label{eq: hat varphi error L2}
\begin{aligned}
     \|\hat\varphi - \varphi\|_2^2 &  \leq C\left( M^{-2\gamma} + (s+M)\cdot\dfrac{\log (pM)}{n}\right).
\end{aligned}
\end{equation}
\end{Proposition}

\begin{proof}[Proof of Proposition \ref{prop: Lasso D}] By the definition of  $\hat\varphi$  and $\hat\kappa$, we have  
\begin{equation}
\frac{1}{n}\|D-K\hat\kappa-X\widehat{\varphi}\|_2^2+\lambda_D\|\widehat{\varphi}\|_1\leq \frac{1}{n}\|D-K \kappa-X\varphi\|_2^2+\lambda_D\|\varphi\|_1.
\end{equation}
which implies 
\begin{equation}
\frac{1}{n}\|K(\kappa-\widehat{\kappa})+X(\varphi-\widehat{\varphi})\|_2^2+\lambda_D\|\widehat{\varphi}\|_1\leq \frac{2}{n}\left( r_\psi+v\right)^\top\left[K(\kappa-\widehat{\kappa})+X(\varphi-\widehat{\varphi})\right] +\lambda_D\|{\varphi}\|_1
\label{eq: basic inequality phi}
\end{equation}
given that $r_{\psi}+v=D-K\kappa-X\varphi$ where $r_{\psi i}:=\sum_{\ell=1}^{p_z}r_{\ell}(Z_{i\ell})$.

In the following, we make use of this basic inequality and further establish the convergence rate of the proposed estimator. The proof consists of two steps.\\
\noindent{\bf Step 1: Deduce a more convenient basic inequality.}  Note that 
\begin{equation}
\frac{1}{n}r_\psi^{\top}\left(K(\kappa-\widehat{\kappa})+X({\varphi}-\widehat{\varphi})\right) \leq R_{D1} \frac{1}{\sqrt{n}}\|K(\kappa-\widehat{\kappa})+X({\varphi}-\widehat{\varphi})\|_2,
\end{equation}
where  \begin{equation}
\begin{aligned}
 R_{D1}=\sqrt{\frac{1}{n}\sum_{i=1}^{n}r_{\psi i}^2} &\lesssim M^{-\gamma}.  
\end{aligned}
\label{eq: R1 bound phi}
\end{equation} 
 By the inequality $ab\leq a^2+ \frac{b^2}{4}$, we have 
\begin{equation}
\begin{aligned}
\frac{1}{n}r_\psi^{\top}\left(K(\kappa-\widehat{\kappa})+X({\varphi}-\widehat{\varphi})\right)\leq R_{D1}^2+ \frac{1}{4{n}}\|K(\kappa-\widehat{\kappa})+X({\varphi}-\widehat{\varphi})\|_n^2
\end{aligned}
\label{eq: basic 1 phi}
\end{equation}
Proposition \ref{prop: DB} implies that  $$ \Pr\left(\left\|\frac{1}{n}X^{\top}v\right\|_{\infty}\geq \dfrac{c_0\lambda_D}{2}\right) \leq c(pM)^{-4} $$  for some $0<c_0<1$ and $\lambda_D = C_\lambda\sqrt{M\log p/n}$ for some $C_\lambda$ large enough, and hence we have with probability at least $1 - c(pM)^{-4}$
\begin{equation}
\frac{1}{n}X^{\top}v({\varphi}-\widehat{\varphi})\leq \dfrac{c_0\lambda_D}{2} \|{\varphi}-\widehat{\varphi}\|_1
\label{eq: basic 2 phi}
\end{equation}

By the uniform boundedness of $v_i$, it has a uniformly bounded sub-Gaussian norm (conditional on $Z$). Given that $\{v_i\}_{i\in[n]}$ are still i.i.d. conditional on $Z$, there exist some absolute constants $c$ and $C$ such that for any $t>0$
\[\Pr\left( \left|\sumn K_{i\ell}v_i\right| > t \Bigg|Z \right)\leq C\exp\left(-\dfrac{ct^2}{\sum_{i=1}^n K_{i\ell}^2}\right).\]
Note that $f(x)=\exp(-1/x)$ is concave when $x>1/2$. By Jensen's inequality, 
\[\begin{aligned}
    \E\left(\exp\left(-\dfrac{ct^2}{\sum_{i=1}^n K_{i\ell}^2}\right)\textbf{1}\left[\sum_{i=1}^n K_{i\ell}^2>ct^2/2\right]\right) &= \E\left(\exp\left(-\dfrac{ct^2}{\sum_{i=1}^n K_{i\ell}^2\textbf{1}\left[\sum_{i=1}^n K_{i\ell}^2>ct^2/2\right]}\right)\right) \\
    &\leq \exp\left(-\dfrac{ct^2}{\E\left(\sum_{i=1}^n K_{i\ell}^2\textbf{1}\left[\sum_{i=1}^n K_{i\ell}^2>ct^2/2\right]\right)}\right) \\ 
    &\lesssim C\exp\left(-\dfrac{Mct^2}{n}\right).
\end{aligned}\]
Then 
\[\begin{aligned}
    \Pr\left( \left|\sumn K_{i\ell}v_i\right| > t \right) &\leq \mathbb{E}\left[\Pr\left( \left|\sumn K_{i\ell}v_i\right| > t \Bigg|Z \right)\textbf{1}\left[\sum_{i=1}^n K_{i\ell}^2>ct^2/2\right]\right] + \Pr\left(\sum_{i=1}^n K_{i\ell}^2>ct^2/2\right) \\
    &\leq C\exp\left(-\dfrac{ct^2}{\sum_{i=1}^n \mathbb{E}(K_{i\ell}^2)}\right) + C(pM)^{-4} \leq C\exp\left(-\dfrac{Mct^2}{n}\right) + C(pM)^{-4}.
\end{aligned}\]
Let $t=\sqrt{3n\log (pM)/(Mc)}$. Then by union bound 
\[\Pr\left( \max_{\ell\in[p_zM]}n^{-1}\left|\sumn K_{i\ell}v_i\right| > \sqrt{\dfrac{3}{c}\cdot\dfrac{\log (pM)}{nM}} \right)  \leq C(pM)^{-4}.\]
Then with probability at least $1-C(pM)^{-4}$, 
\begin{equation}\label{eq:eW phi} 
\|K^{\top}v\|_2^2 \lesssim M \|K^\top v\|_\infty^2 \lesssim \sqrt{\log (pM)}. 
\end{equation} 
By the inequality $\frac{1}{n}K^\top v(\kappa-\widehat{\kappa})\leq \frac{1}{n}\|K^\top v\|_2\|\kappa -\widehat{\kappa}\|_2$, we further obtain 
\begin{equation}
\PP\left(\left|\frac{1}{n}K^\top v( \kappa -\widehat{\kappa })\right|\geq  R_{D2} \|\kappa-\widehat{\kappa}\|_2\right)\leq \PP\left( \frac{1}{n}\|K^\top v\|_2\geq R_{D2} \right)\lesssim (pM)^{-4},
\label{eq: basic 3 phi}
\end{equation}
where \begin{equation}
    R_{D2} \asymp \sqrt{\dfrac{\log (pM)}{n}}. 
    \label{eq: R2 bound phi}
\end{equation}

By plugging the upper bounds \eqref{eq: basic 1 phi}, \eqref{eq: basic 2 phi} and   \eqref{eq: basic 3 phi}   into the basic inequality \eqref{eq: basic inequality phi}, we conclude that with probability at least $1-2(pM)^{-4}$,
\begin{equation}
\begin{aligned}
&\frac{1}{2n}\|K(\kappa - \hat\kappa)+X({\varphi}-\widehat{\varphi})\|_n^2+(1-c_0){\lambda_D}\|(\widehat{\varphi}-{\varphi})_{\mathcal{S}_\varphi^c}\|_1 \\ \leq&
2R_{D1}^2+2R_{D2}\|\widehat{\kappa}-\kappa\|_2+(1+c_0)\lambda_D\|(\widehat{\varphi} - \varphi)_{\mathcal{S}_\varphi}\|_1,  
\label{eq: key basic inequality phi}
\end{aligned}
\end{equation}
with $\mathcal{S}_\varphi=\{j\in[p]:\varphi_j\neq 0\}$.

\noindent {\bf Step 2: Establish ``Restricted Eigenvalue" type concentration.} In the following, we establish concentration bounds for $\frac{1}{2n}\|K(\kappa-\widehat{\kappa})+X({\varphi}-\widehat{\varphi})\|_2^2$. 
We first consider the case
\begin{equation}
2R_{D2}\|\kappa-\widehat{\kappa}\|_2+(1+c_0)\lambda_D\|({\varphi}-\widehat{\varphi})_{\mathcal{S}_\varphi}\|_1  \leq c_1 R_{D1}^2.
\label{eq: special case 1 phi}
\end{equation} for some positive constant $c_1>0$.  Then by \eqref{eq: key basic inequality phi} and \eqref{eq: special case 1 phi}
\begin{equation}\label{eq: varphi hat L1 special case 1}
    \|\varphi - \hat\varphi\|_1 = \|(\varphi - \hat\varphi)_{\mathcal{S}_\varphi}\|_1 + \|(\varphi - \hat\varphi)_{\mathcal{S}_\varphi^c}\|_1 \leq \left(\dfrac{2+c_1}{(1-c_0)} + \dfrac{1}{1-c_0} \right) \dfrac{R_{D1}^2}{\lambda_D}. 
\end{equation}
Besides, \eqref{eq: special case 1 phi} also implies that 
\begin{equation}\label{eq: kappa hat L2 special case 1}
    \|\kappa - \hat\kappa\|_2 \lesssim \dfrac{c_1R_{D1}^2}{R_{D2}}.
\end{equation} 
We further provide the following lemma about the $L_2$ norm of the estimation error of $\hat\varphi$. 
\begin{Lemma}\label{lem: L2 varphi hat}Under the conditions for Proposition \ref{prop: Lasso D}, when (\ref{eq: special case 1 phi}) holds, we have with probability at least $1 - c(pM)^{-4}$ such that
\begin{equation}\label{eq: varphi hat L2 special case 1}
    \|\hat\varphi-\varphi\|_2^2 \lesssim R_{D1}^2.
\end{equation}
\end{Lemma}
% \begin{equation}
% (1-c_0){\lambda_D}\|(\widehat{\varphi}-{\varphi})_{\mathcal{S}_\varphi^c}\|_1\leq (2+c_1) R_{D1}^2
% \end{equation}
% Together with \eqref{eq: special case 1 phi}, we have
% \begin{equation}
% \|\widehat{\varphi}-{\varphi}\|_1=\|(\widehat{\varphi}-{\varphi})_{\mathcal{S}_\varphi}\|_1+\|(\widehat{\varphi}-{\varphi})_{\mathcal{S}_\varphi^c}\|_1\leq \frac{R_{D1}^2}{\lambda_D}\left(\frac{2+c_1 }{1-c_0}+\frac{c_1}{1+c_0}\right)
% \label{eq: upper for alpha 0 phi}
% \end{equation}
% and 
% \begin{equation}
% \|\widehat{\kappa}-\kappa\|_2\leq \frac{c_1 R_{D1}^2}{R_{D2}}\lesssim \frac{\sqrt{n}}{a_n}M^{-2\gamma_z}
% \label{eq: upper for omega 0 phi}
% \end{equation}
When \eqref{eq: special case 1 phi} does not hold, then 
\begin{equation}
2R_{D2}\|\widehat{\kappa}-\kappa\|_2+(1+c_0)\lambda_D\|(\widehat{\varphi}-{\varphi})_{\mathcal{S}_\varphi}\|_1  > c_1 R_{D1}^2.
\label{eq: special case 2 phi}
\end{equation}
And \eqref{eq: key basic inequality phi} implies  
\begin{equation}
\begin{aligned}
&\frac{1}{2n}\|K(\kappa-\widehat{\kappa})+X({\varphi}-\widehat{\varphi})\|_n^2+(1-c_0){\lambda_D}\|({\varphi}-\widehat{\varphi})_{\mathcal{S}_\varphi^c}\|_1\\
&\leq  \frac{4+2c_1}{c_1}R_{D2}\|\hat\kappa-\kappa\|_2+\left(1+c_0+\dfrac{2}{c_1}\right)\lambda_D\|(\widehat{\varphi}-{\varphi})_{\mathcal{S}_\varphi}\|_1.   
\end{aligned}
\label{eq: key basic inequality 0 phi}
\end{equation}
In this case, we consider the restricted parameter space, 
\begin{equation}
\label{eq: Set C0 phi}
\mathcal{C}_D=\left\{\delta= \left(\begin{array}{c}
     \widehat{\varphi}-\varphi \\
    \widehat{\kappa}-\kappa 
\end{array}\right): \|(\widehat{\varphi}-{\varphi})_{\mathcal{S}_\varphi^c}\|_1\leq \frac{c_2R_{D2}}{\lambda_D}\|\widehat{\kappa}-\kappa\|_2 +    c_3\|(\widehat{\varphi}-{\varphi})_{\mathcal{S}_\varphi}\|_1\right\}
\end{equation}
where   $c_2=\frac{4+2c_1}{(1-c_0)c_1}$ and  $c_3=\frac{(c_1+c_0c_1+2)}{(1-c_0)c_1}$.

 \begin{Lemma} Suppose the conditions in Theorem \ref{thm: estimation bound} hold. Then w.p.a.1 
\begin{equation}
\sup_{\widehat{\varphi}-\varphi\in \mathcal{C}_D} \frac{\frac{1}{n}\|K(\kappa-\widehat{\kappa})+X({\varphi}-\widehat{\varphi})\|_2^2}{M^{-1}\|\kappa-\widehat{\kappa}\|_2^2+\|{\varphi}-\widehat{\varphi}\|_2^2} \geq 2c,
\label{eq: RE concentration phi}
\end{equation}
 \label{lem: RE concentration phi}
for some universal positive constant $c$. 
 \end{Lemma}

 By Lemma \ref{lem: RE concentration phi}, we have with probability at least $1-c(pM)^{-4}$
 \begin{equation}\label{eq: lower bound K X phi}
     \frac{1}{2n}\|K(\kappa-\widehat{\kappa})+X({\varphi}-\widehat{\varphi})\|_n^2 \geq {cM^{-1}\|\kappa-\widehat{\kappa}\|_2^2+c\|{\varphi}-\widehat{\varphi}\|_2^2}. 
 \end{equation}
 By adding both sides of \eqref{eq: key basic inequality phi} with $(1-c_0)\lambda_D\|(\widehat{\varphi}-{\varphi})_{\mathcal{S}_\varphi}\|_1$, we have 
\begin{equation}
\begin{aligned}
\ &\dfrac{1}{2n}\|K(\kappa-\hat\kappa) + X(\varphi - \hat \varphi)\|_2^2+2(1-c_0)\lambda_D\|\widehat{\varphi}-{\varphi}\|_1  \\ 
\leq  &2R_{D1}^2 + 2R_{D2}\|\widehat{\kappa}-\kappa\|_2 + 2\lambda_D\|(\widehat{\varphi}-{\varphi})_{\mathcal{S}_\varphi}\|_1.
\end{aligned}
\label{eq: basic inequality 2 phi}
\end{equation}
Since $\frac{1}{2c}a^2+\frac{c }{2}b^2\geq ab$ for any $a,b>0$,  we have 
\begin{equation}
2R_{D2}\|\widehat{\kappa}-\kappa\|_2\leq \frac{2M}{c}R_{D2}^2+\frac{c}{2M}\|\widehat{\kappa}-\kappa\|_2^2,
\label{eq: decouple 1 phi}
\end{equation}
and 
\begin{equation}
\begin{aligned}
2\lambda_D\|(\widehat{\varphi}-{\varphi})_{\mathcal{S}_\varphi}\|_1&\leq 2\sqrt{s}\lambda_D \|(\widehat{\varphi}-{\varphi})_{\mathcal{S}_\varphi}\|_2 \\ 
&\leq  \frac{8}{c} {s}\lambda_D^2+\frac{c}{2} \|\widehat{\varphi}-{\varphi}\|_2^2. 
\label{eq: decouple 2 phi}
\end{aligned}
\end{equation}
By plugging \eqref{eq: decouple 1 phi} and \eqref{eq: decouple 2 phi} into \eqref{eq: basic inequality 2 phi}, we have 
\begin{equation}
 \dfrac{1}{4n}\|K(\kappa-\hat\kappa) + X(\varphi - \hat\varphi)\|_2^2+(1-c_0)\lambda_D\|\widehat{\varphi}-{\varphi}\|_1 - \dfrac{c}{2M}\|\hat\kappa-\kappa\|_2^2 - \dfrac{c}{2}\|\hat\varphi - \varphi\|_2^2 \lesssim R_{D1}^2 + MR_{D2}^2 + s\lambda_D^2.
\label{eq: basic inequality 3 phi}
\end{equation}
By \eqref{eq: lower bound K X phi} and \eqref{eq: basic inequality 3 phi}, 
\begin{equation}\label{eq: Lasso D bound final}
    \dfrac{c}{2M}\|\kappa - \hat\kappa\|_2^2 + c\|\hat\varphi - \varphi\|_2^2 + (1-c_0)\lambda_D\|\hat\varphi-\varphi\|_1 \lesssim R_{D1}^2 + M R_{D2}^2 + s\lambda_D^2.
\end{equation} 
Note that the rates of $R_{D1}$, $R_{D2}$ and $\lambda_D$ are specified in (\ref{eq: R1 bound phi}), (\ref{eq: R2 bound phi}) and the statement of Proposition \ref{prop: Lasso D}, respectively. Then (\ref{eq: kappa hat error L2 prop p}) follows (\ref{eq: kappa hat L2 special case 1}) and (\ref{eq: Lasso D bound final}); (\ref{eq: hat varphi error L1}) and (\ref{eq: hat varphi error L2}) follow (\ref{eq: special case 1 phi}), Lemma \ref{lem: L2 varphi hat} and (\ref{eq: Lasso D bound final}).
\end{proof}
% \par \underline{We then prove (b).} Since $v_i(K_\ell)_{ij}$ is bounded for any $i,j,\ell$ and hence has a bounded sub-Gaussian norm, 
% \[\Pr\left( \max_{j\in[M_z],\ell\in[p_z]}\left|\sumn v_i (K_\ell)_{ij}\right| \geq C\sqrt{n\log (pM)} \right) \leq (pM)^{-3}\]
% for some $C>0$. Then by union bound, we deduce that  
% \begin{equation}\label{eq:eW phi} 
% \PP\left(\| K^{\top}v\|_2^2\geq C^2 p_zMn\log(pM_z)  \right) \leq p_zM \cdot (pM)^{-3}.
% \end{equation} 
% By the inequality $\frac{1}{n}K^\top v(\kappa-\widehat{\kappa})\leq \frac{1}{n}\|K^\top v\|_2\|\kappa -\widehat{\kappa}\|_2$, we further obtain 
% \begin{equation}
% \PP\left(\left|\frac{1}{n}K^\top v( \kappa -\widehat{\kappa })\right|\geq  R_{D2} \|\kappa-\widehat{\kappa}\|_2\right)\leq \PP\left( \frac{1}{n}\|K^\top v\|_2\geq R_{D2} \right)\lesssim (pM)^{-4}
% \label{eq: basic 3 phi}
% \end{equation}
% where \begin{equation}
%     R_{D3}=C\sqrt{\dfrac{p_z M\log (pM)}{n}} \asymp \dfrac{\sqrt{M\log (pM)}}{\sqrt{n}}. 
%     \label{eq: R3 bound phi}
% \end{equation} 
%  Following the same arguments in the proof of (a), we know that \eqref{eq: vhat Linf 3} holds with $R_{D2}$ replaced by $R_{D3}$. Then (\ref{eq: vhat error}) follows. 

\begin{Corollary}
\label{cor: v hat}Recall that $\hat v_i = D_i - K_\ic^\top\hat\kappa - X_\ic^\top\hat\varphi$. Suppose that Assumptions \ref{assu: control func}-\ref{assu: eigen} hold.
 There exist some constants $C$ and $c$ such that with probability at least $1-c(pM)^{-4}$
    \begin{equation}\label{eq: vhat error L2 prop p}
    n^{-1} \sumn \left(\hat v_i - v_i\right)^2 \leq C\left(  \dfrac{(s+M)\log (pM)}{n}+ M^{-2\gamma}\right),
\end{equation}
and
\begin{equation}\label{eq: vhat error}
\begin{aligned}
     \supn\left|\hat v_i - v_i \right| &\leq C\left(\sqrt{\dfrac{n}{\log (pM)}}M^{-2\gamma} + (s+M)\sqrt{\dfrac{\log (pM)}{n}}+ M^{-\gamma}\right).
\end{aligned}
\end{equation}
In addition, when $\hat\kappa$ and $\varphi$ are independent of $D_i,K_\ic,X_\ic$ and $v_i$, 
\begin{equation}\label{eq: conditional Exp vhat}
    \E_{\mathcal{L}}\left(\hat v_i - v_i\right)^2 \leq C\left(   \dfrac{(s+M)\log (pM)}{n}+ M^{-2\gamma}\right)
\end{equation}
with probability at least $1-c(pM)^{-4}$. 
\end{Corollary}

\begin{proof}[Proof of Corollary \ref{cor: v hat}] Note that
\begin{equation}\label{eq: corollary proof vhat error very beginning}
    \begin{aligned}
    n^{-1}\sumn(\hat v_i - v_i)^2 &\lesssim n^{-1}\sumn\left(K_\ic^\top (\hat\kappa-\kappa) + X_\ic^\top(\hat\varphi - \varphi)\right)^2 + n^{-1}\sumn   r_{\psi i}^2\\
    &\lesssim n^{-1}\sumn\left(K_\ic^\top (\hat\kappa-\kappa) + X_\ic^\top(\hat\varphi - \varphi)\right)^2 +M^{-2\gamma},\\
\end{aligned} 
\end{equation} 
and 
\begin{equation}\label{eq: proof vhat L2 begin}
    \begin{aligned}
 &\ \ \ \ n^{-1}\sumn\left(K_\ic^\top (\hat\kappa-\kappa) + X_\ic^\top(\hat\varphi - \varphi)\right)^2 \\
    &\lesssim \|\hat\kappa-\kappa\|_2^2\cdot \left\|\dfrac{\sumn H_\ic H_\ic^\top}{n}-\E\left(H_\ic H_\ic^\top\right)\right\|_2 + \|\hat\varphi-\varphi\|_1^2\left\|\dfrac{\sumn X_\ic X_\ic^\top}{n} - \E(X_\ic X_\ic^\top)\right\|_\infty \\
    &\ \ \ \ + \|\hat\kappa-\kappa\|_2^2\cdot \left\| \E\left(H_\ic H_\ic^\top\right)\right\|_2 + \|\hat\varphi - \varphi\|_2^2 \cdot \left\| \E(X_\ic X_\ic^\top)\right\|_2. 
\end{aligned}
\end{equation}

Following similar arguments for (\ref{eq: BB m EBB}), we can show that  \[\left\|\dfrac{\sumn H_\ic H_\ic^\top}{n}-\E\left(H_\ic H_\ic^\top\right)\right\|_2 \lesssim \sqrt{(nM)^{-1}\log (pM)}\]
with probability at least $1-C(pM)^{-4}$. 
Besides, note that $\|\E(H_\ic H_\ic^\top)\|_2\lesssim M^{-1}$ and $\|\E(X_\ic X_\ic^\top)\|_\infty \lesssim 1$. Together with Propositions \ref{prop: DB} and \ref{prop: Lasso D}, with probability at least $1-c(pM)^{-4}$, 
\begin{equation}\label{eq: Lasso D estimation error proof corollary}
    \begin{aligned}
    &\ \ \ \ n^{-1}\sumn\left(K_\ic^\top (\hat\kappa-\kappa) + X_\ic^\top(\hat\varphi - \varphi)\right)^2 \\
    &\lesssim  \|\hat\kappa - \kappa\|_2^2\cdot\left(\sqrt{\dfrac{\log (pM)}{nM}} + \dfrac{1}{M}\right) + \|\hat\varphi - \varphi\|_1^2 \cdot \sqrt{\dfrac{\log (pM)}{n}} + \|\hat\varphi - \varphi\|_2^2   \\
    &\lesssim \|\hat\kappa - \kappa\|_2^2\cdot\dfrac{1}{M} + \|\hat\varphi - \varphi\|_1^2 \cdot \sqrt{\dfrac{\log (pM)}{n}}  + \|\hat\varphi - \varphi\|_2^2    \\
    & \lesssim \dfrac{n}{\log (pM)}M^{-4\gamma - 1} + \dfrac{(s+M)\log (pM)}{n} \\ &\lesssim   M^{-2\gamma}  + \dfrac{(s+M)\log (pM)}{n},
\end{aligned}
\end{equation} 
 where the last step applies $n = o[M^{2\gamma + 1}]$. Thus by (\ref{eq: corollary proof vhat error very beginning}),
 \[n^{-1}\sumn (\hat v_i - v_i)^2 \lesssim M^{-2\gamma} + \dfrac{(s+M)\log (pM)}{n}.\]
Besides, 
\[\begin{aligned}
    \supn |\hat v_i - v_i| 
&\lesssim \|\hat\kappa-\kappa\|_2\cdot\supn\|K_\ic\|_2 + \|\hat\varphi-\varphi\|_1\cdot \supn\|X_\ic\|_\infty + \supn|r_{\psi i}|\\
&\lesssim \sqrt{\dfrac{n}{\log (pM)}}M^{-2\gamma} + (s+M)\sqrt{\dfrac{\log (pM)}{n}} + M^{-\gamma}.
\end{aligned}\]
Finally,  
\[\begin{aligned}
    \E_{\mathcal{L}}(\hat v_i - v_i)^2 &\lesssim \|\hat\kappa - \kappa\|_2^2\cdot\|\E(H_\ic H_\ic^\top)\|_2 + \|\hat\varphi - \varphi\|_2^2\cdot\|\E(X_\ic X_\ic^\top)\|_2 + \E[r_{\psi i}^2] \\
    &\lesssim \|\hat\kappa - \kappa\|_2^2\cdot M^{-1}+ \|\hat\varphi - \varphi\|_2^2  + M^{-2\gamma}, 
\end{aligned}\]
where the last inequality applies Proposition \ref{prop: spline eigen} and the bounded eigenvalues of $\E(X_\ic X_\ic^\top)$. Then (\ref{eq: conditional Exp vhat}) holds by 
Proposition \ref{prop: Lasso D}. This completes the proof of Corollary \ref{cor: v hat}. 
\end{proof}

\begin{Remark}\label{rem: v hat}
Define \begin{equation} \label{eq: hat v event}
    \hat{\mathcal{V}}_n :=  \left\{\supn\left|\hat v_i - v_i \right| \leq C\left(\sqrt{\dfrac{n}{\log (pM)}}M^{-2\gamma} + (s+M)\dfrac{\log (pM)}{n}+ M^{-\gamma}\right)\right\}.
\end{equation}  
When $M \asymp n^\nu$ with $\nu\in [\frac{1}{2\gamma+1},\frac{1}{2\gamma})$ and $s^2\log p = o(n)$, (\ref{eq: vhat error}) implies $\sup_{i\in[n]}|\hat v_i - v_i| = o(1)$ with probability at least $1-c(pM)^{-4}$. Consequently, $\hat{\mathcal{V}}_n$ implies  $\{\hat v_i\in [a_v-\epsilon_v,b_v+\epsilon_v]\text{ for all }i\in[n]\}$. By Corollary \ref{cor: v hat},
\begin{equation}\label{eq: prob vhat compact}
    \Pr\left[\hat{\mathcal{V}}_n\right] \geq 1-c(pM)^{-4}.
\end{equation}
By Assumption \ref{assu: function}, it is easy to deduce that, when $\hat{\mathcal{V}}_n$ holds,
\[|q(v_i) - q(\hat v_i)| \leq \sup_v|q^\prime(v)|\cdot|v_i - \hat v_i| \lesssim |v_i - \hat v_i|,\]
and likewise, 
\[|q^\prime(v_i) - q^\prime(\hat v_i)| \leq \sup_v|q^{\prime\prime}(v)|\cdot|v_i - \hat v_i| \lesssim |v_i - \hat v_i|.\]
Thus, when $\hat{\mathcal{V}}_n$ holds, 
\begin{equation}\label{eq: sum sq q approx}
 n^{-1}\sumn\left(q(v_i) - q(\hat v_i)\right)^2   + n^{-1}\sumn\left(q^\prime(v_i) - q^\prime(\hat v_i)\right)^2 \lesssim  \dfrac{(s+M)\log (pM)}{n}+ M^{-2\gamma}, 
\end{equation}
and 
\begin{equation}\label{eq: sup q approx}
 \supn\left|q(v_i) - q(\hat v_i)\right|   + \supn\left|q^\prime(v_i) - q^\prime(\hat v_i)\right| \lesssim \sqrt{\dfrac{n}{\log (pM)}}M^{-2\gamma} + (s+M)\sqrt{\dfrac{\log (pM)}{n}} + M^{-\gamma}.
\end{equation}
\end{Remark}
\bigskip 
\begin{Corollary}
\label{cor: H approx}Suppose that Assumptions \ref{assu: control func}-\ref{assu: eigen} hold.
 There exists some constant $c$ such that with probability at least $1-c(pM)^{-4}$
    \begin{equation}\label{eq: vhat error H L2 prop p}
    n^{-1} \sumn \sum_{j=1}^{M} (H_j(v_i)-H_j(\hat v_i))^2 \lesssim M^2\left(   \dfrac{(s+M)\log (pM)}{n}+ M^{-2\gamma}\right),
\end{equation}
and,
\begin{equation}\label{eq: vhat error H prime}
\begin{aligned}
    \supn \sqrt{\sum_{j=1}^M (H_j^\prime (v_i)-H_j^\prime (\hat v_i))^2} \lesssim M^2\left(\sqrt{\dfrac{n}{\log (pM)}}M^{-2\gamma} + \sqrt{\dfrac{(s+M)^2\log (pM)}{n}} + M^{-\gamma}\right). 
\end{aligned}
\end{equation}
Furthermore, when $\hat\kappa$ and $\hat\varphi$ are independent of $D_i,K_\ic,X_\ic$ and $v_i$,  with probability at least $1-c(pM)^{-4}$
\begin{equation}\label{eq: cond exp vhat H}
   \E_{\mathcal{L}}\sum_{j=1}^{M} (H_j(v_i)-H_j(\hat v_i))^2 \lesssim M^2\left(  (s+M)\dfrac{\log (pM)}{n}+ M^{-2\gamma}\right) + \epsilon_n, 
\end{equation}
for some $\epsilon_n = O_p[(pM)^{-2}]$.
\end{Corollary}

\begin{proof}[Proof of Corollary \ref{cor: H approx}]
By the mean value theorem, we know that for any $i\in[n]$ and $j\in[M]$, there exists a $v_{ij}^*$ between $v_i$ and $\hatv_i$ such that 
\[ |H_j(v_i) - H_j(\hatv_i)| = |H^\prime(v_i^*)|\cdot \left|v_i-\hatv_i\right|.\]
As mentioned in Remark \ref{rem: v hat}, $\hat{\mathcal{V}}_n$ implies  $\hat v_i\in [a_v-\epsilon_v,b_v+\epsilon_v]\text{ for all }i\in[n]$ and hence $\{v_i^*\}_{i\in[n]}$ also fall in this compact interval. 
Consequently, when $\hat{\mathcal{V}}_n$ in (\ref{eq: hat v event}) holds
\begin{equation}\label{eq: bound H m H hat}\begin{aligned}
    \sum_{j=1}^{M}\left(H_j(v_i) - H_j(\hatv_i)\right)^2 &= \sum_{j=1}^{M}\left(v_i-\hatv_i\right)^2\left(H_j^\prime(v_{i}^*)\right)^2 \lesssim M^2(v_i-\hat v_i)^2,\\
\end{aligned}
\end{equation}
where the second inequality applies Proposition \ref{prop: SplineDerivativeL2Norm}. Thus, $\sum_{i=1}^n\sum_{j=1}^{M}\left(H_j(v_i) - H_j(\hatv_i)\right)^2 \lesssim M^2\sum_{i=1}^n\left(v_i-\hatv_i\right)^2$. Similarly, using Proposition \ref{prop: SplineDerivativeL2Norm} for second-order derivatives of splines, 
\[\begin{aligned}
\supn\sum_{j=1}^{M}\left(H_j^\prime(v_i) - H_j^\prime(\hatv_i)\right)^2 &= \supn\sum_{j=1}^{M}\left(v_i-\hatv_i\right)^2\left(H_j^{\prime\prime}(v_{i}^*)\right)^2 \lesssim M^4\left(v_i-\hatv_i\right)^2. 
\end{aligned}\]
Then (\ref{eq: vhat error H L2 prop p}) and (\ref{eq: vhat error H prime}) hold by (\ref{eq: vhat error L2 prop p}) and (\ref{eq: vhat error}) in Corollary \ref{cor: v hat}. As for (\ref{eq: cond exp vhat H}), 
\[\begin{aligned}
    \E_{\mathcal{L}}\sum_{j=1}^{M} (H_j(v_i)-H_j(\hat v_i))^2 &=  \E_{\mathcal{L}}\left[1(\hat{\mathcal{V}}_n)\sum_{j=1}^{M} (H_j(v_i)-H_j(\hat v_i))^2\right] + \E_{\mathcal{L}}\left[1(\hat{\mathcal{V}}_n^c)\sum_{j=1}^{M} (H_j(v_i)-H_j(\hat v_i))^2\right] \\ 
    &\lesssim  M^2 \E_{\mathcal{L}}(\hat v_i - v_i)^2 + M{\rm Pr}(\hat{\mathcal{V}}^c|\mathcal{L}),
\end{aligned}\]
where the first term on the RHS applies the (\ref{eq: bound H m H hat}) when $\hat{\mathcal{V}}_n$ holds, and the second term applies the boundness of $H_j$. It suffices to show 
\begin{equation}\label{eq: bound cond prob}
    {\rm Pr}(\hat{\mathcal{V}}^c|\mathcal{L}) = O_p((pM)^{-4}). 
\end{equation}
By Markov inequality, for any $C>0$,
\[\begin{aligned}
    \Pr\left[ {\rm Pr}(\hat{\mathcal{V}}^c|\mathcal{L}) > C(pM)^{-4} \right] \leq \dfrac{\E\left[{\rm Pr}(\hat{\mathcal{V}}^c|\mathcal{L})\right]}{C(pM)^{-4} } = \dfrac{{\rm Pr}(\hat{\mathcal{V}}^c)}{C(pM)^{-4}} = \dfrac{O(1)}{C},
\end{aligned}\]
and hence $\lim_{C\to\infty}\Pr\left[ {\rm Pr}(\hat{\mathcal{V}}^c|\mathcal{L}) > C(pM)^{-4} \right] = 0$. Then 
\[\begin{aligned}
    \E_{\mathcal{L}}\sum_{j=1}^{M} (H_j(v_i)-H_j(\hat v_i))^2  
    &\lesssim  M^2 \E_{\mathcal{L}}(\hat v_i - v_i)^2 + (pM)^{-2},
\end{aligned}\]
and hence (\ref{eq: cond exp vhat H}) holds by (\ref{eq: conditional Exp vhat}).
\end{proof}

\begin{Corollary}\label{cor: Spline2norm}
    For all $1\leq i\leq n$ we have with probability at least $1-c(pM)^{-4}$, 
    \[ \dfrac{1}{\sqrt{k}} \leq \|\hat H_\ic\|_2 \leq 1.\]
\end{Corollary}
\begin{proof}[Proof of Corollary \ref{cor: Spline2norm}]
The result holds by the proof of Proposition \ref{prop: Spline2norm} when $\hat{\mathcal{V}}_n$ in (\ref{eq: hat v event}) holds. 
\end{proof}

%%%%%%%%%%%%%%%%%%%%%%%%%%%%%%%%%%%%%%%%%%%%%%%%%%%%%%%%%%%%%%%%%%%%%%%%%%%%%%%%
\subsection{Proof of Theorem \ref{thm: estimation bound}}
%%%%%%%%%%%%%%%%%%%%%%%%%%%%%%%%%%%%%%%%%%%%%%%%%%%%%%%%%%%%%%%%%%%%%%%%%%%%%%%%

%Define $\widehat{\theta}=s\widehat{\theta}+(1-s)\theta$ where $\widehat{\theta}=\left(\widehat{\theta}^{\top},\widehat{\theta}^{\top}\right)^{\top}$ and $\widehat{\theta}=s\widehat{\theta}+(1-s)\theta$ and $\widehat{\theta}=s\widehat{\theta}+(1-s){\theta}$

%Define $\widehat{\theta}=s\widehat{\theta}+(1-s)\theta$ where $\widehat{\theta}=\left(\widehat{\theta}^{\top},\widehat{\theta}^{\top}\right)^{\top}$ and $\widehat{\theta}=s\widehat{\theta}+(1-s)\theta$ and $\widehat{\theta}=s\widehat{\theta}+(1-s){\theta}$. We define 
%\begin{equation}
%s= \frac{a(n)}{\|\widehat{\theta}-\theta\|_2+a(n)} \quad \text{for} \quad a(n)=C^{*}\sqrt{n^{-\frac{2\gamma}{2\gamma+1}}+\frac{k \log p}{n}}.
%\end{equation}
%Note that 
%\begin{equation}
%\|\widehat{\theta}-\theta\|_2=s \|\widehat{\theta}-\theta\|_2 \leq a(n) \quad \text{and} \quad \|\widehat{\theta}-{\theta}\|_1=s \|\widehat{\theta}-{\theta}\|_1.
%\end{equation}
%By triangle inequality, we have
As a preparatory result for honesty of the confidence band, we will show that the high probability error bounds hold uniformly for all $g\in\mathcal{H}_{\mathcal{D}}(\gamma,L)$. We abbreviate $\inf_{g\in\mathcal{H}_{\mathcal{D}}(\gamma,L)}$ and $\sup_{g\in\mathcal{H}_{\mathcal{D}}(\gamma,L)}$ as $\inf_g$ and $\sup_g$ in the remainder of the proofs. Specifically, we will show 
\begin{equation}
\sup_g \left(\|\widehat{\beta}-\beta\|_2^2 + \|\hat\eta - \eta\|_2^2\right) \lep M^{-2\gamma + 1} + \dfrac{(s+M)M\log (pM)}{n},
\label{eq: sup rate omega}
\end{equation}
\begin{equation}
\sup_g\|\widehat{\theta}-{\theta}\|_1 \lep  M^{-2\gamma}\sqrt{\dfrac{n}{\log p}} + (s+M)\sqrt{\dfrac{\logpM}{n}}. 
\label{eq: sup rate theta}
\end{equation}
Throughout the proof, only $\hat\omega$, $\hat\theta$ and the spline approximation error of function $g$, denoted by $r_g$, are dependent on the function $g$, and thus the uniformity over $g$ needs to be taken care of only for the events dependent on these estimators and approximation errors. 
\par Recall that $\omega=(\beta^\top,\eta^\top)^\top$ and $\hat\omega$ denotes the LASSO estimator. By the definition of LASSO estimators $\hatomega$ and $\widehat{\theta}$, we have
\begin{equation}
\frac{1}{n}\|y-\hatW\hatomega-X\widehat{\theta}\|_2^2+\lambda_Y\|\widehat{\theta}\|_1\leq \frac{1}{n}\|y-\hatW\omega-X{\theta}\|_2^2+\lambda_Y\|{\theta}\|_1,
\end{equation}
which implies 
\begin{equation}
\frac{1}{n}\|\hatW(\omega-\widehat{\omega})+X(\theta-\widehat{\theta})\|_2^2+\lambda_Y\|\widehat{\theta}\|_1\leq \frac{2}{n}\langle r+\varepsilon,\hatW(\omega-\widehat{\omega})+X({\theta}-\widehat{\theta})\rangle +\lambda_Y\|{\theta}\|_1,
\label{eq: basic inequality}
\end{equation}
given that $r+\varepsilon=y-\widehat{W}\omega-X\theta$ where $r=(r_1,r_2,\cdots,r_n)^\top$ with $r_i$ defined in (\ref{eq: def  ri}).

In the following, we use this basic inequality and further establish the convergence rate of the proposed estimator. The proof consists of two steps.\\
\noindent{\bf Step 1: Deduce a more convenient inequality.} We analyze the terms in \eqref{eq: basic inequality} and obtained a more convenient version of \eqref{eq: basic inequality}. Note that 
\begin{equation}
\frac{1}{n}r^{\top}\left(\hatW(\omega-\widehat{\omega})+X({\theta}-\widehat{\theta})\right) \leq R_1 \frac{1}{\sqrt{n}}\|\hatW(\omega-\widehat{\omega})+X({\theta}-\widehat{\theta})\|_2,
\end{equation}
where  \begin{equation}
\begin{aligned}
 R_1=\sqrt{ \supg \frac{1}{n}\sum_{i=1}^{n}r_i^2} &= 
 \sqrt{\dfrac{1}{n}\sum_{i=1}^n\left( \supg r_g(D_i) + r_q(\hatv_i)  + q(v_i) - q(\hatv_i) \right)^2} \\
 & \lesssim \sqrt{\dfrac{1}{n}\sum_{i=1}^n (\supg r_B^2(D_i) + r_H^2(\hatv_i))  + \dfrac{1}{n}\sum_{i=1}^n [q(v_i) - q(\hatv_i)]^2 } \\
 &\lep M^{-\gamma} + \sqrt{\dfrac{(s+M)\logpM}{n}}, 
\end{aligned}
\label{eq: R1 bound}
\end{equation} 
 with probability at least $1-c(pM)^{-4}$, where the last inequality applies Assumption \ref{assu: function},    Proposition \ref{prop: spline approx power}, and Remark \ref{rem: v hat}.  By the inequality $ab\leq a^2+ \frac{b^2}{4}$, we have 
\begin{equation}
\begin{aligned}
\frac{1}{n}r^{\top}\left(\hatW(\omega-\widehat{\omega})+X({\theta}-\widehat{\theta})\right)\leq R_1^2+ \frac{1}{4{n}}\|\hatW(\omega-\widehat{\omega})+X({\theta}-\widehat{\theta})\|_n^2.
\end{aligned}
\label{eq: basic 1}
\end{equation}
Proposition \ref{prop: DB} implies that  $\left\|\frac{1}{n}X^{\top}\varepsilon\right\|_{\infty}\leq \dfrac{c_0\lambda_Y}{2} $ for $0<c_0<1$, and hence we have w.p.a.1 
\begin{equation}
\frac{1}{n}X^{\top}\varepsilon({\theta}-\widehat{\theta})\leq \dfrac{c_0\lambda_Y}{2} \|{\theta}-\widehat{\theta}\|_1
\label{eq: basic 2}
\end{equation}
for any $\hat\theta$. Furthermore,  
\[\begin{aligned}
\begin{aligned}
\E\|B^{\top}\varepsilon\|_2^2=\E\sum_{j=1}^{M}\left(\sum_{i=1}^{n} \varepsilon_i B_{ij}\right)^2&=\sum_{j=1}^{M}\sum_{i=1}^{n} \E (\varepsilon_i^2 B_{ij}^2)\\
&=\sum_{j=1}^{M}\sum_{i=1}^{n} \E (\E(\varepsilon_i^2 B_{ij}^2|X,Z,v)) \\
&=\sum_{j=1}^{M}\sum_{i=1}^{n} \E (B_{ij}^2\E(\varepsilon_i^2 |X,Z,v)) \\
&= \sigma_\varepsilon^2 \sum_{i=1}^{n} \E (\sum_{j=1}^{M}B_{ij}^2) \leq nk\sigma_\varepsilon^2,  
\end{aligned}
\end{aligned}\]
where the last inequality applies Proposition \ref{prop: Spline2norm}. In a similar manner we deduce that 
\[\begin{aligned}
\begin{aligned}
\E\left[\|\widehat{H}^{\top}\varepsilon\|_2^2\right] \leq n k \sigma_\varepsilon^2.
\end{aligned}
\end{aligned}\]
Note that $\|\hatW^{\top}\varepsilon\|_2^2 \leq 2\|B^{\top}\varepsilon\|_2^2 + 2\|\widehat{H}^{\top}\varepsilon\|_2^2$. Then by Markov inequality we deduce that 
\begin{equation}\label{eq:eW} 
\begin{aligned}
\Pr\left(\| \hatW^\top \varepsilon\|_2^2\geq a_n^2 p_z n k \sigma_\varepsilon^2  \right) &\leq \frac{2 \E[\| \hatW^\top \varepsilon\|_2^2]}{ \log (pM)\cdot p_z n k \sigma_\varepsilon^2  } \leq \frac{2}{\log (pM)}
\end{aligned}
\end{equation}
and thus 
\begin{equation}
\Pr\left( \frac{1}{n}\|\hatW^{\top}\varepsilon\|_2\geq R_2 \right)\leq \frac{2}{\log (pM)}\to 0%=n^{-\frac{\gamma}{2\gamma+1}}\|\theta-\widehat{\theta}\|_2
\label{eq: basic 3}
\end{equation}
where \begin{equation}
    R_2=a_n\sigma_v\sqrt{\dfrac{2k}{n}}\lesssim \sqrt{\dfrac{\log (pM)}{n}}. 
    \label{eq: R2 bound}
\end{equation} 
By (\ref{eq: basic 3}) and the inequality $\frac{1}{n}\hatW^{\top}\varepsilon(\omega-\widehat{\omega})\leq \frac{1}{n}\|\hatW^{\top}\varepsilon\|_2\|\omega-\widehat{\omega}\|_2$, we further obtain w.p.a.1 
\begin{equation}
\left|\frac{1}{n}\hatW^{\top}\varepsilon(\omega-\widehat{\omega})\right| \lesssim R_2 
\label{eq: basic 3-1}
\end{equation}
for any $\hat\omega$. By plugging the upper bounds \eqref{eq: basic 1}, \eqref{eq: basic 2} and   \eqref{eq: basic 3-1}   into the basic inequality \eqref{eq: basic inequality}, we conclude that w.p.a.1, for any $\hat\theta$ and $\hat\omega$ 
\begin{equation}
\begin{aligned}
&\ \ \ \ \frac{1}{2n}\|\hatW(\omega-\widehat{\omega})+X({\theta}-\widehat{\theta})\|_2^2+(1-c_0){\lambda_Y}\|(\widehat{\theta}-{\theta})_{\mathcal{S}^c}\|_1 \\
&\leq  2R_1^2+2R_2\|\widehat{\omega}-\omega\|_2+(1+c_0)\lambda_Y\|(\widehat{\theta}-{\theta})_{\mathcal{S}}\|_1,  
\label{eq: key basic inequality}
\end{aligned}
\end{equation}
with $\mathcal{S}=\{j\in[p]:\theta_j\neq 0\}$.

\noindent {\bf Step 2: Establish restricted eigenvalue-type concentration.} In the following, we establish concentration bounds for $\frac{1}{2n}\|W(\omega-\widehat{\omega})+X({\theta}-\widehat{\theta})\|_2^2$. 
We first consider the case
\begin{equation}
2R_2\|\widehat{\omega}-\omega\|_2+(1+c_0)\lambda_Y\|(\widehat{\theta}-{\theta})_{\mathcal{S}}\|_1  \leq c_1 R_1^2.
\label{eq: special case 1}
\end{equation} for some positive constant $c_1>0$. It follows from \eqref{eq: key basic inequality} that 
\begin{equation}\label{eq: error L2 pl sparse L1}
\frac{1}{2n}\|W(\omega-\widehat{\omega})+X({\theta}-\widehat{\theta})\|_2^2 + (1-c_0){\lambda_Y}\|(\widehat{\theta}-{\theta})_{\mathcal{S}^c}\|_1\leq (2+c_1) R_1^2.
\end{equation}
Together with \eqref{eq: special case 1}, we have
\begin{equation}
\|\widehat{\theta}-{\theta}\|_1=\|(\widehat{\theta}-{\theta})_{\mathcal{S}}\|_1+\|(\widehat{\theta}-{\theta})_{\mathcal{S}^c}\|_1\leq \frac{R_1^2}{\lambda_Y}\left(\frac{2+c_1 }{1-c_0}+\frac{c_1}{1+c_0}\right),
\label{eq: upper for alpha 0}
\end{equation}
and 
\begin{equation}
\|\widehat{\omega}-\omega\|_2\leq \frac{c_1 R_1^2}{R_2}. 
\label{eq: upper for omega 0}
\end{equation}

When \eqref{eq: special case 1} does not hold, then 
\begin{equation}
2R_2\|\widehat{\omega}-\omega\|_2+(1+c_0)\lambda_Y\|(\widehat{\theta}-{\theta})_{\mathcal{S}}\|_1  > c_1 R_1^2.
\label{eq: special case 2}
\end{equation}
And \eqref{eq: key basic inequality} implies  
\begin{equation}
\begin{aligned}
&\frac{1}{2n}\|\hatW(\omega-\widehat{\omega})+X({\theta}-\widehat{\theta})\|_n^2+(1-c_0){\lambda_Y}\|(\widehat{\theta}-{\theta})_\mathcal{N}\|_1\\
&\leq  \frac{4+2c_1}{c_1}R_2\|\hatomega-\omega\|_2+\left(1+c_0+\dfrac{2}{c_1}\right)\lambda_Y\|(\widehat{\theta}-{\theta})_\mathcal{S}\|_1.   
\end{aligned}
\label{eq: key basic inequality 0}
\end{equation}
In this case, we consider the restricted parameter space, 
\begin{equation}
\label{eq: Set C0 }
\mathcal{C}_0=\left\{\delta= \left(\begin{array}{c}
     \widehat{\theta}-\theta \\
    \widehat{\omega}-\omega 
\end{array}\right): \|(\widehat{\theta}-{\theta})_{\mathcal{S}^c}\|_1\leq \frac{c_2R_2}{\lambda_Y}\|\widehat{\omega}-\omega\|_2 +    c_3\|(\widehat{\theta}-{\theta})_\mathcal{S}\|_1\right\},
\end{equation}
where   $c_2=\frac{4+2c_1}{(1-c_0)c_1}$ and  $c_3=\frac{(c_1+c_0c_1+2)}{(1-c_0)c_1}$.
 \begin{Lemma}  Suppose the conditions in Theorem \ref{thm: estimation bound} hold. Then w.p.a.1 
\begin{equation}
\sup_{ \delta \in \mathcal{C}_0} \frac{\frac{1}{n}\|\hatW(\omega-\widehat{\omega})+X({\theta}-\widehat{\theta})\|_2^2}{M^{-1}\|\omega-\widehat{\omega}\|_2^2+\|{\theta}-\widehat{\theta}\|_2^2} \geq 2c
\label{eq: RE concentration}
\end{equation}
 \label{lem: RE concentration}
for some universal positive constant $c$. 
 \end{Lemma}

% Together with \eqref{eq: key basic inequality}, we have 
% \begin{equation}
% cM^{-1}\|\hatomega-\omega\|_2^2+c\|\hattheta-\theta\|_2^2+(1-c_0)\lambda_Y\|(\widehat{\theta}-{\theta})_\mathcal{N}\|_1\leq 2R_1^2+2R_2\|\widehat{\omega}-\omega\|_2 + (1+c_0)\lambda_Y\|(\widehat{\theta}-{\theta})_\mathcal{S}\|_1 .
% \label{eq: basic inequality 1}
% \end{equation}
By adding both sides with $(1-c_0)\lambda_Y\|(\widehat{\theta}-{\theta})_\mathcal{S}\|_1$ to  \eqref{eq: key basic inequality}, we have 
\begin{equation}
\dfrac{1}{2n}\|\hat W(\omega - \hat\omega) + X(\theta - \hat\theta)\|_2^2+(1-c_0)\lambda_Y\|\widehat{\theta}-{\theta}\|_1\leq 2R_1^2+2R_2\|\widehat{\omega}-\omega\|_2+ 2\lambda_Y\|(\widehat{\theta}-{\theta})_\mathcal{S}\|_1.
\label{eq: basic inequality 2}
\end{equation}
% \begin{equation}
% cM^{-1}\|\hatomega-\omega\|_2^2+c\|\hattheta-\theta\|_2^2+(1-c_0)\lambda_Y\|\widehat{\theta}-{\theta}\|_1\leq 2R_1^2+2R_2\|\widehat{\omega}-\omega\|_2+ 2\lambda_Y\|(\widehat{\theta}-{\theta})_\mathcal{S}\|_1.
% \label{eq: basic inequality 2}
% \end{equation}
Since $\frac{1}{2(c )}a^2+\frac{c }{2}b^2\geq ab$ for any $a,b>0$,  we have 
\begin{equation}
R_2\|\widehat{\omega}-\omega\|_2\leq \frac{M}{2c} R_2^2+\frac{c}{2M}\|\widehat{\omega}-\omega\|_2^2,
\label{eq: decouple 1}
\end{equation}
and 
\begin{equation}
\begin{aligned}
2\lambda_Y\|(\widehat{\theta}-{\theta})_\mathcal{S}\|_1\leq 2\sqrt{s}\lambda_Y \|(\widehat{\theta}-{\theta})_\mathcal{S}\|_2&\leq \frac{2}{c} {s}\lambda_Y^2+\frac{c}{2} \|(\widehat{\theta}-{\theta})_\mathcal{S}\|_2^2\\
&\leq \frac{2}{c} {s}\lambda_Y^2+\frac{c}{2} \|\widehat{\theta}-{\theta}\|_2^2.
\label{eq: decouple 2}
\end{aligned}
\end{equation}
By Lemma \ref{lem: RE concentration}, 
\begin{equation}\label{eq: compatibility Y}
    \frac{1}{2n}\|\hatW(\omega-\widehat{\omega})+X({\theta}-\widehat{\theta})\|_2^2 \geq {cM^{-1}\|\omega-\widehat{\omega}\|_2^2+c\|{\theta}-\widehat{\theta}\|_2^2}
\end{equation} 
and thus 
\[R_2\|\widehat{\omega}-\omega\|_2 + 2\lambda_Y\|(\widehat{\theta}-{\theta})_\mathcal{S}\|_1 \leq \dfrac{M}{2c}R_2^2 + \dfrac{2}{c}s\lambda_Y^2 + \frac{1}{4n}\|\hatW(\omega-\widehat{\omega})+X({\theta}-\widehat{\theta})\|_2^2.\]
Combining the above with \eqref{eq: basic inequality 2}, we have 
\begin{equation}
\frac{1}{4n}\|\hatW(\omega-\widehat{\omega})+X({\theta}-\widehat{\theta})\|_2^2+(1-c_0)\lambda_Y\|\widehat{\theta}-{\theta}\|_1\leq 2R_1^2+\frac{M}{c}R_2^2 +\frac{2s\lambda_Y^2}{c},
\label{eq: basic inequality 3 with L2 error}
\end{equation}
and using (\ref{eq: compatibility Y}) again,  
\begin{equation}
\dfrac{c}{2M}\|\hat\omega - \omega\|_2^2+\dfrac{c}{2}\|\theta-\hat\theta\|_2^2 + (1-c_0)\lambda_Y\|\widehat{\theta}-{\theta}\|_1\leq 2R_1^2+\frac{M}{c}R_2^2 +\frac{2s\lambda_Y^2}{c}.
\label{eq: basic inequality 3}
\end{equation}
All the inequalities above hold w.p.a.1 for any $\hat\theta$ and $\hat\omega$. Hence 
\begin{equation}
\sup_g\|\widehat{\omega}-\omega\|_2 \lep \sqrt{M}\cdot\sqrt{2R_1^2+ MR_2^2 +  s\lambda_Y^2  },
\label{eq: upper for omega 1}
\end{equation}
% The upper bound for $\|\widehat{\omega}-\omega\|_2$ obtained in \eqref{eq: upper for omega 1}  no less than $\sqrt{\frac{4M}{c}} R_1$ and hence the upper bound for $R_2\|\widehat{\omega}-\omega\|_2+(1+c_0)\lambda_Y\|(\widehat{\theta}-{\theta})_\mathcal{S}\|_1$  is no less than $\sqrt{\frac{4M}{c}}R_2R_1$. As   $M\asymp n^{\frac{1}{1+2\gamma}}$ and {\red Add these to assumptions.} $a_n\lesssim a_n\cdot n^{-\frac{\gamma}{2\gamma+1}}$, we have $\sqrt{\frac{4M}{c}}R_2R_1\geq c_1 R_1^2$ for some $c_1>0$, and hence there exists $\widehat{\theta}$ satisfying \eqref{eq: special case 2}.
% By \eqref{eq: basic inequality 3}, we also have 
and 
\begin{equation}
\sup_g\|\widehat{\theta}-{\theta}\|_1 \lep \frac{1}{ \lambda_Y}\left( R_1^2+ MR_2^2 + s\lambda_Y^2 \right).
\label{eq: upper for theta 1}
\end{equation}
Recall that the bounds of $R_1$ and $R_2$ are given as \eqref{eq: R1 bound} and \eqref{eq: R2 bound}. By combining \eqref{eq: upper for omega 0} and \eqref{eq: upper for omega 1}, we establish \eqref{eq: sup rate omega} and thus \eqref{eq: rate omega}; By combining \eqref{eq: upper for alpha 0} and \eqref{eq: upper for theta 1}, we establish \eqref{eq: sup rate theta} and thus  \eqref{eq: rate theta}.  
\par For \eqref{eq: rate derivative}, note that for any fixed $d$
\[(\hat g^\prime(d) - g^\prime(d) )^2 = ((\hat\beta-\beta)^\top B^\prime(d) - r^\prime_g(d))^2 \leq 2[(\hat\beta-\beta)^\top B^\prime(d)]^2 + 2[r^\prime_g(d)]^2 \]
and $\sup_g [r^\prime_g(d)]^2 \lep M^{-2\gamma+2}$ by Proposition \ref{prop: spline approx power}. Let $f_{\mathcal{D}}(s)$ denote the probability density function of $D_i$. Then \eqref{eq: rate derivative} follows by 
\[\begin{aligned}
\sup_g  \int_{\mathcal{D}}(\hat g^\prime - g^\prime )^2 &\lesssim \sup_g (\hat\beta-\beta)^\top \int_{\mathcal{D}} B^\prime(s)B^\prime(s)^\top {\rm d}s (\hat\beta-\beta) +  M^{-2\gamma+2}\\
&\leq \sup_g c_f^{-1}(\hat\beta-\beta)^\top \int_{\mathcal{D}} f_{\mathcal{D}}(s) B^\prime(s)B^\prime(s)^\top {\rm d}s (\hat\beta-\beta) +  M^{-2\gamma+2} \\
&= c_f^{-1}\sup_g (\hat\beta-\beta)^\top \E[B^\prime(D_i)B^\prime(D_i)^\top](\hat\beta-\beta) + M^{-2\gamma+2} \\
&\lesssim  M \cdot \sup_g\|\hat\beta-\beta\|_2^2 + M^{-2\gamma+2} \\
&\lep \dfrac{M^2(s+M)\logpM}{n} + M^{-2\gamma+2} 
\end{aligned}\]
where the second row applies Assumption \ref{assu: D v}, the fourth row applies Proposition \ref{prop: spline eigen}, and the last row applies \eqref{eq: rate omega}. 
%It remains to verify \eqref{eq: special case 2}
%\begin{equation}
% c_1 R_1^2 \leq R_2\|\widehat{\theta}-\theta\|_2+(1+c_0)\lambda_Y\|(\widehat{\theta}-{\theta})_{S}\|_1 \leq \sqrt{\frac{8}{\lambda_Y_{\rm min}(\Sigma)}}R_1R_2.
%\end{equation}

\subsection{Proof of Proposition \ref{prop: feasible}}
\par Starting from this proof, we introduce some new definitions and notations to facilitate the proof of the honesty of the uniform confidence band. For any positive sequences $a_n$ and $b_n$, we use $a_n = o_{u.p.}(b_n)$ to represent the fact that for any $\epsilon>0$, $\supg \Pr\left\{a_n/b_n > \epsilon \right\} \to 0$ as $n\to\infty$. The abbrevation ``u.p.'' means the probability converges uniformly for all functions $g$. Also, $a_n \lepg b_n$ means there exists some absolute constant $C$ such that $\inf_g \Pr\left\{a_n \leq Cb_n\right\} \to 1$; $a_n \gepg b_n$ and means $b_n \lepg a_n$; $a_n \asymppg b_n$ means $a_n \lepg b_n$ and $b_n \lepg a_n$. We will show stronger results that 
\begin{equation}\label{eq: sup feasibility 1}
    \|\hat\Sigma_F\Sigma_{F|\mathcal{L}}^{-1} - I_{p_F}\|_\infty \lepg M\sqrt{\dfrac{\logpM}{n}}
\end{equation}
and 
\begin{equation}\label{eq: sup feasibility 2}
   \|n^{-1/2}\hat F\Sigma_{F|\mathcal{L}}^{-1}\|_\infty \lepg M\sqrt{\dfrac{\logpM}{n}}
\end{equation}
\par Define $q_i^\prime := q^\prime(v_i)$, $\hatqp_i:=\hatqp(\hat v_i)$. Recall that $H$ is the matrix with the $(i,j)$-th element being $H_{ij} = H_j(v_i)$, $\tilde X_{i\cdot} := q_i^\prime X_{i\cdot}$, $\tilde K_{i\cdot} := q_i^\prime K_{i\cdot}$ and $F_{i\cdot} := \left(B_{i\cdot}^\top,H_{i\cdot}^\top,\tilde K_{i\cdot}^\top, X_{i\cdot}^\top, \tilde X_{i\cdot}^\top\right)^\top$. Define $\mathbb{E}_{\mathcal{L}}(\cdot) := \mathbb{E}\left[\cdot|\mathcal{L}\right]$. The (conditional) covariance matrices are denoted as $\Sigma_F := \E(F_\ic F_\ic^\top)$ and $\Sigma_{F|\mathcal{L}}:=\E_{\mathcal{L}}(\hat F_\ic \hat F_\ic^\top )$. 
% \par \textbf{Step 1: Show that $M^{-1}\lesssim \lambda_{\min}(\Sigma_F) \leq \lambda_{\max}(\Sigma_F) \lesssim 1$}. The conclusion follows the theoretical compatibility condition. 
\par We first state the following Lemma about the eigenvalues of $\Sigma_{F|\mathcal{L}}$. 
\begin{Lemma}\label{lem: cond cov eigen} Under the conditions of Proposition \ref{prop: feasible}, we have $$1 \lep \inf_g \lambda_{\min}(\Sigma_{F|\mathcal{L}}^{-1}) \leq \sup_g \lambda_{\max}(\Sigma_{F|\mathcal{L}}^{-1}) \lep M. $$
\end{Lemma}
 
\par \textbf{Step 1. Show (\ref{eq: sup feasibility 1}).} 
Recall the definitions of sub-Gaussian and sub-exponential norms given in (\ref{eq:def-subG-vec}) and (\ref{eq:def-sube-vec}). 
We first bound the sub-Gaussian norm of $\hat F_\ic$.
Note that for all $\beta\in\mathbb{R}^M$ such that $\|\beta\|_2=1$, 
\[(\beta^\top B_\ic)^2 \leq \|B_\ic\|_2 \leq \sum_{j=1}^MB_{ij}=1.\]
Thus, $\|B_\ic\|_{\psi_2|\mathcal{L}}\lesssim 1$ and similarly, $\|\hat H_\ic\|_{\psi_2|\mathcal{L}}, \|(K_\ell)_\ic\|_{\psi_2|\mathcal{L}}\lesssim 1$. Additionally, 
\[\begin{aligned}
    \supn\left|\hat q^\prime(\hat v_i)\right|  &\leq \|\hat\eta^{\ind}-\eta\|_2\supn\|\hat H_\ic^\prime\|_2 + \supn|\eta^\top\hat H_\ic^\prime - q^\prime(\hat v_i)| + \supn|q^\prime(\hat v_i)| \\ &\lesssim \|\hat\eta^{\ind}-\eta\|_2\cdot M + M^{-\gamma+1} + \sup_{v}|q^\prime(v)| \\
    &\lesssim \|\hat\eta^{\ind}-\eta\|_2\cdot M + 1.
\end{aligned}\] 
Thus, 
\[\begin{aligned}
    \|\hat F_\ic\|_{\psi_2|\mathcal{L}} &\leq   \|B_\ic\|_{\psi_2|\mathcal{L}} +  \|\hat H_\ic\|_{\psi_2|\mathcal{L}} + 
    \supn |\hat q^\prime(\hat v_i)|\sum_{\ell = 1}^{p_z}\sup_{j\in[M]}\|(K_\ell)_\ic\|_{\psi_2|\mathcal{L}}  \  + \left(1+\supn |\hat q^\prime(\hat v_i)|\right)\cdot\|X_\ic\|_{\psi_2} \\ 
    &\lesssim  \|\hat\eta^{\ind}-\eta\|_2\cdot M + 1  
\end{aligned}\]
and thus $\|\Sigma_{F|\mathcal{L}}^{-1}\hat F_\ic\|_{\psi_2|\mathcal{L}} \lesssim \|\Sigma_{F|\mathcal{L}}^{-1}\|_2 \left(\|\hat\eta^{\ind}-\eta\|_2\cdot M + 1\right) $. In the same manner we can also deduce $|\hat F_\ic|\lesssim \|\hat\eta^{\ind}-\eta\|_2\cdot M + 1$. Thus, any coordinate of $ \hat F_\ic \hat F_\ic^\top \Sigma_{F|\mathcal{L}}^{-1}$ then has a sub-Gaussian norm norm (conditionally on $\mathcal{L}$) bounded by $${\rm SG}_{\max} := C\left[\|\hat\eta^{\ind}-\eta\|_2\cdot M + 1\right]^2 \cdot\|\Sigma_{F|\mathcal{L}}^{-1}\|_2$$
where by \eqref{eq: sup rate omega} and Lemma \ref{lem: cond cov eigen} 
\begin{equation}\label{eq: SG max bound}
   \supg {\rm SG}_{\max} \lep (o_p(1) + 1)^2M \lep M. 
\end{equation}
 Then by \citep[Proposition 5.10]{vershynin2010introduction} and union bound, for any $t>0$
\[\begin{aligned}
    &\ \ \ \ \supg {\rm \Pr}\left(\|\hat\Sigma_F\Sigma_{F|\mathcal{L}}^{-1} - I_{p_F}\|_\infty > t\Bigg|\mathcal{L}\right) \\
    &= \supg {\rm \Pr}\left(\max_{j,k\in[p_F]}\left| \left(\dfrac{1}{n}\sumn \left(\hat F_\ic \hat F_\ic^\top \Sigma_{F|\mathcal{L}}^{-1} - \E_{\mathcal{L}}\left(\hat F_\ic \hat F_\ic^\top \Sigma_{F|\mathcal{L}}^{-1}\right)\right)\right)_{jk} \right| > t\Bigg|\mathcal{L}\right) \\ 
    &\leq 2p_F^2 \cdot \exp\left(-cn \cdot \dfrac{t^2}{\supg {\rm SG}_{\max}^2  } \right) \lep 2p_F^2 \cdot \exp\left(-cn \cdot \dfrac{t^2}{ M^2  } \right)
\end{aligned}\]
where the last inequality applies \eqref{eq: SG max bound}.  Taking $t =  \sqrt{3M^2\log p_F / (cn)}$, we have 
\[\begin{aligned}
    \supg {\rm \Pr}\left(\|\hat\Sigma_F\Sigma_{F|\mathcal{L}}^{-1} - I_{p_F}\|_\infty >  \sqrt{\dfrac{3M^2\log p_F}{cn}}\Bigg|\mathcal{L}\right)   
    &\leq 2p_F^2\cdot\exp\left(-3\log p_F\right) \to 0.
\end{aligned}\]
In other words, the supreme of conditional probability, as a random variable uniformly bounded in $[0,1]$, is $o_p(1)$. Thus by the Bounded Convergence Theorem, 
\[ \begin{aligned}
    \supg {\rm \Pr}\left(\|\hat\Sigma_F\Sigma_{F|\mathcal{L}}^{-1} - I_{p_F}\|_\infty >\sqrt{\dfrac{3\cdot M^2\log p_F}{cn}}\right)  
    &\leq   \E\left[\supg {\rm \Pr}\left(\|\hat\Sigma_F\Sigma_{F|\mathcal{L}}^{-1} - I_{p_F}\|_\infty >\sqrt{\dfrac{3\cdot M^2\log p_F}{cn}}\Bigg|\mathcal{L}\right)  \right]\\
    &\to 0
\end{aligned}\]
and thus $\|\hat\Sigma_F\Sigma_{F|\mathcal{L}}^{-1} - I\|_\infty \lepg M\sqrt{\dfrac{\log p_F}{n}}$. 
\par \textbf{Step 2. Show (\ref{eq: sup feasibility 2}).} Recall that conditionally on $\mathcal{L}$, $\hat F_\ic^\top \Sigma_{F|\mathcal{L}}^{-1}$ has sub-Gaussian norm bounded by ${\rm SG}_{\max}^{(1)}:=C\left[\|\hat\eta^{\ind} - \eta\|_2\cdot M + 1\right]\cdot \|\Sigma_{F|\mathcal{L}}^{-1}\|_2$ with 
\[\supg {\rm SG}_{\max}^{(1)} \lep M\]
following the arguments for \eqref{eq: SG max bound}. Then for any $t>0$
\[\begin{aligned}
   \supg {\rm Pr}\left(\|\hat F\Sigma_{F|\mathcal{L}}^{-1}\|_\infty > t\Bigg|\mathcal{L}\right) &= \supg  {\rm Pr}\left(\supn\left\|\Sigma_{F|\mathcal{L}}^{-1}\hat F_\ic \right\|_\infty > t\Bigg|\mathcal{L}\right) \\
   &\leq n\cdot p_F \exp\left(-\dfrac{ct^2}{\supg {\rm SG}_{\max}^{(1)}}\right) \lep n\cdot p_F \exp\left(-\dfrac{ct^2}{M^2}\right).
\end{aligned}\]
Taking $t =2 M\sqrt{\log ( n p_F ) / c}$, by similar arguments we deduce 
\[\begin{aligned}
   \supg {\rm Pr}\left(\|\hat F\Sigma_{F|\mathcal{L}}^{-1}\|_\infty \lep CM\sqrt{\log ( np_F)}\Bigg|\mathcal{L}\right)
   &\leq n\cdot p_F \exp\left( -2\log (np_F) \right)  \to 0.
\end{aligned}\]
Using the Bounded Convergence Theorem again, we deduce 
\[\supg {\rm Pr}\left(\|\hat F\Sigma_{F|\mathcal{L}}^{-1}\|_\infty >  2M\sqrt{\log (np_F) / c}\right) \to 0,\]
and hence $n^{-1/2}\|\hat F\Sigma_{F|\mathcal{L}}^{-1}\|_\infty \lepg M\sqrt{\dfrac{\log (n p_F )}{n}} \lesssim M\sqrt{\dfrac{\log (pM)}{n}}$. 
%%%%%%%%%%%%%%%%%%%
\subsection{Proof of Proposition \ref{prop: debiased estimator}} 
To prepare for the honesty, we will show the results with $o_p$ and $\asymp_p$ replaced by $o_{u.p.}$ and $\asymppg$.   Recall that $\hat m(d) := \hat\Omega_B B^\prime(d) $. Note that 
\begin{equation}\label{eq: sqrt n g prime}
\begin{aligned}
   &\ \ \ \ \sqrt{n}\left(\tilde{g}^\prime(d) - g^\prime(d)\right)  = \sqrt{n}\left( \sum_{i=1}^{M}\beta_jB_j^\prime(d) - g^\prime(d) \right) + \sqrt{n} B^\prime(d)^\top (\tilde\beta-\beta). 
\end{aligned}
\end{equation}
By Proposition \ref{prop: spline approx power}, 
\[\sup_g \sup_{d\in\mathcal{D}} \left|\sqrt{n}\left( \sum_{i=1}^{M}\beta_jB_j^\prime(d) - g^\prime(d) \right)\right| = O(\sqrt{n}M^{-\gamma+1}).\]
By the definition of the debiased estimator $\tilde\beta$ in \eqref{eq: db beta}, 
\begin{equation*}
    \begin{aligned}
        \sqrt{n} B^\prime(d)^\top (\tilde\beta-\beta) &= \dfrac{\hatm(d)^\top  \sumn \hat F_\ic \varepsilon_i}{\sqrt{n}} + \dfrac{\hatm(d)^\top \sumn \hat F_\ic (\hat\varepsilon_i-\varepsilon_i)}{\sqrt{n}} + \sqrt{n}B^\prime(d)^\top(\hat\beta-\beta) \\
        &= \mathcal{Z}(d) + \varDelta(d), 
    \end{aligned}
\end{equation*}
where \[\mathcal{Z}(d) :=\dfrac{\hatm(d)^\top \sumn \hat F_\ic \varepsilon_i}{\sqrt{n}},\ \varDelta(d) := \dfrac{\hatm(d)^\top \sumn \hat F_\ic (\hat\varepsilon_i-\varepsilon_i)}{\sqrt{n}} + \sqrt{n}B^\prime(d)^\top(\hat\beta-\beta).\]
 We then need to prove the following results:
 \begin{enumerate}[(S1)] 
     \item The scale of $\hat s(d)$   
        \begin{equation}\label{eq: variance scale}
         \hat{s}(d) := \sqrt{\hat m(d)^\top\hat \Sigma_F\hat m(d)} \asymppg M^{1.5}.
     \end{equation} 
      \item $ \sup_{d\in\mathcal{D}}\dfrac{\sqrt{n}M^{-\gamma+1}  + |\varDelta(d)|}{\hat s(d)} \lepg \frac{1}{\log n}$ uniformly for all $d$. The additional $1/\log n$ handles the Gaussian approximation for Theorem \ref{thm: limiting distribution}.
\end{enumerate} 
  \par {\bf\underline{Proof of (S1)}}. By Proposition \ref{prop: feasible}, the vector $j$-th column of $\Sigma_{F|\mathcal{L}}^{-1}$ belongs to the feasible set of the optimization algorithm (\ref{eq: projection direction 1}). Define $\textbf{i}_j$ as the $j$-th standard basis with the $j$-th element being one and others being zero. Then,
\begin{equation}\begin{aligned}\label{eq: upper bound theory}
    \widehat{m}(d)^{\top}\hat\Sigma_F\widehat{m}(d) &= B^\prime(d)^\top \hat\Omega_B\hat\Sigma_F\hat\Omega_B^\top B^\prime(d) \leq \|B^\prime(d)\|_1^2 \cdot \|\hat\Omega_B\hat\Sigma_F\hat\Omega_B^\top\|_\infty \lesssim M^2 \|\hat\Omega_B\hat\Sigma_F\hat\Omega_B^\top\|_\infty\\
    \end{aligned}
\end{equation}
where the last step applies Proposition \ref{prop: SplineDerivativeL2Norm}, and 
\begin{equation}\label{eq: OmegaSigmaOmega sup norm}
\begin{aligned}
  \|\hat\Omega_B\hat\Sigma_F\hat\Omega_B^\top\|_\infty &\leq     \max_{j\in[M]}|\hat \Omega_{Bj}^\top\hat\Sigma_F\hat\Omega_{Bj}| \\
  &\lepg  \max_{j\in[M]}|\textbf{i}_j^\top\Sigma_{F|\mathcal{L}}^{-1}\hat\Sigma_F\Sigma_{F|\mathcal{L}}^{-1}\textbf{i}_j| \\
    &\leq  \max_{j\in[M]}|\textbf{i}_j^\top \Sigma_{F|\mathcal{L}}^{-1}\textbf{i}_j| +   \|\textbf{i}_j\|_1^2 \cdot \|\Sigma_{F|\mathcal{L}}^{-1}\hat\Sigma_F\Sigma_{F|\mathcal{L}}^{-1} - \Sigma_{F|\mathcal{L}}^{-1}\|_\infty \\
      &\lepg   M +   \|\Sigma_{F|\mathcal{L}}^{-1}\hat\Sigma_F\Sigma_{F|\mathcal{L}}^{-1} - \Sigma_{F|\mathcal{L}}^{-1}\|_\infty 
\end{aligned}
\end{equation}
where the first step applies the positive semi-definiteness of $\hat\Omega_B\hat\Sigma_F\hat\Omega_B$, the second step follows by the definition of $\hat\Omega_B$ and the fact that $\Sigma_{F|\mathcal{L}}^{-1}$ is feasible for (\ref{eq: projection direction 1}) with uniformly high probability by \eqref{eq: sup feasibility 1} and \eqref{eq: sup feasibility 2}, and the last step applies Lemma \ref{lem: cond cov eigen}.  
% We then aim at an upper bound of $m(d)^{\top}\hat F^{\top}\hat Fm(d)$. Note that 
% \[\begin{aligned}
%   m(d)^{\top}\hat\Sigma_F m(d) &\leq m(d)^{\top}   \Sigma_{F|\mathcal{L}} m(d) + \left|  m(d)^{\top} \left(\hat\Sigma_F - \Sigma_{F|\mathcal{L}}\right)m(d) \right| \\
%   &= B^\prime(d)^{\top} I_B   \Sigma_{F|\mathcal{L}}^{-1} \xi(d) + \left|  B^\prime(d)^{\top} I_B \left(\Sigma_{F|\mathcal{L}}^{-1}\hat\Sigma_F\Sigma_{F|\mathcal{L}}^{-1} - \Sigma_{F|\mathcal{L}}^{-1}\right)\xi(d) \right| \\
%    &\lesssim M\|B^\prime(d)\|_2^2 +  \|\xi(d)\|_1^2 \|\Sigma_{F|\mathcal{L}}^{-1}\hat\Sigma_F\Sigma_{F|\mathcal{L}}^{-1} - \Sigma_{F|\mathcal{L}}^{-1}\|_\infty \\
%    &\leq M\|B^\prime(d)\|_2^2 +  M\|B^\prime(d)\|_2^2 \cdot \|\Sigma_{F|\mathcal{L}}^{-1}\hat\Sigma_F\Sigma_{F|\mathcal{L}}^{-1} - \Sigma_{F|\mathcal{L}}^{-1}\|_\infty \\
% \end{aligned}\]
It then suffices to show that $$\|\Sigma_{F|\mathcal{L}}^{-1}\hat \Sigma_F \Sigma_{F|\mathcal{L}}^{-1} - \Sigma_{F|\mathcal{L}}^{-1}\|_\infty = o_{u.p.}(1).$$
We have shown $\|\Sigma_{F|\mathcal{L}}^{-1}\hat F_\ic\|_{\psi_2|\mathcal{L}} \lesssim  \|\Sigma_{F|\mathcal{L}}^{-1}\|_2\cdot (\|\hat\eta^\ind - \eta\|_2\cdot M+1)$ in the proof of (\ref{eq: sup feasibility 1}) for Proposition \ref{prop: feasible}. 
Any coordinate of $ \Sigma_{F|\mathcal{L}}^{-1}\hat F_\ic \hat F_\ic^\top \Sigma_{F|\mathcal{L}}^{-1}$ then has a sub-exponential norm (conditionally on $\mathcal{L}$) bounded by 
${\rm SE}_{\max} = C\left[\|\hat\eta^{\ind}-\eta\|_2\cdot M + 1\right]^2 \cdot\|\Sigma_{F|\mathcal{L}}^{-1}\|_2^2$, and by \eqref{eq: sup rate omega} and Lemma \ref{lem: cond cov eigen}
$$\supg {\rm SE}_{\max} \lep \left(o(1) + 1\right)^2M^2 \lesssim M. $$
Since $\E_{\mathcal{L}}[\Sigma_{F|\mathcal{L}}^{-1}\hat \Sigma_F \Sigma_{F|\mathcal{L}}^{-1} ] = \Sigma_{F|\mathcal{L}}^{-1}$, by \citet[Corollary 5.17]{vershynin2010introduction} for centered sup-exponential variables and the union bound, for any $t>0$
\[\begin{aligned}
    &\ \ \ \ \supg {\rm \Pr}\left(\|\Sigma_{F|\mathcal{L}}^{-1}\hat\Sigma_F\Sigma_{F|\mathcal{L}}^{-1} - \Sigma_{F|\mathcal{L}}^{-1}\|_\infty > t\Bigg|\mathcal{L}\right) \\
    &\leq 2p_F^2 \cdot \exp\left(-cn\cdot\min\left(\dfrac{t^2}{\supg {\rm SE}_{\max}^2},\dfrac{t}{\supg {\rm SE}_{\max}}\right)  \right) \\
    &\lep 2p_F^2 \cdot \exp\left(-cn\cdot\min\left(\dfrac{t^2}{M^4},\dfrac{t}{M^2}\right)  \right)
\end{aligned}\]
Taking $t = 2\sqrt{M^4\log p_F / cn}$ with $C_t$ large enough, we have 
\[\dfrac{t^2}{M^4} \lesssim \dfrac{\log p_F}{n} \to 0\]
with $n$ large enough, implying that $t^2/M^4 \leq t/M^2$.  
Thus,  
\[\begin{aligned}
   \supg {\rm \Pr}\left(\|\Sigma_{F|\mathcal{L}}^{-1}\hat\Sigma_F\Sigma_{F|\mathcal{L}}^{-1} - \Sigma_{F|\mathcal{L}}^{-1}\|_\infty > C_t\sqrt{\dfrac{M^4\log p_F}{n}}\Bigg|\mathcal{L}\right)   
    &\lep 2p_F^2\cdot\exp\left(-4\log p_F\right) \to 0.
\end{aligned}\]
In other words, the supreme of the conditional probability is $o_p(1)$. Thus by the Bounded Convergence Theorem, 
\[ \begin{aligned}
    &\ \ \ \ \supg{\rm \Pr}\left(\|\Sigma_{F|\mathcal{L}}^{-1}\hat\Sigma_F\Sigma_{F|\mathcal{L}}^{-1} - \Sigma_{F|\mathcal{L}}^{-1}\|_\infty > C_t\sqrt{\dfrac{M^4\log p_F}{n}}\right) \\
    &\leq  \E\left[\supg{\rm \Pr}\left(\|\Sigma_{F|\mathcal{L}}^{-1}\hat\Sigma_F\Sigma_{F|\mathcal{L}}^{-1} - \Sigma_{F|\mathcal{L}}^{-1}\|_\infty > C_t\sqrt{\dfrac{M^4\log p_F}{n}}\Bigg|\mathcal{L}\right)  \right] \to 0,
\end{aligned} \]
and thus $\|\Sigma_{F|\mathcal{L}}^{-1}\hat\Sigma_F\Sigma_{F|\mathcal{L}}^{-1} - \Sigma_{F|\mathcal{L}}^{-1}\|_\infty \lepg M^2\sqrt{\dfrac{\log p_F}{n}} \lesssim  M^2\sqrt{\dfrac{\log (pM)}{n}} = o_p(M)$. Then by (\ref{eq: OmegaSigmaOmega sup norm}) we have 
\begin{equation}\label{eq: result OmegaSigmaOmega sup norm}
    \|\hat\Omega_B\hat\Sigma_F\hat\Omega_B^\top\|_\infty \lepg M.
\end{equation} 
Together with (\ref{eq: upper bound theory}), we deduce 
\begin{equation}\label{eq: scale result upper}
    \hat m(d)^\top\hat\Sigma_F\hat m(d) \lepg M^3.
\end{equation}
%%%%%%%%%%%%%%%%%%%%%%%%%%%%%%%%%%%%%%%%%%%%%%%%%%%%%%%%%%%%%%%%%%
 \par The proof the other side of the inequality \eqref{eq: variance scale} follows \citet[Lemma 12]{javanmard2014confidence}. We construct the following estimator   
\begin{equation}
\begin{aligned}
{m}^{*}(d)=\arg\min\ &m^{\top}\hat\Sigma_F m\\
\text{subject to}& \left|  B^\prime(d)^\top I_B\hat\Sigma_F m- \|B^\prime(d)\|_2^2\right|\leq \mu\|B^\prime(d)\|_1^2\end{aligned}
\label{eq: proof projection direction}
\end{equation}
where $\mu = \max_{j\in[M]}\mu_j$ and $I_B$ is the first $M$ rows of the $p_F$-dimensional identity matrix. Note that for any $M\times p_F$ matrix $\hat\Omega_{B}  = (\hat\Omega_{Bj})^\top_{j\in[M]}$ belonging to the feasible set of \eqref{eq: projection direction 1}, $\hat\Omega_{B}^\top B^\prime(d)$ belongs to the feasible set of \eqref{eq: proof projection direction} due to the fact that 
\begin{equation}
\begin{aligned}
\left|B^\prime(d)^{\top} I_B \hat\Sigma_F I_B^\top B^\prime(d) -\|B^\prime(d)\|_2^2\right|&=\left|B^\prime(d)^{\top} I_B\left(\hat\Sigma_F \hat\Omega_B- I_B^\top \right) B^\prime(d)\right|\\ 
&\leq \|B^\prime (d)\|_1^2\left\|\hat\Sigma_F\hat\Omega_B- I_B^\top \right\|_{\infty}\\
&\leq\|B^\prime(d)\|_1^2 \mu ,
\end{aligned}
\end{equation}
where the last inequality follows from the feasibility of $\hat\Omega_B$ for  \eqref{eq: projection direction 1}. Recall that $\hat m(d) := \hat\Omega_B B^\prime(d)$. Hence, we have 
\begin{equation}
(m^{*}(d))^{\top}\hat\Sigma_F m^{*}(d)\leq \widehat{m}(d)^{\top}\hat\Sigma_F\widehat{m}(d).
\label{eq: proof lower bound}
\end{equation}
It is thus sufficient to establish a lower bound for $(m^{*}(d))^\top\hat\Sigma_F m^{*}(d)$. Due to the feasibility condition of \eqref{eq: proof projection direction}, we have $-B^\prime(d)^{\top} I_B \hat\Sigma_F m^{*}(d)+\|B^\prime(d)\|_2^2-\mu\|B^\prime(d)\|_1^2\leq 0$ and hence for any $c_0>0$,
\begin{equation}
\begin{aligned}
&(m^{*}(d))^{\top}\hat\Sigma_Fm^{*}(d)\geq (m^{*}(d))^{\top}\hat\Sigma_Fm^{*}(d)+c_0\left(-B^\prime(d)^{\top} I_B\hat\Sigma_Fm^{*}(d)+\|B^\prime(d)\|_2^2-\mu\|B^\prime(d)\|_1^2\right)\\
&\geq \min_{m\in \R^{p_F}}\left(m^{\top}\hat\Sigma_Fm+c_0\left(-B^\prime(d)^{\top} I_B\hat\Sigma_Fm + \|B^\prime(d)\|_2^2-\mu\|B^\prime (d)\|_1^2\right)\right)\\
&\geq -\frac{c_0^2}{4}B^\prime(d)^{\top} I_B\hat\Sigma_F I_B^\top B^\prime(d)+c_0\left(\|B^\prime(d)\|_2^2-\mu\|B^\prime(d)\|_1^2\right)\\
&= -\frac{c_0^2}{4}B^\prime(d)^{\top} \hat\Sigma_B B^\prime(d)+c_0\left(\|B^\prime(d)\|_2^2-\mu\|B^\prime(d)\|_1^2\right)\\
\end{aligned}
\label{eq: theoretical lower bound}
\end{equation}
where $\hat\Sigma_B = n^{-1}B^\top B$. 
By Proposition \ref{prop: SplineDerivativeL2Norm}, $\|B^\prime(d)\|_2 \asymp \|B^\prime(d)\|_1 \asymp M$. Thus, $\|B^\prime(d)\|_2^2-\mu\|B^\prime(d)\|_1^2 \geq (1-C\mu)\|B^\prime(d)\|_2^2\geq c\|B^\prime(d)\|_2^2$
where the last step applies the fact that $\mu = o(1)$. 
Taking maximum of the right hand side of \eqref{eq: theoretical lower bound} over all $c_0>0$, we have 
\begin{equation}\label{eq: m star lower bound}
\begin{aligned}
(m^{*}(d))^{\top}\hat\Sigma_Fm^{*}(d)&\geq\max_{c_0} \left(-\frac{c_0^2}{4}B^\prime(d)^{\top} \hat\Sigma_B B^\prime(d)+c_0\left(\|B^\prime(d)\|_2^2-\mu\|B^\prime(d)\|_1^2\right)\right)\\
&= \frac{\left(\|B^\prime(d)\|_2^2-\mu\|B^\prime(d)\|_1^2\right)^2}{B^\prime(d)^{\top} \hat\Sigma_B B^\prime(d)}\geq \frac{c^2 \|B^\prime(d)\|_2^4}{B^\prime(d)^{\top} \hat\Sigma_B B^\prime(d)}. 
\end{aligned}
\end{equation}
It suffices to find an upper bound of $B^\prime(d)^{\top} \hat\Sigma_B B^\prime(d)$. By (\ref{eq: BB m EBB}), 
\[\|B^\top B - \E(B^\top B)\|_2 \lep \sqrt{n M^{-1}\log (pM)}\]
and $\|\E(B^\top B)\|_2\asymp nM^{-1}$ by Proposition \ref{prop: spline eigen}. Thus,  
\[\begin{aligned}\sup_{d\in\mathcal{D}}\left|\dfrac{B^\prime(d)^\top\left( \hat\Sigma_B - n^{-1}\E(B^\top B)\right)B^\prime(d)}{B^\prime(d)^\top n^{-1}\E(B^\top B)B^\prime(d)}\right| &\lep
\sqrt{\dfrac{\log(pM)}{n}} \to 0
\end{aligned}\]
which implies $B^\prime(d)^\top \hat\Sigma_B B^\prime(d) \asymp_p  B^\prime(d)^\top n^{-1}\E(BB^\top) B^\prime(d) \asymp M^{-1}\|B^\prime(d)\|_2^2$. Then 
\begin{equation}\label{eq: m star lower lower bound}
    \inf_{d\in\mathcal{D}} \dfrac{\|B^\prime(d)\|_2^4}{B^\prime(d)^\top \hat\Sigma_B B^\prime(d)} \gep \inf_{d\in\mathcal{D}} M\|B^\prime(d)\|_2^2 \gtrsim M^3.
\end{equation} 
Then by  \eqref{eq: proof lower bound}, (\ref{eq: m star lower bound}), and \eqref{eq: m star lower lower bound}, 
\begin{equation}\label{eq: scale result lower}
    \inf_g \inf_{d\in\mathcal{D}} \hat m(d)^\top\hat\Sigma_F\hat m(d) \geq  \inf_g \inf_{d\in\mathcal{D}} (m^{*}(d))^\top\hat\Sigma_Fm^{*}(d) \gtrsim_p M^3.
\end{equation}
The proof is completed by combining (\ref{eq: scale result upper}) and (\ref{eq: scale result lower}).  
  \par {\bf\underline{Proof of (S2)}}. 
By (\ref{eq: variance scale}) and the fact that $M = n^\nu$ with $\nu > \frac{1}{2\gamma+1}$ in Assumption \ref{assu: asym more},  we have $\sqrt{n}M^{-\gamma+1} = o(M^{\frac{2\gamma+1}{2}-\gamma+1}) = o(M^{1.5}) = o_p(\hat s(d))$. 
\par We then carefully decompose the residual $\hat\epsilon_i$. Recall that we define $\varDelta^Y_i = \hatW_\ic^\top(\omega-\hat\omega) + X_\ic(\theta-\hat\theta)$ and $\varDelta^v_i = q(v_i) - q(\hat v_i)$ in \eqref{eq: residual begin}. Then 
\begin{equation}\label{eq: resid in proof}
    \hat\epsilon_i - \varepsilon_i = \varDelta^Y_i + q(v_i) - q(\hat v_i) + r_g(D_i) + r_q(\hat v_i). 
\end{equation} 
Using Taylor expansion, 
\[\begin{aligned}
   \varDelta^v_i &= q^\prime(\hat v_i)(v_i - \hat v_i) + \qpp(v_i^*)(v_i - \hat v_i)^2 \\
    &= q^\prime(\hat v_i)(X_\ic^\top(\hat\varphi^\ind - \varphi) + K_\ic(\hat\kappa^\ind - \kappa) ) + q^\prime(\hat v_i)r_{\psi i} + \qpp(v_i^*)(v_i - \hat v_i)^2 \\
\end{aligned}\]
where $r_{\psi i}$ is the spline approximation error in \eqref{eq: d model spline}. Define 
\[\hat\varDelta^v_i = \hat q^\prime(\hat v_i)(X_\ic^\top(\hat\varphi^\ind - \varphi)+K_\ic^\top(\hat\kappa^\ind-\kappa))\]
with $\hat q^\prime(\hat v_i) = H^\prime(\hat v_i)^\top\hat\eta^\ind$. 
We deduce that
\begin{equation*}\label{eq: hat delta v}
\varDelta^v_i = \hat\varDelta^v_i +  q^\prime(\hat v_i)r_{\psi i} + \left(q^\prime(\hatv_i)-\hatqp_i\right)\left[X_{i\cdot}^\top(\hatvarphi^\ind-\varphi)+K_{i\cdot}^\top(\hatkappa^\ind-\kappa)\right] +\qpp(v_i^*)(v_i-\hatv_i)^2.
\end{equation*}
By \eqref{eq: resid in proof}, 
\begin{equation}\label{eq: resid in proof 2}
    \hat\epsilon_i - \varepsilon_i = \tilde r_i + \varDelta^Y_i + \hat\varDelta^v_i. 
\end{equation} 
where $\tilde r_i :=  (\varDelta^v_i - \hat\varDelta^v_i) + r_g(D_i) + r_q(\hat v_i)$ the high order terms in \eqref{eq: resid in proof 2}. \eqref{eq: resid in proof 2} yields a decomposition of the error we need to bound, given as 
\begin{equation}
\begin{aligned}
    \varDelta(d) &= \dfrac{\hat m(d)^\top \sumn \hat F_\ic \tilde r_i}{\sqrt{n}} + \left(\dfrac{\hat m(d)^\top \sumn \hat F_\ic(\varDelta^Y_i + \hat\varDelta^v_i)}{\sqrt{n}} + \sqrt{n}B^\prime(d)^\top(\beta-\hat\beta) \right) \\
    &=: \varDelta_1(d) + \varDelta_2(d).     
\end{aligned}
\end{equation}
\par \underline{We first bound $\varDelta_1(d)$}. Define $\check r_i :=  r_g(D_i) + r_q(\hat v_i) + q^\prime(\hat v_i)r_{\psi i}$ as the spline approximation errors. Then 
\[\tilde r_i = \check r_i + q^{\prime\prime}(v_i^*)(v_i - \hat v_i)^2 + (q^\prime(\hat v_i) - \hat q^\prime(\hat v_i))\left[X_\ic^\top (\hat\varphi^\ind-\varphi ) + K_\ic^\top (\hat\kappa^\ind - \kappa) \right]\]
where 
\[|\varDelta_1(d)| \leq \left|\dfrac{\hatm(d)^\top \sumn \hat F_\ic \check r_i}{\sqrt{n}}\right| + \left|\dfrac{\hatm(d)^\top \sumn \hat F_\ic (\tilde r_i - \check r_i)}{\sqrt{n}}\right| =: \varDelta_{11}(d) + \varDelta_{12}(d).\]
When $\nu\in(\frac{1}{2\gamma},\frac{1}{4.5})$ in Assumption \ref{assu: asym more} holds, by Cauchy-Schwartz inequality
\begin{equation} \label{eq: Delta11 bound}
    |\varDelta_{11}(d)|  \leq \sqrt{\hatm(d)^\top\dfrac{\hat F^\top\hat F}{n}\hatm(d)}\cdot\sqrt{ \sum_{i=1}^n  \check r_i^2} = \hat s(d)\sqrt{ \sumn  \check r_i^2} = o(\hat s(d) / \log n)
\end{equation} 
uniformly for all $d$, where the last equality applies Proposition \ref{prop: spline approx power} and the boundness of the function $q^\prime(\cdot)$, implying that $\sumn  \check r_i^2 
 = O(n\cdot M^{-2\gamma}) = o(1/\log n)$. For $\varDelta_{12}(d)$, again by Cauchy-Schwartz inequality 
  \[|\varDelta_{12}(d)|  \leq  \hat s(d)\sqrt{ \sumn  (\tilde r_i - \check r_i)^2}.\]
  It thus suffices to show the following lemma.
  \begin{Lemma}\label{lem: tilde r} Under the conditions for Theorem \ref{thm: limiting distribution},  $\supg \sqrt{\sumn (\tilde r_i - \check r_i)^2}=o_{u.p.}(1/\log n)$. 
  \end{Lemma}
 \par \underline{It remains to bound $\varDelta_2(d)$}. Define 
 \[\pi = (\beta^\top,\eta^\top,\theta^\top,-\kappa^\top,-\varphi^\top)^\top,\]
 \[\hat\pi = (\hat\beta^\top,\hat\eta^{\ind\top},\hat\theta^\top,-\hat\kappa^{\ind\top},-\hat\varphi^{\ind\top})^\top,\]
 and recall the definition of $\hat F_\ic$ in \eqref{eq: def F hat i}. Thus, \[\varDelta^Y_i+\hat\varDelta^v_i = \hat F_\ic^\top (\pi - \hat \pi).\] 
 We use $I_B$ to denote the first $M_D$ rows of the $p_F$ demensional identify matrix. Then uniformly for all $d$,
 \[\begin{aligned}
     |\varDelta_2(d)| &= \left|\sqrt{n}B^\prime(d)^\top\hat\Omega_B \sumn\hat F_\ic \hat F_\ic^\top(\pi-\hat\pi)  - B^\prime(d)^\top I_B(\pi-\hat\pi)\right| = \left|\sqrt{n}B^\prime(d)^\top(\hat\Omega_B\hat\Sigma_F - I_B)(\pi-\hat\pi)\right|. \\
   \sup_{d\in\mathcal{D}}|\varDelta_2(d)|  &\leq \sup_{d\in\mathcal{D}}\sqrt{n}\cdot \|B^\prime(d)\|_1 \cdot \|\hat\Omega_B\hat\Sigma_F - I_B\|_\infty \cdot \|\hat\pi-\pi\|_1 \\
     &\lepg \sqrt{n}\cdot M^2 \sqrt{\dfrac{\log (pM)}{n}} \cdot \|\hat\pi-\pi\|_1\\
     % &\ \ \ \ \left(\sqrt{M}\|\beta-\hatbeta\|_2 + \sqrt{M}\|\eta-\hateta\|_2 + \sqrt{p_zM}\|\kappa - \hat\kappa^\ind\|_2 + \|\theta-\hattheta\|_1 + \|\varphi - \hat\varphi^\ind\|_1\right). 
 \end{aligned}\]
 where the last step applies the bound of $\|B^\prime(d)\|_1$ by Proposition \ref{prop: SplineDerivativeL2Norm}, and the first restriction of \eqref{eq: projection direction 1}.  By \eqref{eq: sup rate omega} and Proposition \ref{prop: Lasso D}, we can deduce that under the conditions for Theorem \ref{thm: limiting distribution}, 
 \[\supg\|\beta - \hat\beta\|_2 + \supg\|\eta-\hat\eta\|_2 + \|\kappa - \hat\kappa^\ind\|_2 = o_p\left(\frac{1}{M\sqrt{\log (pM) }\log n }\right)\]
 and by \eqref{eq: sup rate theta}
  \[\supg\|\theta - \hat\theta\|_1 + \|\varphi-\hat\varphi^\ind\|_1  = o_p\left(\frac{1}{\sqrt{M \log (pM) } \log n}\right).\]
Thus, 
\[\begin{aligned}
  \supg\|\hat\pi-\pi\|_1 &\leq  \sqrt{M}\supg\|\beta-\hatbeta\|_2 + \sqrt{M}\supg\|\eta-\hateta\|_2 + \sqrt{p_zM}\|\kappa - \hat\kappa^\ind\|_2 + \supg\|\theta-\hattheta\|_1 + \|\varphi - \hat\varphi^\ind\|_1 \\
  &=o_p\left(\frac{1}{\sqrt{M \log (pM) } \log n}\right),
\end{aligned}\]
which implies 
\[ \sup_{d\in\mathcal{D}}|\varDelta_2(d)|\lepg \sqrt{n}\cdot M^2 \sqrt{\dfrac{\log (pM)}{n}} \cdot o_{u.p.}\left(\frac{1}{\sqrt{M \log (pM) } \log n}\right) = o_{u.p.}\left(\frac{1}{  \log n}\right) \]
Then, $$\sup_{d\in\mathcal{D}} \frac{\sup_{d}|\varDelta_2(d)|}{\hat s(d)} = o_p(\frac{M^{1.5}}{\hat s(d)\log n}) = o_{u.p.}(1/\log n).$$ We complete the proof of Proposition \ref{prop: debiased estimator}. 
\subsection{Proof of Theorem \ref{thm: limiting distribution}}
It suffices to prove the following three results. 
\begin{enumerate}[(R1)]
    \item  $\supg \left|\hat\sigma_\varepsilon/\sigma_\varepsilon - 1\right| \lep \frac{1}{n^{1/4}}$.  This result implies $\inf_g \hat\sigma_\varepsilon \gtrsim_p 1$, and together with (S2) in the proof of Proposition \ref{prop: debiased estimator} implies that \[\sup_{d\in\mathcal{D}}\left|\dfrac{\sqrt{n}\left(\tildegp(d) - \gp(d)\right) - \mathcal{Z}(d)}{\hat\sigma_\varepsilon\cdot \hat s(d)}\right| \lepg \frac{1}{\log n}.\] 
    \item  Show $\sup_{d\in\mathcal{D}}|\mathcal{Z}(d)/\hat s(d)| \lepg n^{1/4} / \log n$. This result together with the error bound for $\hat\sigma_\varepsilon$ in (R1) implies 
\begin{equation}\label{eq: bound H and tilde H}
  \sup_{d}\left|\mathbb{H}(d) - \tilde{\mathbb{H}}(d)\right| = o_{u.p.}\left(\dfrac{1}{\sqrt{\log n}}\right) 
\end{equation}
where $\mathbb{H}(d) := \dfrac{ \sqrt{n}\left(\tildegp(d) - \gp(d)\right)} {\hat\sigma_\varepsilon\cdot \hat s(d)}$ and $\tilde{\mathbb{H}}(d) := \dfrac{\mathcal{Z}(d)}{  \sigma_\varepsilon\cdot \hat s(d)}$.
 \item  There exists a version of Gaussian process $\mathbb{H}^{(1)}(d)$  such that 
     \begin{equation}
         \left|\sup_{d}|\tilde{\mathbb{H}}(d)| - \sup_{d}|\mathbb{H}^{(1)}(d)|\right| \lepg \frac{1}{\log n}
    \end{equation}
     with 
    \[\mathbb{H}^{(1)}(d) := (\sigma_\varepsilon^2 n)^{-1/2}\sumn B^\prime(d)^\top  \hat\Omega_B \hat F_\ic e^{(1)}_i\]
    where $e^{(1)}_i$ are i.i.d.\ $N(0,\sigma_\varepsilon^2)$ variables.
    Together with (\ref{eq: bound H and tilde H}) in (R2), we can deduce that 
      \begin{equation}\label{eq: H m H1 target}
         \left|\sup_{d}|{\mathbb{H}}(d)| - \sup_{d}|\mathbb{H}^{(1)}(d)|\right| \lepg \frac{1}{\log n}.
    \end{equation}
     We will also show
     \begin{equation}\label{eq: exp sup upper bound root log n}
         \supg \E\left[\sup_{d}| \mathbb{H}^{(1)}(d)|\right] \lesssim \sqrt{\log n} 
     \end{equation}     
     which by \citet[Corollary 2.1]{chernozhukov2014anti} implies that for any $\epsilon > 0$
     \begin{equation}\label{eq: anti concentration}
         \supg \sup_{x\in\mathbb{R}}\Pr\left[ \left|\sup_{d}|\mathbb{H}^{(1)}(d)| - x \right|\leq \epsilon \right] \lesssim \epsilon\sqrt{\log n}.
     \end{equation}
         % \[\dfrac{\sqrt{n}(\hat g^\prime(d)-g^\prime(d))}{\hat\sigma_\varepsilon\cdot \sqrt{\hat m()^\top\hat \Sigma_F\hat m}} = \dfrac{\mathcal{Z}+O(\sqrt{n}M^{-\gamma+1})+\varDelta_1 + \varDelta_2}{\hat\sigma_\varepsilon\cdot \sqrt{\hat m^\top\hat \Sigma_F\hat m}} \convd N(0,1).\]
    \item Recall that $\hat c_n(\alpha)$ is defined as the $(1-\alpha)$-quantile of $\sup_{d}|\hat{\mathbb{H}}(d)|$ with $\hat{\mathbb{H}}(d)$ defined in (\ref{eq: H hat}). Let $c_n(\alpha)$ be the $(1-\alpha)$-th quantile of $\sup_{d}|\mathbb{H}^{(1)}(d)|$. Show that for some sequences $\tau_{n}= o(1)$ and $\epsilon_n = o(1/\sqrt{\logpM})$, 
    \begin{equation}\label{eq: quantile bound}
        \begin{aligned}
             \hat c_n(\alpha) \geq c(\alpha + \tau_n) - \epsilon_n. \\
        \end{aligned}
    \end{equation}
 \end{enumerate}
 Then through (R1)-(R4), we deduce that 
 \[\begin{aligned}
     &\ \ \ \ \inf_g\Pr\left(g^\prime(d)\in\mathcal{C}_{n,\alpha}(d)\text{ for all }d\in\mathcal{D}\right) \\
     &\geq  \inf_g\Pr\left(\sup_{d}|{\mathbb{H}}(d)| \leq \hat c_n(\alpha)\right) \\
      &\geq  \inf_g\Pr\left(\sup_{d}|\mathbb{H}^{(1)}(d)| \leq \hat c_n(\alpha) + o((\logpM)^{-1/2}) \right) + o(1) \\
      &\geq \inf_g\Pr\left(\sup_{d}|\mathbb{H}^{(1)}(d)| \leq c_n(\alpha + \tau_n) - \epsilon_n + o((\logpM)^{-1/2}) \right) + o(1) \\
      &\geq 1 - \alpha - \tau_n -c\left[\epsilon_n + o((\logpM)^{-1/2})\right]\sqrt{\logpM} + o(1) \to 1-\alpha.
 \end{aligned}\]
 The second inequality applies (\ref{eq: H m H1 target}). The third inequality applies (\ref{eq: quantile bound}). The last inequality applies (\ref{eq: anti concentration}).
 \par {\bf\underline{Proof of (R1)}}. Note that 
 \[\left|\dfrac{\hat\sigma_\varepsilon}{\sigma_\varepsilon}-1\right|=\dfrac{\left|\dfrac{\hat\sigma_\varepsilon^2}{\sigma_\varepsilon^2}-1\right|}{\left|\dfrac{\hat\sigma_\varepsilon}{\sigma_\varepsilon}+1\right|} = \sigma_\varepsilon\dfrac{\left| \hat\sigma_\varepsilon^2 -\sigma_\varepsilon^2\right|}{\left|\hat\sigma_\varepsilon+\sigma_\varepsilon\right|},\]
 It thus suffices to show that $\supg|\hat\sigma_\varepsilon^2 -\sigma_\varepsilon^2|=o_p(1)$. Since 
 $n^{-1}\sumn \varepsilon_i^2 - \sigma_\varepsilon^2 = O_p(n^{-1/2})$ by the bounded fourth order moments of $\varepsilon_i$, and 
 it suffices to show that $\supg|\hat\sigma_\varepsilon^2 - n^{-1} \sumn\varepsilon_i^2| =  \supg|n^{-1}\sumn(\hat\varepsilon_i^2 - \varepsilon_i^2)| = o_p(1/n^{-1/4})$. As 
  \[\begin{aligned}
      \left|n^{-1}\sumn (\hat\varepsilon_i^2 - \varepsilon_i^2)\right| &\leq n^{-1}\sumn (\hat\varepsilon_i - \varepsilon_i)^2 + 2\left|n^{-1}\sumn (\hat\varepsilon_i - \varepsilon_i)\varepsilon_i\right| \\
      &\leq n^{-1}\sumn (\hat\varepsilon_i - \varepsilon_i)^2 + 2\sqrt{n^{-1}\sumn (\hat\varepsilon_i - \varepsilon_i)^2 }\cdot \sqrt{n^{-1}\sumn \varepsilon_i^2},
  \end{aligned}\]
  it remains to show $n^{-1}\sumn (\hat\varepsilon_i - \varepsilon_i)^2 = o_p(1/\sqrt{n})$. By (\ref{eq: residual begin}), $\hat\varepsilon_i - \varepsilon_i = \hat W_i^\top(\omega-\hat\omega) + X_\ic(\theta - \hat\theta) + q(v_i) - q(\hat v_i) + r_g(D_i) + r_q(\hat v_i)$, then 
  \[ \begin{aligned}
      \dfrac{1}{n}\sumn(\hat\varepsilon_i - \varepsilon_i)^2 &\lesssim \dfrac{1}{n}\sumn\left[\hat W_i^\top(\omega-\hat\omega) + X_\ic(\theta - \hat\theta)\right]^2 + \dfrac{1}{n}\sumn [q(v_i) - q(\hat v_i)]^2 + \dfrac{1}{n}\sumn[r_g^2(D_i) + r_q^2(\hat v_i)] \\
      &=: \varDelta^\varepsilon_1 +  \varDelta^\varepsilon_2 +  \varDelta^\varepsilon_3.
  \end{aligned}\]
 \eqref{eq: R1 bound}, \eqref{eq: R2 bound}, (\ref{eq: error L2 pl sparse L1}), (\ref{eq: basic inequality 3 with L2 error}), and $M\gg n^{1/(2\gamma)}$ in Assumption \ref{assu: asym more} imply $\supg \varDelta^\varepsilon_1 = o_p(n^{-1/2})$.  
(\ref{eq: sum sq q approx}) implies $\varDelta^\varepsilon_2 = o_p(n^{-1/2})$. Proposition \ref{prop: spline approx power} implies $\supg \varDelta^\varepsilon_3 = o_p(n^{-1/2})$. It completes the proof of $\supg|\hat\sigma_\varepsilon / \sigma_\varepsilon - 1| = o_p(n^{-1/4})$. 
%%%%%%%%%%%%%%%%%%%%%%%%%%%%%%%%%%%%%%%%%%%%%%%%%
\par {\bf \underline{Proof of (R2)}}. Note that
\begin{equation}\label{eq: bound Zd sd begin}
    \begin{aligned}
   \sup_{d\in\mathcal{D}} \dfrac{|\mathcal{Z}(d)|}{\hat s(d)} \leq \sup_{d\in\mathcal{D}}\dfrac{\|B^\prime(d)\|_1}{\hat s(d)}\cdot n^{-1/2}\|\sumn \hat\Omega_B \hat F_\ic \varepsilon_i \|_\infty \lepg (nM)^{-1/2}\|\sumn \hat\Omega_B \hat F_\ic \varepsilon_i \|_\infty
\end{aligned}
\end{equation} 
where the last inequality applies the bound of $\|B^\prime(d)\|_1$ by Proposition \ref{prop: SplineDerivativeL2Norm} and the lower bound of $\hat s(d)$ by \eqref{eq: variance scale}. 
\par It suffices to bound $\|\sumn \hat\Omega_B \hat F_\ic \varepsilon_i \|_\infty$. By the Chebyshev inequality   and the union bound,  
   \[\begin{aligned}
\Pr\left(\max_{j\in[M_D]}\dfrac{n^{-1/2}\left|\sumn\hat\Omega_{Bj}^\top\hat F_\ic\varepsilon_i\right|}{\sqrt{ \Omega_{Bj}^\top\hat\Sigma_F\Omega_{Bj}}\cdot \sigma_\varepsilon  } > \tau \Bigg|\mathcal{G} \right) \leq \dfrac{M}{\tau^2}. 
   \end{aligned}\] 
Taking $\tau = \sqrt{M\logpM}$,
     \[\begin{aligned}
       \supg \Pr\left(\max_{j\in[M_D]}\dfrac{n^{-1/2}\left|\sumn\hat\Omega_{Bj}^\top\hat F_\ic\varepsilon_i\right|}{\sqrt{ \Omega_{Bj}^\top\hat\Sigma_F\Omega_{Bj}}\cdot \sigma_\varepsilon  } >  \sqrt{M\logpM}\Bigg|\mathcal{G} \right)  \leq \dfrac{1}{\logpM} \to 0.
   \end{aligned}\]
   Using Bounded Convergence Theorem on the supreme of the conditional probability, 
   \[\begin{aligned}
       &\ \ \ \ \supg \Pr\left(\max_{j\in[M_D]}\dfrac{n^{-1/2}\left|\sumn\hat\Omega_{Bj}^\top\hat F_\ic\varepsilon_i\right|}{\sqrt{ \Omega_{Bj}^\top\hat\Sigma_F\Omega_{Bj}}\cdot \sigma_\varepsilon  } > \sqrt{ M\logpM} \right) \\ &\leq  \E\left(\supg \Pr\left(\max_{j\in[M_D]}\dfrac{n^{-1/2}\left|\sumn\hat\Omega_{Bj}^\top\hat F_\ic\varepsilon_i\right|}{\sqrt{ \Omega_{Bj}^\top\hat\Sigma_F\Omega_{Bj}}\cdot \sigma_\varepsilon  } >   \sqrt{ M\logpM}\Bigg|\mathcal{G} \right)\right) \to 0
   \end{aligned}\]
   which implies 
   \begin{equation}\label{eq: sup norm OFe bound}
       \begin{aligned}
   \left\|\sumn\hat\Omega_B\hat F_\ic\varepsilon_i\right\|_\infty \lepg  \sqrt{nM\logpM\cdot \max_{j\in[M_D]}\Omega_{Bj}^\top\hat\Sigma_F\Omega_{Bj}} 
      \cdot \sigma_\varepsilon \lepg M\sqrt{n\logpM}
   \end{aligned}
   \end{equation}
   where the last inequality applies by (\ref{eq: variance scale}) that $$\max_{j\in[M_D]}\Omega_{Bj}^\top\hat\Sigma_F\Omega_{Bj} \leq  \|\Omega_{B}^\top\hat\Sigma_F\Omega_{B}^\top \|_\infty \lepg M.$$
   \eqref{eq: bound Zd sd begin} and \eqref{eq: sup norm OFe bound} imply 
   \[\sup_{d\in\mathcal{D}} \dfrac{|\mathcal{Z}(d)|}{\hat s(d)} \leq \sup_{d\in\mathcal{D}}\dfrac{\|B^\prime(d)\|_1}{\hat s(d)}\cdot n^{-1/2}\|\sumn \hat\Omega_B \hat F_\ic \varepsilon_i \|_\infty \lepg \sqrt{M\logpM} = o(n^{1/4}/\log n),\]
   which completes the proof of (R2). 
%%%%%%%%%%%%%%%%%%%%%%%%%%%%%%%%%%%%%%%%%%%%%%%%%
\par {\bf\underline{Proof of (R3)}}. Let $\zeta_i(d) := \hat s(d)^{-1} \hat m(d)^\top\hat F_\ic \varepsilon_i$ and $\zeta_i^*(d) := \hat s(d)^{-1}\hat m(d)^\top \hat F_\ic e^{(1)}_i$ where $e^{(1)}_i$ are i.i.d.\ $N(0,\sigma_\varepsilon^2)$ variables. Then for any $d_0, d_1, \cdots$, the vector $(\zeta_i^*(d_0),\zeta_i^*(d_1),\cdots)^\top$ is jointly normal conditionally on $\mathcal{G}$ with pairwise covariance for any $(d_0,d_1)$, given as 
 $$ \E_{\mathcal{G}}[\zeta^*_i(d_0)\zeta^*_i(d_1)] =  \sigma_\varepsilon^2 \frac{\hat m(d_0)^\top\hat F_\ic \hat F_\ic \hat m(d_1)}{\hat s(d_0) \hat s(d_1) } = \E_{\mathcal{G}}[\zeta_i(d_0)\zeta_i(d_1)]. $$
 Define 
 \[ \mathbb{G} := \sup_{d\in\mathcal{D}} \sumn \zeta_i(d),\ \mathbb{G}^\prime := \sup_{d\in\mathcal{D}} \sumn [- \zeta_i(d)],\]
  \[ \mathbb{G}^* := \sup_{d\in\mathcal{D}} \sumn \zeta_i^*(d),\ \mathbb{G}^{*\prime} := \sup_{d\in\mathcal{D}} \sumn [-\zeta_i^*(d)].\]
We have the following lemma. 
\begin{Lemma}\label{lem: gaussian approx}Under the conditions for Theorem \ref{thm: limiting distribution}, 
\[ n^{-1/2}|\mathbb{G} - \mathbb{G}^*| \lepg \dfrac{1}{\sqrt{\log n}} ,\]
\[ n^{-1/2}|\mathbb{G}^\prime - \mathbb{G}^{*\prime}| 
\lepg \dfrac{1}{\sqrt{\log n}}.\]
\end{Lemma}
Note that $\sup_{d\in\mathcal{D}}|\sumn\zeta_i(d)| = \max\{\mathbb{G},\mathbb{G}^\prime \}=\frac{\mathbb{G}+\mathbb{G}^\prime+|\mathbb{G}-\mathbb{G}^\prime|}{2}$ and $\sup_{d\in\mathcal{D}}|\sumn\zeta_i^*(d)| = \max\{\mathbb{G}^*,\mathbb{G}^{*\prime} \}=\frac{\mathbb{G}^*+\mathbb{G}^{*\prime}+|\mathbb{G}^*-\mathbb{G}^{*\prime}|}{2}$. Thus by Lemma \ref{lem: gaussian approx}, 
\begin{equation}\label{eq: gaussian approx zeta}
    \begin{aligned}
 &\ \ \ \ n^{-1/2} \left|\sup_{d\in\mathcal{D}}|\sumn\zeta_i(d)| - \sup_{d\in\mathcal{D}}|\sumn\zeta_i^*(d)|\right| \\
 &= \left|\dfrac{\mathbb{G}+\mathbb{G}^\prime+|\mathbb{G}-\mathbb{G}^\prime|}{2\sqrt{n}} - \dfrac{\mathbb{G}^*+\mathbb{G}^{*\prime}+|\mathbb{G}^*-\mathbb{G}^{*\prime}|}{2\sqrt{n}}\right|  \\
 &\leq \dfrac{1}{2\sqrt{n}}\left(\left|\mathbb{G} - \mathbb{G}^*\right| + \left|\mathbb{G} - \mathbb{G}^*\right|+ \left||\mathbb{G}-\mathbb{G}^\prime| - |\mathbb{G}^*-\mathbb{G}^{*\prime}|\right|\right) \\
  &\lepg  \dfrac{1}{\sqrt{\log n}} + \dfrac{|\mathbb{G}  - \mathbb{G}^* -(\mathbb{G}^\prime - \mathbb{G}^{*\prime})|}{\sqrt{n}} \\
  &\leq  \dfrac{1}{\sqrt{\log n}} + \dfrac{|\mathbb{G}  - \mathbb{G}^*| +|\mathbb{G}^\prime - \mathbb{G}^{*\prime}|}{\sqrt{n}} \lepg \dfrac{1}{\sqrt{\log n}}.
\end{aligned}
\end{equation}
Recall that $\tilde{\mathbb{H}}(d) = \hat\sigma_\varepsilon n^{-1/2}\sumn\zeta_i(d)$, $\hat{\mathbb{H}}(d) = \hat\sigma_\varepsilon n^{-1/2}\sumn\zeta_i^*(d)$ and $\mathbb{H}^{(1)}(d) = \sigma_\varepsilon n^{-1/2}\sumn\zeta_i^*(d)$. By (R1) and (\ref{eq: gaussian approx zeta}),
\[\left|\sup_{d}|\hat{\mathbb{H}}(d)| - \sup_{d}|\tilde{\mathbb{H}}(d)|\right| \lepg \dfrac{1}{\hat\sigma_\varepsilon} n^{-1/2} \left|\sup_{d\in\mathcal{D}}|\sumn\zeta_i(d)| - \sup_{d\in\mathcal{D}}|\sumn\zeta_i^*(d)|\right| \lepg   \dfrac{1}{\sqrt{\log n}}.\]
If suffices to show 
\begin{equation}\label{eq: approx hat H H 1}
    \left|\sup_{d}|\hat{\mathbb{H}}(d)| - \sup_{d}|\mathbb{H}^{(1)}(d)|\right| \lepg \dfrac{1}{\sqrt{\log n}}.
\end{equation}
Note that 
\[\begin{aligned}
    \left|\sup_{d}|\hat{\mathbb{H}}(d)| - \sup_{d}|\mathbb{H}^{(1)}(d)|\right| &\leq \left|\dfrac{1}{\hat\sigma_\varepsilon} - \dfrac{1}{\sigma_\varepsilon}\right| \cdot \sigma_\varepsilon\sup_{d}\left|\mathbb{H}^{(1)}(d)\right| \lepg \dfrac{1}{(\logpM)^{1.5}} \cdot   \sup_{d}\left|\mathbb{H}^{(1)}(d)\right| 
\end{aligned} \]
where the last step applies (R1). If (\ref{eq: exp sup upper bound root log n}) holds, we can use the Markov inequality to deduce 
\[\sup_{d}\left|\mathbb{H}^{(1)}(d)\right| = o_{u.p.}((\log n)^{0.75})\]
and thus 
\[\left|\sup_{d}|\hat{\mathbb{H}}(d)| - \sup_{d}|\mathbb{H}^{(1)}(d)|\right| = o_{u.p.}\left(\frac{1}{(\log n)^{0.75}}\right).\]  
Then the proof of (R2) will end with the verification of (\ref{eq: exp sup upper bound root log n}). Let $d_j^\varDelta = a_D + j\varDelta_J$ for $j=0,1,\cdots,J$  with some integer $J$ large enough, where $\varDelta_J := \frac{b_D - a_D}{J}$. Then for any $d$, we can find some $d_j^\varDelta$ such that $|d_j^\varDelta - d_j| \leq \varDelta_J$ and hence 
\begin{equation}
    \label{eq: H1 bound 2 terms}
    \begin{aligned}
    |\mathbb{H}^{(1)}(d)| = |n^{-1/2}\mathbb{G}^*(d)| &\leq n^{-1/2}|\mathbb{G}^*(d)-\mathbb{G}^*(d^\varDelta_j)| + |\mathbb{G}^*(d^\varDelta_j)|\\ 
    &\leq \sup_{|d^\prime-d_2|\leq \varDelta_J} n^{-1/2}|\mathbb{G}^*(d)-\mathbb{G}^*(d^\varDelta_j)| + \max_{j\in[J]}|\mathbb{G}^*(d^\varDelta_j)|.
\end{aligned}
\end{equation}
Note that $\mathbb{G}^*(d^\varDelta_j)\sim N(0,1)$ for each $j\in[J]$, and thus 
\[ \E\left[\max_{j\in[J]} |\mathbb{G}^*(d^\varDelta_j)| \right] \lesssim \sqrt{\log J}. \] 
For any $d,d^\prime\in\mathcal{D}$, 
\begin{equation}\label{eq: G star d0d1 begin}
    \begin{aligned}
    n^{-1/2}|\mathbb{G}^*(d) - \mathbb{G}^*(d^\prime)| &\leq \|\hat s(d)^{-1}B^\prime(d) - \hat s(d^\prime)^{-1}B^\prime(d^\prime)\|_1 \cdot \sqrt{n}\|\hat F\hat\Omega_B^\top\|_\infty\cdot \supn|\varepsilon_i| \\
    &\lesssim \|\hat s(d)^{-1}B^\prime(d) - \hat s(d^\prime)^{-1}B^\prime(d^\prime)\|_1 \cdot \sqrt{n}M\sqrt{\logpM} \cdot \supn|\varepsilon_i| \\
\end{aligned}
\end{equation} 
where the last inequality applies 
the definition of $\hat\Omega_B$ when (\ref{eq: projection direction 1}) is feasible\footnote{We define $\hat\Omega_B = O$ and $\hat{\mathbb{H}}(d) = \mathbb{H}^{(1)}(d) = \tilde{\mathbb{H}}(d) = 0$ when (\ref{eq: projection direction 1}) is infeasible.}  with $\mu_j \asymp M\sqrt{\logpM / n}$. Additionally, by (\ref{eq: proof lower bound}) and (\ref{eq: m star lower bound}), we have almost surely that $$ \inf_g \inf_{d\in\mathcal{D}}\hat s(d) \gtrsim \inf_{d\in\mathcal{D}}\|B^\prime(d)\|_1^2 / \sqrt{B^\prime(d)^\top\hat\Sigma_B B^\prime(d)}\gtrsim \inf_{d\in\mathcal{D}}\dfrac{ \|B^\prime(d)\|_1^2 }{\sqrt{\|B^\prime(d)\|_1^2\cdot \|\hat\Sigma_B\|_\infty}} \gtrsim M$$
with $\hat\Sigma_B = n^{-1}\sumn B_\ic B_\ic^\top$ and thus $\|\hat\Sigma_B\|_\infty$ is uniformly bounded. In addition, by the second restriction in \eqref{eq: projection direction 1},
\begin{equation}\label{eq: Omega Sigma Omega}
    \supg \hat\Omega_B \hat\Sigma_F \hat\Omega_B^\top \leq   \supg  \|\hat F \hat\Omega_B^\top\|_\infty^2 \lesssim  M^2\logpM.
\end{equation} 
Also, $\|B^\prime(d)+B^\prime(d^\prime)\|_1 \lesssim M$ by Proposition \ref{prop: SplineDerivativeL2Norm}. Then 
 \[\begin{aligned}
      &\ \ \ \ \supg \|\hat s(d)^{-1}B^\prime(d) - \hat s(d^\prime)^{-1}B^\prime(d^\prime)\|_1 \\
      &\leq \|B^\prime(d)\|_1\cdot\supg \dfrac{\left|\hat s(d^\prime)^2 - \hat s(d)^2\right|}{\hat s(d)\hat s(d^\prime)[\hat s(d) + \hat s(d^\prime)]} + \dfrac{\|B^\prime(d) - B^\prime(d^\prime)\|_1}{\hat s(d^\prime)} \\
      &\lesssim \dfrac{\|B^\prime(d)\|_1\cdot\supg |(B^\prime(d) - B^\prime(d^\prime))^\top\hat\Omega_B \hat\Sigma_F \hat\Omega_B^\top (B^\prime(d) + B^\prime(d^\prime))| }{M^3} +\dfrac{\|B^\prime(d) - B^\prime(d^\prime)\|_1}{M}\\
      &\lesssim \left|\dfrac{\|B^\prime(d) - B^\prime(d^\prime)\|_1\cdot \supg \|\hat\Omega_B \hat\Sigma_F \hat\Omega_B^\top\|_\infty\cdot \|B^\prime(d) + B^\prime(d^\prime)\|_1}{M^{2}}\right| + \dfrac{\|B^\prime(d) - B^\prime(d^\prime)\|_1}{M} \\
      &\lesssim \|B^\prime(d) - B^\prime(d^\prime)\|_1\cdot M\logpM.     
   \end{aligned}\]
By Proposition \ref{prop: SplineDerivative}, there are at most a fixed number (denoted as $K_0$) nonzero elements in the vector $B^\prime(d) - B^\prime(d^\prime)$. Denote this active set as $\mathcal{K}_0$. Also, note that 
   \begin{equation}\label{eq: L1 bound Bd0 m Bd1}
       \begin{aligned}
       \|B^\prime(d) - B^\prime(d^\prime)\|_1 &= \sum_{j\in\mathcal{K}_0} |B_j^\prime(d) - B_j^\prime(d^\prime)| = \sum_{j\in \mathcal{K}_0}|B^{\prime\prime}(d_{2j})|\cdot|d - d^\prime| \\
      &\lesssim M^2\sum_{j\in[\mathcal{K}_0]}|d - d^\prime| = K_0M^2|d - d^\prime|
   \end{aligned}
   \end{equation} 
   where the inequality applies Proposition \ref{prop: SplineDerivativeL2Norm} for the upper bound of second-order derivative $B^{\prime\prime}$, and $d_{2j}$ is between $d$ and $d^\prime$. Thus, $$\supg \|\hat s(d)^{-1}B^\prime(d) - \hat s(d^\prime)^{-1}B^\prime(d^\prime)\|_1 \lesssim M^3\logpM |d - d^\prime|,$$ together with \eqref{eq: G star d0d1 begin} implying that 
   \[\begin{aligned}
    n^{-1/2}|\mathbb{G}^*(d) - \mathbb{G}^*(d^\prime)| \lesssim \sqrt{n}M^4(\logpM)^{1.5}\cdot |d - d^\prime|\cdot \supn|\varepsilon_i|. \\
\end{aligned}\] 
Thus, 
\[\begin{aligned}
    \supg \E\left(\sup_{|d-d^\prime|\leq\varDelta_J}  |n^{-1/2}[\mathbb{G}^*(d) - \mathbb{G}^*(d^\prime)]|\right) &\lesssim \sqrt{n}M^4(\logpM)^{1.5}\Delta_J \E(\supn|\varepsilon_i|) \\&\lesssim  \sqrt{n}M^4(\logpM)^2\Delta_J. 
\end{aligned}\]
Taking $J \asymp n^5$ , we have $\Delta_j \asymp n^{-5}$ and by (\ref{eq: H1 bound 2 terms})
\[\begin{aligned}
   \E\left( \sup_{d\in\mathcal{D}} |\mathbb{H}^{(1)}(d)| \right)  
    &\leq \E\left(\sup_{|d-d^\prime|\leq \varDelta_J} n^{-1/2}|\mathbb{G}^*(d)-\mathbb{G}^*(d^\varDelta_j)|\right) + \E\left(\max_{j\in[J]}|\mathbb{G}^*(d^\varDelta_j)|\right) \\
    &\lesssim \sqrt{n}M^4(\logpM)^2\cdot n^{-5} + \sqrt{5\log n} \lesssim \sqrt{\log n}. 
\end{aligned}
\]
%%%%%%%%%%%%%%%%%%%%%%%%%%%%%%%%%
\par {\bf \underline{Proof of (R4)}}. By (\ref{eq: approx hat H H 1}),
\begin{equation*} 
   \supg \Pr\left\{\left|\sup_{d}|\hat{\mathbb{H}}(d)| - \sup_{d}|\mathbb{H}^{(1)}(d)|\right| > \dfrac{C}{(\log n)^{0.75}}\right\} \leq \tau_{n}
\end{equation*}
with some $\tau_{n}\to 0$. Thus, 
\[\begin{aligned}
    \supg \Pr\left( \sup_{d}|\mathbb{H}^{(1)}(d)| > \hat c_n(\alpha ) + \dfrac{C}{(\log n)^{0.75}}\right) &\geq 
   \supg \Pr\left( \sup_{d}|\mathbb{H}^{(1)}(d)| > \hat c_n(\alpha ) \right)  - \tau_{n} \\
   &= 1- \alpha - \tau_{n}
\end{aligned}\]
which implies $\hat c_n(\alpha ) < c_n(\alpha + \tau_{n}) - \dfrac{C}{(\log n)^{0.75}}$. The proof completes by taking $\epsilon_n = \dfrac{C}{(\log n)^{0.75}}$.   

\subsection{Proof of Technical Lemmas}

\begin{proof}[Proof of Lemma \ref{lem: L2 varphi hat}]
Following the previous notations, we use $S_x$ to denote ${\rm diag}(\E(x_{ij}^2))_{j\in[p]}$ for any generic random vector $x_\ic = (x_{ij})_{j\in[p]}$, and $\tilde x_\ic := S_x^{-1/2} x_\ic$ to denote the standardized version of any random vector $x_\ic$. Note that 
\begin{equation}\label{eq: phi L2 norm begin}
    \begin{aligned}
    \|\hat\varphi - \varphi\|_2^2 &\lesssim (\hat\varphi - \varphi)^\top \E[X_\ic X_\ic^\top] (\hat\varphi - \varphi) + (\hat\kappa - \kappa)^\top \E[K_\ic K_\ic^\top] (\hat\kappa - \kappa) \\
    &= (\hat\varphi - \varphi)^\top S_X^{1/2}\E[\tilde X_\ic \tilde X_\ic^\top]S_X^{1/2} (\hat\varphi - \varphi) + (\hat\kappa - \kappa)^\top S_K^{1/2} \E[\tilde K_\ic \tilde K_\ic^\top]S_K^{1/2} (\hat\kappa - \kappa) \\
     &\lesssim (\hat\varphi - \varphi)^\top S_X^{1/2} S_X^{1/2} (\hat\varphi - \varphi) + (\hat\kappa - \kappa)^\top S_K^{1/2}  S_K^{1/2} (\hat\kappa - \kappa) \\
    &\lesssim (\hat\varphi - \varphi)^\top S_X^{1/2} \E[\tilde X_\ic \tilde X_\ic^\top] S_X^{1/2}(\hat\varphi - \varphi) + (\hat\kappa - \kappa)^\top S_K^{1/2} \E[\tilde K_\ic \tilde K_\ic^\top] S_K (\hat\kappa - \kappa) + \\
    &\ \ \ \  2(\hat\kappa - \kappa)^\top S_K^{1/2} \E[\tilde K_\ic \tilde X_\ic^\top] S_X^{1/2} (\hat\varphi - \varphi)\\
    &= (\hat\varphi - \varphi)^\top \E[ X_\ic  X_\ic^\top] (\hat\varphi - \varphi) + (\hat\kappa - \kappa)^\top \E[ K_\ic  K_\ic^\top] (\hat\kappa - \kappa) + \\
    &\ \ \ \  2(\hat\kappa - \kappa)^\top \E[ K_\ic  X_\ic^\top] (\hat\varphi - \varphi)\\
\end{aligned}
\end{equation} 
where the third and the fourth steps apply Assumption \ref{assu: eigen}. 
By (\ref{eq: key basic inequality phi}) and (\ref{eq: special case 1 phi}), we deduce that 
\begin{equation}\label{eq: phi L2 error bound special case 1 proof}
     n^{-1}\|K(\hat\kappa-\kappa) + X(\varphi - \hat\varphi)\|_2^2 \lesssim R_{D1}^2.
\end{equation} 
and also note the following decomposition 
\[\begin{aligned}
     n^{-1}\|K(\hat\kappa-\kappa) + X(\varphi - \hat\varphi)\|_2^2 &= (\hat\varphi - \varphi)^\top \dfrac{\sumn X_\ic X_\ic^\top}{n} (\hat\varphi - \varphi) + (\hat\kappa - \kappa)^\top  \dfrac{\sumn K_\ic K_\ic^\top}{n} (\hat\kappa - \kappa) + \\
    &\ \ \ \  2(\hat\kappa - \kappa)^\top \dfrac{\sumn K_\ic X_\ic^\top}{n} (\hat\varphi - \varphi). 
\end{aligned} \]
By Proposition \ref{prop: DB}, with probability at least $1 - c(pM)^{-4}$
\[\|\dfrac{1}{n}\sumn\left[K_\ic X_\ic^\top - \E(K_\ic X_\ic^\top)  \right]\|_\infty + \|\dfrac{1}{n}\sumn\left[X_\ic X_\ic^\top - \E(X_\ic X_\ic^\top)  \right]\|_\infty \lesssim \sqrt{\dfrac{\log (pM)}{n}},\]
\[n^{-1}\left\|\sumn [K_\ic K_\ic^\top - \E(K_\ic K_\ic^\top)] \right\|_\infty \lesssim \sqrt{\dfrac{\log (pM)}{n}}.\]
Together with (\ref{eq: varphi hat L1 special case 1}) and (\ref{eq: kappa hat L2 special case 1}), we deduce that with probability at least $1 - c(pM)^{-4}$
\[\begin{aligned}
    \varDelta^D_1 := (\hat\varphi - \varphi)^\top \dfrac{1}{n}\sumn\left[X_\ic X_\ic^\top - \E(X_\ic X_\ic^\top)  \right] (\hat\varphi - \varphi) &\leq \|\hat\varphi - \varphi\|_1^2 \cdot \|\dfrac{1}{n}\sumn\left[X_\ic X_\ic^\top - \E(X_\ic X_\ic^\top)  \right]\|_\infty\\
    &\lesssim \dfrac{R_{D1}^4}{\lambda_D^2} \sqrt{\dfrac{\log (pM)}{n}} \lesssim \sqrt{\dfrac{n}{\log (pM)}}R_{D1}^4, 
\end{aligned}\]
\[\begin{aligned}
    \varDelta^D_2 := (\hat\kappa - \kappa)^\top \dfrac{1}{n}\sumn\left[X_\ic X_\ic^\top - \E(X_\ic X_\ic^\top)  \right] (\hat\kappa - \kappa) &\leq \|\hat\kappa - \kappa\|_1^2 \cdot \|\dfrac{1}{n}\sumn\left[K_\ic K_\ic^\top - \E(K_\ic K_\ic^\top)  \right]\|_\infty \\
    &\lesssim \dfrac{MR_{D1}^4}{R_{D2}^2} \sqrt{\dfrac{\log (pM)}{n}} \lesssim \sqrt{\dfrac{nM^2}{\log (pM)}}R_{D1}^4, 
\end{aligned}\]
and 
\[\begin{aligned}
    \varDelta^D_3 &:= 2(\hat\kappa - \kappa)^\top \dfrac{1}{n}\sumn\left[K_\ic X_\ic^\top - \E(K_\ic X_\ic^\top)  \right] (\hat\varphi - \varphi) \\
    &\leq \sqrt{p_zM} \|\hat\kappa - \kappa\|_2 \cdot \|\hat\varphi - \varphi\|_1  \cdot \|\dfrac{1}{n}\sumn\left[K_\ic X_\ic^\top - \E(K_\ic X_\ic^\top)  \right]\|_\infty\\
    &\lesssim \dfrac{\sqrt{M} R_{D1}^4}{R_{D2}\lambda_D} \sqrt{\dfrac{\log (pM)}{n}} \lesssim \sqrt{\dfrac{nM}{\log (pM)}}R_{D1}^4.
\end{aligned}\]
Recall that $R_{D1} \asymp M^{-\gamma}$ and thus 
\[\sqrt{\dfrac{nM^2}{\log (pM)}}R_{D1}^2 \lesssim  o[\sqrt{n}\cdot M^{-2\gamma + 1}] = o[M^{\frac{2\gamma+1}{2}-2\gamma+1}] = o(M^{-\gamma + 1.5}) = o(1)\]
as $\gamma \geq 2$, which implies $\varDelta^D_1 + \varDelta^D_2 + \varDelta^D_3 = o(R_{D1}^2)$. Consequently, by (\ref{eq: phi L2 norm begin}) and (\ref{eq: phi L2 error bound special case 1 proof}), 
\[\begin{aligned}
    \|\hat\varphi - \varphi\|_2^2 &= n^{-1}\|K(\hat\kappa - \kappa) + X(\hat\varphi-\varphi)\|_2^2 + \varDelta^D_1 + \varDelta^D_2 + \varDelta^D_3  \\
    &\lesssim R_{D1}^2 + o(R_{D1}^2) \lesssim R_{D1}^2 \\
\end{aligned}\]
with probability at least $1-c(pM)^{-4}$.
\end{proof}

 \begin{proof}[Proof of Lemma \ref{lem: RE concentration phi}]
 \par Let $K_\ell$ denote the matrix with the $(i,j)$-th element $K_{ij}=K(v_i)$. Define $Q=(K,X)$ and $\delta = \left(\begin{array}{c}
    \hatkappa -\kappa    \\
      \hat\varphi-\varphi
 \end{array}\right)$. Then we decompose the target quadratic form as 
 \begin{equation}
 \begin{aligned}
 \delta^\top \dfrac{Q^\top Q}{n} \delta  
 &= \delta^\top \dfrac{\E[Q^\top Q]}{n} \delta +   \Delta_0  
 \end{aligned}
 \end{equation}
 where 
 \[\Delta_0 =   \delta^\top\dfrac{Q^\top Q-\E[Q^\top Q]}{n} \delta.\]
  The remaining proofs are composed of two steps.

 \par {\bf Step 1: Find a lower bound of $\delta^\top \dfrac{\E[Q^\top Q]}{n} \delta$.}
 Note that by Assumption \ref{assu: eigen} we have 
  \begin{equation}
  \begin{aligned}
  \delta^\top \dfrac{\E[Q^\top Q]}{n} \delta &=   \delta^\top S_Q^{1/2}\dfrac{\E[\tilde Q^\top \tilde Q]}{n} S_Q^{1/2}\delta \\
  &\gtrsim (\hatkappa-\kappa)^\top S_K^{1/2}\dfrac{\E[\tilde K^\top \tilde K]}{n} S_K^{1/2}(\hatkappa-\kappa)   + 
 (\hatvarphi-\varphi)^\top S_X^{1/2}\dfrac{\E[\tilde X^\top \tilde X]}{n} S_X^{1/2} (\hatvarphi-\varphi) \\ 
   &= (\hatkappa-\kappa)^\top \dfrac{\E[ K^\top K]}{n} (\hatkappa-\kappa)   + 
 (\hatvarphi-\varphi)^\top \dfrac{\E[ X^\top  X]}{n} (\hatvarphi-\varphi) \\ 
&\gtrsim M^{-1}\|\hat\kappa - \kappa\|_2^2 + \|\hat\varphi - \varphi\|_2^2.
  \end{aligned}
    \label{eq: decompose target quadratic} 
\end{equation}

 \par {\bf Step 2:  Bound $\Delta_0$.} Note that each row in the $n\times(p_zM + p)$ matrix $Q$ are independent bounded variables. By Bernstein-type inequality, with probabilty at least $1-c(M\vee p)^{-4}$  
 \[
 \left\|\dfrac{Q^\top Q}{n}-\dfrac{\mathbb{E}(Q^\top Q)}{n}\right\|_\infty \lesssim \sqrt{\dfrac{\log(M\vee p)}{n}}
 \]
 and hence 
 \begin{equation*}
     |\Delta_0| \lep \|\delta\|_1^2\sqrt{\dfrac{\log(M\vee p)}{n}}. 
 \end{equation*}
 Note that 
 \[\begin{aligned}
  \|\delta\|_1^2  &=   \|\hat\kappa - \kappa\|_1^2 + \|(\hat\varphi-\varphi)_\mathcal{S}\|_1^2 + \|(\hat\varphi-\varphi)_{\mathcal{S}^c}\|_1^2 \\
  &\leq  M\|\hat\kappa - \kappa\|_2^2 +  \|(\hat\varphi-\varphi)_\mathcal{S}\|_1^2 + 2\dfrac{c_2^2R_{D2}^2}{\lambda_D^2}\|\hat\kappa - \kappa\|_2^2 + 2c_3^2\|(\hat\varphi-\varphi)_\mathcal{S}\|_1^2 \\
  &= \left(2M + 2\dfrac{c_2^2R_{D2}^2}{\lambda_D^2}\right)\|\hat\kappa - \kappa\|_2^2 + (1+2c_3^2)s\|\hat\varphi-\varphi\|_2^2 \\
    &\lesssim M\|\hat\kappa - \kappa\|_2^2 +  s\|\hat\varphi-\varphi\|_2^2 
 \end{aligned}
 \]
 which implies 
 \begin{equation}
 \begin{aligned}
     |\Delta_0|  
    & \lep M\sqrt{\dfrac{\logpM}{n}}\|\hat\kappa-\kappa\|_2^2 +  s\sqrt{\dfrac{\logpM}{n}}\|\hat\varphi-\varphi\|_2^2 \\
    &= o_{\rm{a.s.}}\left(M^{-1}\|\hat\kappa-\kappa\|_2^2 + \|\hat\varphi-\varphi\|_2^2\right)
    \label{eq: Delta2 Bound D}. 
 \end{aligned}
 \end{equation}
\end{proof}

 \begin{proof}[Proof of Lemma \ref{lem: RE concentration}]Note that the proof of Lemma \ref{lem: RE concentration} depends on Lemma \ref{lem: RE concentration phi} since the LASSO algorithm \eqref{eq: partial penalization} depends on $\hat v_i$ from \eqref{eq: partial penalization for D}. 
 \par Let $H$ denote the matrix with the $(i,j)$-th element $H_{ij}=H_j(v_i)$. Define $W=(B,H)$, $\widehat{U}=(\widehat{W},X)$, $U=(W,X)$ and $\delta = \left(\begin{array}{c}
    \hatomega-\omega    \\
      \hattheta-\theta
 \end{array}\right)$. Then we decompose the target quadratic form as 
 \begin{equation}
 \begin{aligned}
 \delta^\top \dfrac{\hatU^\top \hatU}{n} \delta &= \delta^\top \dfrac{\E[U^\top U]}{n} \delta + \delta^\top \dfrac{(U-\hatU)^\top (U-\hatU)}{n} \delta + 2\delta^\top \dfrac{(U-\hatU)^\top  U  }{n} \delta  \\
 &\ \ \ + \delta^\top\dfrac{U^\top U-\E[U^\top U]}{n} \delta \\
 &\geq \delta^\top \dfrac{\E[U^\top U]}{n} \delta +   \Delta_1  + \Delta_2
 \end{aligned}
 \end{equation}
 where 
 \[\Delta_1 = 2\delta^\top \dfrac{(U-\hatU)^\top  U  }{n} \delta,\ \  \Delta_2 = \delta^\top\dfrac{U^\top U-\E[U^\top U]}{n} \delta.\]
  The remaining proofs are composed of three steps.

 \par {\bf Step 1: Find a lower bound of $\delta^\top \dfrac{\E[U^\top U]}{n} \delta$.}
 Note that following (\ref{eq: decompose target quadratic}), we have  by Assumption \ref{assu: eigen}
  \begin{equation}
  \begin{aligned}
  \delta^\top \dfrac{\E[U^\top U]}{n} \delta &\geq  
c^*(\hatbeta-\beta)^\top \dfrac{\E[B^\top B]}{n}(\hatbeta-\beta) +
c^*(\hateta-\eta)^\top \dfrac{\E[H^\top H]}{n}(\hateta-\eta)  \\ &\ \ \  + 
c^*(\hattheta-\theta)^\top \dfrac{\E[X^\top X]}{n}(\hattheta-\theta) \\ 
&\geq  
 c^*c_B M^{-1}\|\hatbeta-\beta\|_2^2  +
 c^*c_H M^{-1}\|\hateta-\eta\|_2^2   + 
c^*c_\Sigma\|\hattheta-\theta\|_2^2  \\
&\geq  
 c^*(c_B\wedge c_H) M^{-1}\|\hatomega-\omega\|_2^2  +
c^*c_\Sigma\|\hattheta-\theta\|_2^2. 
  \end{aligned}
    \label{eq: decompose target quadratic 2} 
\end{equation}

\par {\bf Step 2: Bound $\Delta_1$.} Note that
\begin{equation}
   \Delta_1 = \dfrac{2}{n}(\hatomega - \omega)^\top(W-\hatW)^\top W(\hatomega - \omega) + \dfrac{2}{n}(\hatomega - \omega)^\top(W-\hatW)^\top X(\hattheta-\theta).
    \label{eq: decompose quadratic error} 
\end{equation} 
Define $a_n :=  M^{-\gamma + 1} + \sqrt{\dfrac{M^2(s+M)\logpM}{n}}$. 
By Corollary \ref{cor: H approx},  with probability at least $1-c(pM)^{-4}$
 \begin{equation}
 \begin{aligned}
   \|W-\hatW\|_2  &\leq \sqrt{\sum_{i=1}^n\sum_{j=1}^{M}[H_j(v_i)-H_j(\hatv_i)]^2} \\
   &\lesssim \sqrt{n}a_n.
 \end{aligned}
 \label{eq: diff W}
 \end{equation}
Besides, $ \|W\|_2 \leq \|B\|_2 + \|H\|_2$ and 
 \begin{equation}
 \begin{aligned}
     \|B\|_2  = \sqrt{\|B^\top B\|_2}\leq \sqrt{\|B^\top B - \E(B^\top B)\|_2} + \sqrt{\|\E(B^\top B)\|_2}.\\ 
 \end{aligned}
 \label{eq: Bnorm2 bound 0}
 \end{equation}

 Since $\|B_{i\cdot} B_{i\cdot}^\top\|_2\leq \|B_{i\cdot}\|_2^2\leq k_D$ and 
 \[ 
 \begin{aligned}
  \|\sum_{i=1}^n \mathbb{E}[ B_{i\cdot} B_{i\cdot}^\top B_{i\cdot} B_{i\cdot}^\top ] \|_2 &= \| \sum_{i=1}^n \mathbb{E}[ \|B_{i\cdot}\|_2^2 B_{i\cdot}  B_{i\cdot}^\top  ]\|_2  \\
   &= k_D \| \sum_{i=1}^n \mathbb{E}[ B_{i\cdot}  B_{i\cdot}^\top ]\|_2 \leq k_D n M^{-1},  
 \end{aligned}
 \]
 Then applying the matrix Bernstein inequality \citep[Theorem 1.6.2]{tropp2015introduction}, we have for any $t>0$
 \begin{equation*}
     \PP\left(\left\|\dfrac{B^\top B}{n} - \E(\dfrac{B^\top B}{n})\right\|_2 > \dfrac{t}{n} \right) \leq 2M\cdot \exp\left(-\dfrac{t^2}{2k_D nM^{-1}+k_D t}\right).
 \end{equation*}
 Taking $t=\sqrt{12k_D nM^{-1}\log (M\vee p)}$ we have with probability at least $1-c(M\vee p)^{-4}$
  \begin{equation}
      \sqrt{\left\|B^\top B - \E(B^\top B)\right\|_2} \lesssim  (nM^{-1}\log (M\vee p))^{1/4}.
      \label{eq: BB m EBB}
 \end{equation}
Together with the fact that $\sqrt{\|\E(B^\top B)\|_2}\lesssim \sqrt{nM^{-1}}$ by Proposition \ref{prop: spline eigen}, \eqref{eq: Bnorm2 bound 0} and \eqref{eq: BB m EBB} imply $\|B\|_2  \lesssim \sqrt{nM^{-1}\log (pM)}$. Following the same procedures we have $ \|H\|_2 \lesssim \sqrt{nM^{-1}\log (pM)} $ and hence 
  \begin{equation}
 \begin{aligned}
     \|W\|_2 \leq \|B\|_2 + \|H\|_2   \lesssim \sqrt{nM^{-1}\log  (pM)}. 
 \end{aligned}
 \label{eq: Wnorm2 bound}
 \end{equation}
\eqref{eq: diff W} and \eqref{eq: Wnorm2 bound} imply that with probability at least $1-c(M\vee p)^{-4}$
 \begin{equation}
 \begin{aligned}
      \left|\dfrac{2}{n}(\hatomega-\omega)^\top(W-\hatW)^\top W (\hatomega-\omega)\right| &\leq \dfrac{2}{n}\|\hatomega-\omega\|_2^2 \|W-\hatW\|_2\|W\|_2 \\
      &\lesssim \|\hatomega-\omega\|_2^2 a_n \sqrt{M^{-1}\log (pM)}. 
 \end{aligned}
 \label{eq: Quadratic Bound 1}
 \end{equation}
 In addition, \eqref{eq: diff W} also implies 
 \begin{equation}
     \begin{aligned}
&\ \ \ \ \left| \dfrac{2}{n}(\hatomega-\omega)^\top(W-\hatW)^\top X (\hattheta-\theta) \right| \\
 &\leq 2\sqrt{(\hatomega-\omega)^\top\dfrac{(W-\hatW)^\top(W-\hatW)}{n}(\hatomega-\omega)}\sqrt{(\hattheta-\theta)^\top\dfrac{X^\top X}{n}(\hattheta-\theta)} \\
  &\leq  (\hatomega-\omega)^\top\dfrac{4C_\Sigma(W-\hatW)^\top(W-\hatW)}{c^*c_\Sigma n}(\hatomega-\omega) +  (\hattheta-\theta)^\top\dfrac{c^* c_\Sigma X^\top X}{4C_\Sigma n}(\hattheta-\theta) \\
   &\leq  \|\hatomega-\omega\|^2_2\dfrac{4C_\Sigma\|W-\hatW\|^2_2}{c^*c_\Sigma n}  +  (\hattheta-\theta)^\top\dfrac{c^*c_\Sigma X^\top X}{4C_\Sigma n}(\hattheta-\theta) \\
   &\leq  C\|\hatomega-\omega\|^2_2 a_n^2  +  (\hattheta-\theta)^\top\dfrac{c^*c_\Sigma X^\top X}{4C_\Sigma n}(\hattheta-\theta) \\
 \end{aligned}
\label{eq: Quadratic Bound 2}
 \end{equation}
 for some $C>0$. Besides, by Proposition \ref{prop: DB} we have with probability at least $1-(M\vee p)^{-4}$
 \[ \|\dfrac{X^\top X}{n} - \dfrac{\E(X^\top X)}{n}\|_\infty \leq C_x \sqrt{\dfrac{\log (pM)}{n}} \]
 for some $C_x>0$ large enough. Note that $((\hatomega-\omega),(\hattheta-\theta)^\top)^\top \in \mathcal{C}_0$ as defined in \eqref{eq: Set C0 }, 
 \begin{equation}
 \begin{aligned}\label{eq: upper bound theta X}
 (\hattheta-\theta)^\top\dfrac{X^\top X}{n}(\hattheta-\theta) &\leq (\hattheta-\theta)^\top\dfrac{\E(X^\top X)}{n}(\hattheta-\theta) + \|\hattheta-\theta\|_1^2\|\dfrac{X^\top X}{n} - \dfrac{\E(X^\top X)}{n}\|_\infty \\
 &\leq C_\Sigma \|\hattheta-\theta\|_2^2   + \left(\dfrac{2c_2^2R_2^2}{\lambda_Y^2}\|\hatomega-\omega\|_2^2 + 2c_3^2s\|\hattheta-\theta\|_2^2\right) \cdot C_x\sqrt{\dfrac{\log p}{n}} \\
  &\leq 2C_\Sigma \|\hattheta-\theta\|_2^2   +  \dfrac{2c_2^2R_2^2}{\lambda_Y^2}C_x\sqrt{\dfrac{\log p}{n}}\|\hatomega-\omega\|_2^2.    \\
 \end{aligned}
 \end{equation}
 Combining \eqref{eq: decompose quadratic error}, \eqref{eq: Quadratic Bound 1}, \eqref{eq: Quadratic Bound 2} and \eqref{eq: upper bound theta X} we deduce that with probability at least $1-c(M\vee p)^{-4}$
 \begin{equation}
 \begin{aligned}
   |\Delta_1| &\lesssim C\|\hatomega-\omega\|_2^2(a_n^2 + a_n\sqrt{M^{-1}\log (pM)}) +  (\hattheta-\theta)^\top\dfrac{c^*c_\Sigma X^\top X}{4C_\Sigma n}(\hattheta-\theta) \\
   &\leq C\|\hatomega-\omega\|_2^2(a_n^2 + a_n\sqrt{M^{-1}\log (pM)} + \dfrac{R_2^2}{\lambda_Y^2}\sqrt{\dfrac{\log p}{n}}) +   \dfrac{c^*c_\Sigma \|\hattheta-\theta\|_2^2}{2}  \\
   &= o_p\left(M^{-1}\|\hatomega-\omega\|_2^2\right) +  \dfrac{c^*c_\Sigma \|\hattheta-\theta\|_2^2}{2}.
 \end{aligned}
    \label{eq: decompose quadratic bound} 
\end{equation} 
 
 \par {\bf Step 3:  Bound $\Delta_2$.} Note that each row in the $n\times(2M+p)$ matrix $U$ are independent bounded variables. By Bernstein-type inequality, with probabilty at least $1-c(M\vee p)^{-4}$  
 \[
 \left\|\dfrac{U^\top U}{n}-\dfrac{\mathbb{E}(U^\top U)}{n}\right\|_\infty \lesssim \sqrt{\dfrac{\log(M\vee p)}{n}}
 \]
 and hence 
 \begin{equation*}
     |\Delta_2| \lesssim \|\delta\|_1^2\sqrt{\dfrac{\log(M\vee p)}{n}}. 
 \end{equation*}
 Note that 
 \[\begin{aligned}
  \|\delta\|_1^2  &=   \|\hatomega-\omega\|_1^2 + \|(\hattheta-\theta)_\mathcal{S}\|_1^2 + \|(\hattheta-\theta)_\mathcal{N}\|_1^2 \\
  &\leq  2M\|\hatomega-\omega\|_2^2 +  \|(\hattheta-\theta)_\mathcal{S}\|_1^2 + 2\dfrac{c_2^2R_2^2}{\lambda_Y^2}\|\hatomega-\omega\|_2^2 + 2c_3^2\|(\hattheta-\theta)_\mathcal{S}\|_1^2 \\
  &= \left(2M + 2\dfrac{c_2^2R_2^2}{\lambda_Y^2}\right)\|\hatomega-\omega\|_2^2 + (1+2c_3^2)s\|\hattheta-\theta\|_2^2 \\
    &\lesssim  M \|\hatomega-\omega\|_2^2 +  s\|\hattheta-\theta\|_2^2 
 \end{aligned}
 \]
 which implies 
 \begin{equation}
 \begin{aligned}
     |\Delta_2|  
    &\lesssim M \sqrt{\dfrac{\logpM}{n}}\|\hat\omega-\omega\|_2^2 +  s\sqrt{\dfrac{\logpM}{n}}\|\hat\theta-\theta\|_2^2 \\
    &= o_{\rm{a.s.}}\left(M^{-1}\|\hat\omega-\omega\|_2^2 + \|\hat\theta-\theta\|_2^2\right)
    \label{eq: Delta2 Bound}. 
 \end{aligned}
 \end{equation}
\end{proof}

\begin{proof}[Proof of Lemma \ref{lem: cond cov eigen}]
Define $\Sigma_F = \E\left(F_\ic F_\ic^\top \right)$, $S_F := \text{diag}(\Sigma_F)$ and recall that $F^{\rm std}_\ic = S_F^{-1/2}F_\ic$ in  Assumption \ref{assu: more compatibility}. The eigenvalues of $\E[S_F^{-1/2}\Sigma_{F}S_F^{-1/2}] = \E[F^{\rm std}_\ic (F^{\rm std}_\ic)^\top]$ are bounded away from zero and above by Assumption \ref{assu: more compatibility}. It then suffices to show that 
\begin{equation} \label{eq: target bound eigen}
   \supg \left\|S_F^{-1/2}\left[\Sigma_F - \Sigma_{F|\mathcal{L}}\right] S_F^{-1/2}  \right\|_2 = o_p(1) 
\end{equation}
which implies that the eigenvalues of $S_F^{-1/2}\Sigma_{F|\mathcal{L}}S_F^{-1/2}$ are bounded away from zero and above. As $1 \lep \lambda_{\min}(S_F^{-1/2}) \leq \lambda_{\max}(S_F^{-1/2}) \lep M$, we deduce that $M^{-1} \lep \inf_g \lambda_{\min}(\Sigma_{F|\mathcal{L}})\leq \sup_g \lambda_{\max}(\Sigma_{F|\mathcal{L}}) \lep 1$ by (\ref{eq: target bound eigen}). 
\par We then start proving (\ref{eq: target bound eigen}). By standard arguments, 
\[\begin{aligned}
    \|S_F^{-1/2}[\Sigma_F - \Sigma_{F|\mathcal{L}}]S_F^{-1/2}  \|_2 &=\left\|S_F^{-1/2}\E_{\mathcal{L}}(F_{i\cdot}^{\ot}-\hat F_{i\cdot}^{\ot})S_F^{-1/2}\right\|_2 \\ &\lesssim \left\|S_F^{-1/2}\E_{\mathcal{L}}(F_{i\cdot}-\hat F_{i\cdot})^{\ot}S_F^{-1/2}\right\|_2 + \left\|\E_{\mathcal{L}}S_F^{-1/2}(F_{i\cdot}-\hat F_{i\cdot})F_{i\cdot}^\top S_F^{-1/2}\right\|_2\\
    &=: L_1 + L_2.
\end{aligned}\]
\underline{We first bound $L_1$}. Given that $\max_{j\in[M_D]}\E(H_{ij}^2) \lesssim M$, $\max_{j\in[p_zM]}\E[(q^\prime(v_i)K_{ij})^2] \lesssim M$ and $\max_{j\in[p]}\E[(q^\prime(v_i)X_{ij})^2] \lesssim 1$, we deduce that 
\[\begin{aligned}
    L_1 &\lesssim \max_{j\in[M_D]}\E(H_{ij}^2)^{-1}\left\|\E_{\mathcal{L}}(H_{i\cdot}-\hat H_{i\cdot})^{\ot}\right\|_2 +  \max_{j\in[p]}\E[(q^\prime(v_i)X_{ij})^2]^{-1}\left\|\E_{\mathcal{L}}(\hat q_i^\prime - q_i^\prime)^2X_\ic^{\ot}\right\|_2  + \\ 
&\ \ \ \ \max_{j\in[p_zM]}\E[(q^\prime(v_i)K_{ij})^2]^{-1}\left\|\E_{\mathcal{L}}(\hat q_i^\prime - q_i^\prime)^2K_\ic^{\ot}\right\|_2 \\
&\lesssim M\left\|\E_{\mathcal{L}}(H_{i\cdot}-\hat H_{i\cdot})^{\ot}\right\|_2 +\left\|\E_{\mathcal{L}}(\hat q_i^\prime - q_i^\prime)^2X_\ic^{\ot}\right\|_2 + M\left\|\E_{\mathcal{L}}(\hat q_i^\prime - q_i^\prime)^2K_\ic^{\ot}\right\|_2 \\
&=: L_{11} + L_{12} + L_{13}. 
\end{aligned}\]
\par \underline{Bound $L_{11}$}. By (\ref{eq: cond exp vhat H}), 
\[\begin{aligned}
    L_{11} = \left\|n^{-1}\sumn M\E_{\mathcal{L}}(H_{i\cdot}-\hat H_{i\cdot})^{\ot}\right\|_2 \leq  n^{-1}M\sumn\sum_{j=1}^{M} \E_{\mathcal{L}}(H_{ij}-\hat H_{ij})^2 = o_p(1).  
\end{aligned}\]
Since the $H$ and $\hat H$ are not dependent on $g$, the $o_p(1)$ holds uniformly for all $g$. 

\par \underline{Bound $L_{12}$}. Note that
\begin{equation}\label{eq: bound q hat approx}
    \begin{aligned}
    (\hat q_i^\prime - q_i^\prime)^2 
   &\lesssim \left|q^\prime(\hat v_i) - q^\prime(v_i)\right|^2 + \left|\eta^\top \hat H_{i\cdot}^\prime - q^\prime(\hat v_i)\right|^2  + \left|(\hat\eta^{\ind} - \eta)^\top \hat H_{i\cdot}^\prime\right|^2  \\
   &=: \varDelta^q_{1i} + \varDelta^q_{2i} + \varDelta^q_{3i}. 
\end{aligned}
\end{equation}
Define $\varDelta^q_{j}:=\supn \varDelta^q_{ji}$ for $j=1,2,3.$ Eq. (\ref{eq: sup q approx}) shows when $\hat{\mathcal{V}}_n$ in \eqref{eq: hat v event} holds, 
\begin{equation}\label{eq: Delta q1 bound}
    \varDelta^q_1 \lesssim c_{1n} := \dfrac{n}{\log (pM)}M^{-4\gamma} + \dfrac{(s+M)^2\log (pM)}{n} + M^{-2\gamma} = o(1). 
\end{equation}
When $\hat{\mathcal{V}}_n$ in (\ref{eq: hat v event}) does not hold, we claim $\varDelta^q_1$ is uniformly bounded since $q^\prime(\cdot)$ is bounded over a compact interval. Thus 
\begin{equation}\label{eq: Vn not hold E bound}
    \begin{aligned} 
    \sup_g\left\|\E_{\mathcal{L}}\left[  \varDelta^q_1   X_\ic^\ot\right]\right\|_2  
    &\leq \sup_g\left\|\E_{\mathcal{L}}\left[ 1(\hat{\mathcal{V}}_n) \varDelta^q_1   X_\ic^\ot\right]\right\|_2 + \sup_g\left\|\E_{\mathcal{L}}\left[ 1(\hat{\mathcal{V}}_n^c) \varDelta^q_1 X_\ic^\ot\right]\right\|_2 \\
    &\lesssim c_{1n}\cdot \|\E(X_\ic X_\ic^\top)\|_2 + \Pr(\hat{\mathcal{V}}_n|\mathcal{L}) \cdot \E( \|X_\ic\|_2^2) \\
    &\lesssim c_{1n} + \Pr(\hat{\mathcal{V}}_n|\mathcal{L})\cdot p = o_p(1)
\end{aligned}
\end{equation} 
where the third inequality applies that boundness of $X_\ic$ and the $o_p(1)$ applies \eqref{eq: prob vhat compact}. 
\par Furthermore, by Proposition \ref{prop: spline approx power} 
\begin{equation}\label{eq: Delta q2 bound}
    \Delta^q_2 = \sup_{v\in[a_v-\epsilon_v,b_v+\epsilon_v]}\left|\sum_{j\in[M]}\eta_j H_j^\prime(v) -q^\prime(v)\right|^2  \lesssim M^{-2\gamma+2}.
\end{equation}
 By  Proposition \ref{prop: SplineDerivativeL2Norm}  and the functional extensions stated in the beginning of Section \ref{sec: proof},
 \begin{equation}\label{eq: Delta q3 bound}
     \varDelta^q_3 \leq \supn \|\hat\eta^\ind - \eta\|_2^2\cdot \|\hat H^\prime_\ic\|_2^2  \lesssim \|\hat\eta^\ind - \eta\|_2^2\cdot M^2. 
 \end{equation}
Consequently,
\begin{equation}\label{eq: Vn holds E bound}
    \begin{aligned}
    \sup_g\left\|\E_{\mathcal{L}}\left[  (  \varDelta^q_2 + \varDelta^q_3 ) X_\ic^\ot\right]\right\|_2 
    &\lesssim \left( M^{-2\gamma+2}+\sup_g\|\hat\eta^\ind - \eta\|_2^2 M^2\right) \cdot \lambda_{\max}(\E[X_\ic X_\ic^\top]) = o_p(1).
\end{aligned}
\end{equation} 
The $o_p(1)$ applies the uniform error bound for $\hat\eta^\ind$ in \eqref{eq: sup rate omega}.
Therefore,  
\[ \sup_g L_{12} = \sup_g \left\|\E_{\mathcal{L}}(\hat q_i^\prime - q_i^\prime)^2X_\ic^{\ot}\right\|_2 = \sup_g \left\|\E_{\mathcal{L}}(\varDelta^q_1 + \varDelta^q_2 + \varDelta^q_3)X_\ic^{\ot}\right\|_2 = o_p(1).\]
\par \underline{Bound $L_{13}$}. Note that $\|M\E(K_\ic K_\ic^\top)\|_2 \lesssim 1$ by Proposition \ref{prop: spline eigen}, and
\begin{equation}\label{eq: Bound K}
     M \|K_\ic\|_2^2 \leq M\sup_{\{z_\ell\}}\sum_{\ell\in[p_z]}\sum_{j\in[M]} [K_{j\ell}(z_\ell)]^2 \leq M\sup_{\{z_\ell\}}\sum_{\ell\in[p_z]}\left(\sum_{j\in[M]} |K_{j\ell}(z_\ell)|\right)^2 = Mp_z.
\end{equation}
Using these two results for $K_\ic$, we can deduce 
\[\begin{aligned} 
   \supg L_{13} &\lesssim \supg M\left\|\E_{\mathcal{L}}\left[ (\varDelta^q_1 + \varDelta^q_2 + \varDelta^q_3 ) K_\ic^\ot\right]\right\|_2 = o_p(1).
\end{aligned}\]
following similar arguments when bounding $L_{12}$,  
\par \underline{We then bound $L_2$}. It follows that 
\[\begin{aligned}
   \sup_g L_2 &\lesssim \sup_g\sqrt{\left\|\E\left[S_F^{-1/2}(F_\ic-\hat F_\ic)^\ot S_F^{-1/2}\right]\right\|_2} \cdot \sqrt{\|\E(F^{\rm std}_\ic(F^{\rm std}_\ic)^\top)\|_2}\\ &= \sup_g\sqrt{L_1}\cdot O(1) = o_p(1)
\end{aligned}\] 
where the second step applies Assumption \ref{assu: more compatibility} where $F^{\rm std}_\ic = S_F^{-1/2}F_\ic$. (\ref{eq: target bound eigen}) then follows the bounds of $L_1$ and $L_2$. 
\end{proof}

\begin{proof}[Proof of Lemma \ref{lem: tilde r}] Note that
\[\sumn (\tilde r_i - \check r_i)^2 \leq \sumn q^{\prime\prime}(v_i^*)^2(v_i - \hat v_i)^4 + \sumn (q^\prime(\hat v_i) - \hat q^\prime(\hat v_i))^2\left[X_\ic^\top (\hat\varphi^\ind-\varphi ) + K_\ic^\top (\hat\kappa^\ind - \kappa) \right]^2.\]
By (\ref{eq: vhat error}), we deduce that 
\begin{equation}\label{eq: vhat error 4 power}
    \begin{aligned}
    \sumn (v_i - \hat v_i)^4 \lep \dfrac{n^3}{(\log (pM))^2}M^{-8\gamma} + \dfrac{(s+M)^4[\log (pM)]^2}{n} + n\cdot M^{-4\gamma} = o(1) 
\end{aligned}
\end{equation}
under the conditions for Theorem \ref{thm: limiting distribution}.  Besides, $q^{\prime\prime}$ is uniformly bounded when $\hat{\mathcal{V}}_n$ holds. Hence, 
\[\begin{aligned}
    \sumn q^{\prime\prime}(v_i^*)^2(v_i - \hat v_i)^4 \lep \sumn (v_i - \hat v_i)^4 = o_p(1).
\end{aligned}\]
It remains to show 
\[  \sumn (q^\prime(\hat v_i) - \hat q^\prime(\hat v_i))^2\left[X_\ic^\top (\hat\varphi^\ind-\varphi ) + K_\ic^\top (\hat\kappa^\ind - \kappa) \right]^2 = o_p\left(\frac{1}{(\log n)^2}\right).\]
By (\ref{eq: bound q hat approx}), 
\[\begin{aligned}
    (q^\prime(\hat v_i) - \hat q^\prime(\hat v_i))^2 &\lesssim \varDelta^q_{1i} + \varDelta^q_{2i} + \varDelta^q_{3i}  
\end{aligned}.\] 
It suffices to show that
\[ \supg \sumn  (\varDelta^q_{1i} + \varDelta^q_{2i} + \varDelta^q_{3i} )\left[X_\ic^\top (\hat\varphi^\ind-\varphi ) + K_\ic^\top (\hat\kappa^\ind - \kappa) \right]^2 = o_p\left(\frac{1}{(\log n)^2}\right).\]
Note that by Proposition \ref{prop: Lasso D},  
\[\begin{aligned}
     \supn \left[X_\ic^\top (\hat\varphi^\ind-\varphi ) + K_\ic^\top (\hat\kappa^\ind - \kappa) \right]^2 &\lesssim \|\hat\varphi^\ind-\varphi\|_1^2 \supn \|X_\ic\|_\infty^2 + \|\hat\kappa^\ind - \kappa\|_2^2 \cdot \supn \|K_\ic\|_2^2 \\
     &\lep c_{0n} := \dfrac{n}{\log (pM)} M^{-4\gamma} + \dfrac{(s+M)^2 \log (pM)}{n}.  
\end{aligned}\] 
By Assumption \ref{assu: asym more}, $c_{0n} = o(n^{-1/2}(\log n)^{-2})$. Furthermore, recall that $\varDelta^q_{j} = \supn\varDelta^q_{ji}$ defined right after (\ref{eq: bound q hat approx}). Equations (\ref{eq: Delta q1 bound}) and (\ref{eq: Delta q2 bound}) show that 
\[ \varDelta^q_1  
     \lep \dfrac{n}{\log (pM)} M^{-4\gamma} + \dfrac{(s+M)^2\log (pM)}{n} + M^{-2\gamma} = o(n^{-1/2})  \]
 and 
 \[ \varDelta^q_2  
     \lesssim M^{-2\gamma + 2} = o(n^{-1/2}).  \]
Note that $\varDelta^q_1$ and $\varDelta^q_2$ are not dependent on $g$. Thus, \[\begin{aligned}
     &\ \ \ \ \supg \sumn (\varDelta^q_{1i} + \varDelta^q_{2i}  )\left[X_\ic^\top (\hat\varphi^\ind-\varphi ) + K_\ic^\top (\hat\kappa^\ind - \kappa) \right]^2 \\
     &\lesssim n(\varDelta^q_1 + \varDelta^q_2)c_{0n}  = n\cdot o_p(n^{-1/2}) \cdot o_p(1/(\log n)^2) = o_p(1/(\log n)^2).
\end{aligned} \] 
For $\varDelta_{3i}^q$, by the definition in \eqref{eq: bound q hat approx}, 
\[\varDelta_{3i}^q \lesssim \|\hat\eta^\ind - \eta\|_2^2\cdot \supn \|\hat H^\prime_\ic - H^\prime_\ic\|_2^2 + \left|(\hat\eta^\ind - \eta)^\top H^\prime_\ic \right|^2 =: \varDelta_{3}^{(1)} + \varDelta_{3i}^{(2)} \]
If suffices to show 
\[\supg \sumn (\varDelta_{3}^{(1)} + \varDelta_{3i}^{(2)}  )\left[X_\ic^\top (\hat\varphi^\ind-\varphi ) + K_\ic^\top (\hat\kappa^\ind - \kappa) \right]^2 = o(1/(\log n)^2).\]
By \eqref{eq: sup rate omega} 
\begin{equation}\label{eq: rate omega simplified}
    \supg \|\hat\eta^\ind - \eta\|_2^2 \lep M^{-2\gamma + 1} + \dfrac{(s+M)^2 \log (pM)}{n} \lesssim \dfrac{(s+M)^2 \log (pM)}{n}
\end{equation} 
where the last inequality applies the fact that $M^{2\gamma+1} \gtrsim n$ under Assumption \ref{assu: asym more}. A similar argument applies to the last inequalities in the following two equations. By \eqref{eq: vhat error H prime},
\[\begin{aligned}
    \supn \|\hat H^\prime_\ic - H^\prime_\ic\|_2^2 &\lep M^4\left(\dfrac{n}{\logpM}M^{-4\gamma} + \dfrac{(s+M)^2\logpM}{n}+M^{-2\gamma}\right) \\
    &\lesssim \dfrac{M^4(s+M)^2\logpM}{n}. 
\end{aligned}\]
Besides, by (\ref{eq: Lasso D estimation error proof corollary}) 
\[ n^{-1}\sumn \left[X_\ic^\top (\hat\varphi^\ind-\varphi ) + K_\ic^\top (\hat\kappa^\ind - \kappa) \right]^2 \lep M^{-2\gamma} + \dfrac{(s+M)\log (pM)}{n} \lesssim \dfrac{(s+M)\log (pM)}{n}.\] 
Thus  
\[\begin{aligned}
    &\ \ \ \ \supg \sumn \varDelta_{3}^{(1)} \left[X_\ic^\top (\hat\varphi^\ind-\varphi ) + K_\ic^\top (\hat\kappa^\ind - \kappa) \right]^2\\
    &\lep \dfrac{(s+M)^2\logpM}{n}\cdot \dfrac{M^4(s+M)^2\logpM}{n} \cdot (s+M)\log (pM) \\
    &= \dfrac{M^4(s+M)^5[\logpM]^3}{n^2} \lesssim \dfrac{(s^9+M^9)[\logpM]^3}{n^2} = o(1/(\log n)^2)
\end{aligned}\]
where the $o(1/(\log n)^2)$ applies Assumption \ref{assu: asym more}. It remains to show 
\[\supg \sumn \left[ \varDelta^{(2)}_{3i} \left[X_\ic^\top (\hat\varphi^\ind-\varphi ) + K_\ic^\top (\hat\kappa^\ind - \kappa) \right]^2\right] = o_{p.g.}\left(\frac{1}{(\log n)^2}\right). \]
When we consider $\hat\eta^\ind$, $\hat\varphi^\ind$ and $\hat\kappa^\ind$ that generate the sigma-field $\mathcal{L}$ as given, $\varDelta^{(2)}_{3i}$ is a function of $v_i$, and $X_\ic^\top (\hat\varphi^\ind-\varphi ) + K_\ic^\top (\hat\kappa^\ind - \kappa) $ is a function of $X_\ic$ and $Z_\ic$. By independence between $(X_\ic^\top, Z_\ic^\top)^\top$ and $v_i$ (conditionally on $\mathcal{L}$), we deduce that 
\[n\E_{\mathcal{L}}\left( \varDelta^{(2)}_{3i} \left[X_\ic^\top (\hat\varphi^\ind-\varphi ) + K_\ic^\top (\hat\kappa^\ind - \kappa) \right]^2\right) = n\E_{\mathcal{L}}\left( \varDelta^{(2)}_{3i} \right)\E_{\mathcal{L}}\left[X_\ic^\top (\hat\varphi^\ind-\varphi ) + K_\ic^\top (\hat\kappa^\ind - \kappa) \right]^2.\] 
By Proposition \ref{eq: spline eigen} and \eqref{eq: rate omega simplified}, 
\[n\cdot \supg \E_{\mathcal{L}}\left( \varDelta^{(2)}_{3i} \right)\lep n\|\hat\eta^\ind - \eta\|_2^2 \cdot \lambda_{\max}(\E(H^\prime_\ic {H^\prime_\ic}^\top )) \lep (s+M)^2\logpM\cdot M \] 
In addition, by \eqref{eq: conditional Exp vhat} and \eqref{eq: corollary proof vhat error very beginning}
\[\begin{aligned}
    \E_{\mathcal{L}}\left[X_\ic^\top (\hat\varphi^\ind-\varphi ) + K_\ic^\top (\hat\kappa^\ind - \kappa) \right]^2  
    &\lep M^{-2\gamma} + \dfrac{(s+M)\logpM}{n} \lesssim \dfrac{(s+M)\logpM}{n}
\end{aligned}\]
where the last inequality applies the rate of $M$ in Assumption \ref{assu: asym more}.  Thus, 
\[\begin{aligned}
     \supg n\E_{\mathcal{L}}\left( \varDelta^{(2)}_{3i} \left[X_\ic^\top (\hat\varphi^\ind-\varphi ) + K_\ic^\top (\hat\kappa^\ind - \kappa) \right]^2\right) \lep \dfrac{(s+M)^4(\logpM)^2}{n} = o\left(\dfrac{1}{(\log n)^3}\right)
\end{aligned}\]
where the $o(1)$ applies Assumption \ref{assu: asym more}. By Markove inequality, 
\[\begin{aligned}
    &\ \ \ \ \supg\Pr\left\{\left| \sumn \left[ \varDelta^{(2)}_{3i} \left[X_\ic^\top (\hat\varphi^\ind-\varphi ) + K_\ic^\top (\hat\kappa^\ind - \kappa) \right]^2\right]\right| > (\log n)^{-2.5}\Bigg|\mathcal{L}\right\} \\
    &\leq (\log n)^{2.5}\supg n\E_{\mathcal{L}}\left( \varDelta^{(2)}_{3i} \left[X_\ic^\top (\hat\varphi^\ind-\varphi ) + K_\ic^\top (\hat\kappa^\ind - \kappa) \right]^2\right) = o_p(1). 
\end{aligned}\]
Then 
\[\begin{aligned}
    &\ \ \ \ \supg \Pr\left\{\left|  \sumn \left[ \varDelta^{(2)}_{3i} \left[X_\ic^\top (\hat\varphi^\ind-\varphi ) + K_\ic^\top (\hat\kappa^\ind - \kappa) \right]^2\right]\right| > (\log n)^{-2.5} \right\} \\
    &= \supg \E\left(\Pr\left\{\left|  \sumn \left[ \varDelta^{(2)}_{3i} \left[X_\ic^\top (\hat\varphi^\ind-\varphi ) + K_\ic^\top (\hat\kappa^\ind - \kappa) \right]^2\right]\right| > (\log n)^{-2.5} \Bigg|\mathcal{L}\right\}\right) \\
    &\leq  \E\left(\supg\Pr\left\{\left|  \sumn \left[ \varDelta^{(2)}_{3i} \left[X_\ic^\top (\hat\varphi^\ind-\varphi ) + K_\ic^\top (\hat\kappa^\ind - \kappa) \right]^2\right]\right| > (\log n)^{-2.5} \Bigg|\mathcal{L}\right\}\right) \\ 
    &\to 0
\end{aligned}\]
where the limit applies the Bounded Convergence Theorem to supreme of the conditional probability as a random variable. Thus, $\sumn \left[ \varDelta^{(2)}_{3i} \left[X_\ic^\top (\hat\varphi^\ind-\varphi ) + K_\ic^\top (\hat\kappa^\ind - \kappa) \right]^2\right] = o_{p.g.}(1)$. We complete the proof of Lemma \ref{lem: tilde r}. 
\end{proof}

\begin{proof}[Proof of Lemma \ref{lem: gaussian approx}]
     We only show the Gaussian approximation error for $\mathbb{G}$ and the same arguments apply to $\mathbb{G}^\prime$. Let $\mathcal{G}$ be the sigma-field generated by $\{X_\ic,Z_\ic,v_i\}_{i\in[n]}$ and the LASSO estimators $\hat\kappa^{\ind}$, $\hat\varphi^{\ind}$, $\hat\eta^{\ind}$. Let $d_j^\varDelta = a_D + j\varDelta_J$ for $j=0,1,\cdots,J$  with some integer $J$ large enough, where $\varDelta_J := \frac{b_D - a_D}{J}$. Recall $\mathcal{D}=[a_D,b_D]$ and define  $\mathcal{D}^\varDelta:=\{d_j^\varDelta\}_{0\leq j\leq J}$. Define 
 \[ \mathbb{G} := \sup_{d\in\mathcal{D}} \sumn \zeta_i(d),\ \mathbb{G}^\varDelta := \sup_{d\in\mathcal{D}^\varDelta} \sumn \zeta_i(d),\]
  \[ \mathbb{G}^* := \sup_{d\in\mathcal{D}} \sumn \zeta_i^*(d),\ \mathbb{G}^{*\varDelta} := \sup_{d\in\mathcal{D}^\varDelta} \sumn \zeta_i^*(d)\]
  and using triangular inequalities 
  \begin{equation} \label{eq: G decompose}
      | \mathbb{G} - \mathbb{G}^*| \leq |\mathbb{G} - \mathbb{G}^\varDelta| + |\mathbb{G}^* - \mathbb{G}^{*\varDelta}| +  |\mathbb{G}^\varDelta - \mathbb{G}^{*\varDelta}|.
  \end{equation} 
  \underline{Bound $|\mathbb{G} - \mathbb{G}^\varDelta| + |\mathbb{G}^* - \mathbb{G}^{*\varDelta}|$}. Observe that $|\mathbb{G} - \mathbb{G}^\varDelta| \leq \sup_{|d-d^\prime| \leq \varDelta_J} |\sumn \left(\xi_i(d) - \xi_i(d^\prime)\right) |$ and we further deduce that 
 \begin{equation}\label{eq: GG bound 1}
     \begin{aligned}
      |\mathbb{G} - \mathbb{G}^\varDelta| &\leq \sup_{|d-d^\prime| \leq \varDelta_J}\|\hat s(d)^{-1}B^\prime(d) - \hat s(d^\prime)^{-1}B^\prime(d^\prime)\|_1 \cdot \left\|\sumn \hat\Omega_B \hat F_\ic \varepsilon_i \right\|_\infty. 
  \end{aligned}
 \end{equation} 
  We first bound the supreme of $L_1$ norm. We have $\|B^\prime(d)\|_1$ by Proposition \ref{prop: SplineDerivativeL2Norm}, $\hat s(d) \asymppg M^{1.5}$ by (\ref{eq: variance scale}). Then 
  \[\begin{aligned}
      &\ \ \ \ \|\hat s(d)^{-1}B^\prime(d) - \hat s(d^\prime)^{-1}B^\prime(d^\prime)\|_1 \\
      &\leq \|B^\prime(d)\|_1\cdot\dfrac{\left|\hat s(d^\prime)^2 - \hat s(d)^2\right|}{\hat s(d)\hat s(d^\prime)[\hat s(d) + \hat s(d^\prime)]} + \dfrac{\|B^\prime(d) - B^\prime(d^\prime)\|_1}{\hat s(d^\prime)} \\
      &\lepg \dfrac{M\cdot \left|\hat s(d^\prime)^2 - \hat s(d)^2\right| }{M^{4.5}} + \dfrac{\|B^\prime(d) - B^\prime(d^\prime)\|_1}{M^{1.5}}\\
      &=\dfrac{|(B^\prime(d) - B^\prime(d^\prime))^\top\hat\Omega_B \hat\Sigma_F \hat\Omega_B^\top (B^\prime(d) + B^\prime(d^\prime))|}{M^{3.5}} + \dfrac{\|B^\prime(d) - B^\prime(d^\prime)\|_1}{M^{1.5}} \\
      &=\left|\dfrac{\|B^\prime(d) - B^\prime(d^\prime)\|_1\cdot \|\hat\Omega_B \hat\Sigma_F \hat\Omega_B^\top\|_\infty\cdot \|B^\prime(d) + B^\prime(d^\prime)\|_1}{M^{3.5}}\right| + \dfrac{\|B^\prime(d) - B^\prime(d^\prime)\|_1}{M^{1.5}} \\
      &\lepg \dfrac{\|B^\prime(d) - B^\prime(d^\prime)\|_1}{M^{1.5}}.
   \end{aligned}\]
    By (\ref{eq: L1 bound Bd0 m Bd1}), 
   \[\|\hat s(d)^{-1}B^\prime(d) - \hat s(d^\prime)^{-1}B^\prime(d^\prime)\|_1 \lepg M^{0.5}|d - d^\prime|.\]
      and thus 
   \begin{equation}\label{eq: GG bound 2}
       \begin{aligned}
      |\mathbb{G} - \mathbb{G}^\varDelta| &\lepg  M^{0.5}\sup_{|d-d^\prime|\leq \varDelta_J}|d - d^\prime|\left\|\sumn \hat\Omega_B \hat F_\ic \varepsilon_i \right\|_\infty \leq M^{0.5}\varDelta_J\left\|\sumn \hat\Omega_B \hat F_\ic \varepsilon_i \right\|_\infty. 
  \end{aligned}
   \end{equation}
 \eqref{eq: GG bound 2} together with \eqref{eq: sup norm OFe bound}
   \begin{equation}\label{eq: diff G GDelta}
       |\mathbb{G} - \mathbb{G}^\varDelta| \lepg \sqrt{n\logpM}M^{1.5}\cdot \Delta_J.
   \end{equation} 
   Following exactly the same arguments,
   \begin{equation}\label{eq: diff G star GStarDelta}
       |\mathbb{G}^* - \mathbb{G}^{*\varDelta}| \lepg \sqrt{n\logpM}M^{1.5}\cdot \Delta_J.
   \end{equation} 
\underline{Bound $| \mathbb{G}^\varDelta - \mathbb{G}^{*\varDelta}|$}. Using \citet[Corollary 4.1]{chernozhukov2014gaussian}, for any $t>0$ 
\begin{equation}\label{eq: Chernozhukov bound}
    \begin{aligned}
    \Pr\left\{| \mathbb{G}^\varDelta - \mathbb{G}^{*\varDelta}| > 16 t \Bigg|\mathcal{G}\right]\} \lesssim \dfrac{b_1\log (J\vee n)}{t^2} + \dfrac{(b_2+b_4)[\log (J\vee n)]^2}{t^3} + \dfrac{\log n}{n} 
\end{aligned}
\end{equation} 
where 
\[\begin{aligned}
    b_1 &:= \E_{\mathcal{G}}\left[\max_{d,d^\prime\in\mathcal{D}^\varDelta}\left|\sumn(\zeta_{i}(d)\zeta_{i}(d^\prime) - \E_{\mathcal{G}}[\zeta_{i}(d)\zeta_{i}(d^\prime)])\right|\right],
\end{aligned}\]
\[\begin{aligned}
    b_2 &:= \E_{\mathcal{G}}\left[\max_{d\in\mathcal{D}^\varDelta}\sumn\left|\zeta_i(d)\right|^3\right]  
\end{aligned}\]
and 
\[\begin{aligned}
    b_4 &:= \sumn\E_{\mathcal{G}}\left[\max_{d\in\mathcal{D}^\varDelta}\left|\zeta_i(d)\right|^3\cdot 1\left(\max_{d\in\mathcal{D}^\varDelta}\left|\zeta_i(d)\right|>t/\log (J\vee n)\right)\right]  
\end{aligned}\]
Note that 
 $\hat F, \hat\Omega, \hat s(d)\in\mathcal{G}$ and 
 \begin{equation}\label{eq: bound zeta i d0}
      |\zeta_{i}(d)| = \left|  \dfrac{B^\prime(d)^\top\hat\Omega_B \hat F_\ic\varepsilon_i}{\hat s(d)}\right| \leq  \dfrac{\|B^\prime(d)\|_1\cdot \|\hat F\hat\Omega_B^\top\|_\infty\cdot |\varepsilon_i|}{\hat s(d)},
 \end{equation} 
 and 
 \[\sup_g \sup_{d\in\mathcal{D}} \dfrac{\|B^\prime(d)\|_1\cdot \|\hat F\hat\Omega_B^\top\|_\infty }{\hat s(d)} \lep \sqrt{M\logpM }\]
by the fact that $\sup_{d\in\mathcal{D}}\|B^\prime(d)\|_1\lesssim M$ by Proposition \ref{prop: SplineDerivativeL2Norm}, $\sup_g \|\hat F\hat\Omega_B^\top\|_\infty\lesssim M\sqrt{\logpM / n}$ by the definition of $\hat\Omega_B^\top$ in (\ref{eq: projection direction 1}), and the lower bound for $\hat s(d)$ for (\ref{eq: scale result lower}). 
Thus,  
\[\begin{aligned}
  \sup_g  b_2 &\leq \sup_g \sup_{d\in\mathcal{D}} \left(\dfrac{\|B^\prime(d)\|_1\cdot \|\hat F\hat\Omega_B^\top\|_\infty}{\hat s(d)}\right)^3\sumn \E_{\mathcal{G}}|\varepsilon_i^3|  \lesssim n\cdot (M\logpM)^{1.5},
\end{aligned}\]
and similarly 
\[\sup_g b_4 \lep n\cdot (M\logpM)^{1.5}.\]
We then bound $b_1$. Using the Chebyshev's inequality, for any $d,d^\prime\in\mathcal{D}^\varDelta$ and $\tau >0$,  
\[\begin{aligned}    &\ \ \ \ \supg \Pr\left( n^{-1/2}\left|\sumn(\zeta_{i}(d)\zeta_{i}(d^\prime) - \E_{\mathcal{G}}[\zeta_{i}(d)\zeta_{i}(d^\prime)])\right|>\tau \Bigg|\mathcal{G}\right) \\
&\leq \dfrac{1}{n\tau^{2}} \supg \E_{\mathcal{G}}\left[ \left|\sumn(\zeta_{i}(d)\zeta_{i}(d^\prime) - \E_{\mathcal{G}}[\zeta_{i}(d)\zeta_{i}(d^\prime)])\right|^{2}\right] \\
&= \dfrac{1}{n\tau^{2}}\supg \sumn \E_{\mathcal{G}} \left(   \zeta_{i}(d)\zeta_{i}(d^\prime) - \E_{\mathcal{G}}[\zeta_{i}(d)\zeta_{i}(d^\prime)] \right)^{2}  \\
&\leq \dfrac{1}{n\tau^{2}}\supg  \sumn\left(\E_{\mathcal{G}}\left[ \zeta_{i}(d)^2\zeta_{i}(d^\prime)^2\right] \right) \\
&\lesssim \sup_g \sup_{d\in\mathcal{D}} \left(\dfrac{\|B^\prime(d)\|_1\cdot \|\hat F\hat\Omega_B^\top\|_\infty}{\hat s(d)}\right)^2 n^{-1}\supg  \sumn \E_{\mathcal{G}}(\xi_i(d^\prime)^2\varepsilon_i^2) \\ &\lep \dfrac{M[\logpM]}{\tau^2},\\
\end{aligned}\]
where the fourth row applies (\ref{eq: bound zeta i d0}), and the last step applies \[
\begin{aligned}
    \sumn \E_{\mathcal{G}}(\xi_i(d^\prime)^2\varepsilon_i^2) &= \sumn \dfrac{(B^\prime(d)^\top \hat\Omega_B\hat F_\ic )^2}{\hat s(d)^2 }\E_{\mathcal{G}}(\varepsilon_i^4) \lesssim \sumn \dfrac{(B^\prime(d)^\top \hat\Omega_B\hat F_\ic )^2}{\hat s(d)^2 } = n 
\end{aligned}\] 
where the second step applies the bounded fourth condition moment of $\varepsilon_i$, and the last step applies the definition of $\hat s(d)$. By union bound, \[\begin{aligned}    &\supg \Pr\left( n^{-1/2}\max_{d,d^\prime\in\mathcal{D}^\varDelta}\left|\sumn(\zeta_{i}(d)\zeta_{i}(d^\prime) - \E_{\mathcal{G}}[\zeta_{i}(d)\zeta_{i}(d^\prime)])\right|>\tau \Bigg|\mathcal{G}\right) &\lep \dfrac{J^2 M[\logpM]}{\tau^2}\\
\end{aligned}\]
and thus 
\[\begin{aligned}    &\ \ \ \ \supg \Pr\left( (nJ^2M)^{-1/2}(\logpM)^{-1}\max_{d,d^\prime\in\mathcal{D}^\varDelta}\left|\sumn(\zeta_{i}(d)\zeta_{i}(d^\prime) - \E_{\mathcal{G}}[\zeta_{i}(d)\zeta_{i}(d^\prime)])\right|>\tau \Bigg|\mathcal{G}\right) \\
&\lep \dfrac{J^2 M[\logpM]^2}{ (\sqrt{J^2 M}\logpM\tau)^2} = \dfrac{1}{\tau^2}.\\
\end{aligned}\]
By $\E(x) = \int_0^\infty \Pr(x > e)de$ for any random variable $x$, 
\[\begin{aligned}
    &\ \ \ \ (nJ^2M)^{-1/2}(\logpM)^{-1}\cdot \supg b_1 \\
    &\leq \int_0^\infty \supg \Pr\left( (nJ^2M)^{-1/2}(\logpM)^{-1}\max_{d,d^\prime\in\mathcal{D}^\varDelta}\left|\sumn(\zeta_{i}(d)\zeta_{i}(d^\prime) - \E_{\mathcal{G}}[\zeta_{i}(d)\zeta_{i}(d^\prime)])\right|>\tau \Bigg|\mathcal{G}\right)  d\tau \\
    &\lep 1 + \int_1^\infty \dfrac{1}{\tau^2} d\tau = 2
\end{aligned}\]
and hence $\supg b_1\lep \sqrt{nM}J\logpM$. Let $t = \sqrt{n}t_n$. Using \eqref{eq: Chernozhukov bound} and the bounds for $b_1$, $b_2$, and $b_4$, we have 
\[\begin{aligned}
   \supg  \Pr\left\{n^{-1/2}|\mathbb{G}^\varDelta-\mathbb{G}^{*\varDelta}|>16t_n\Bigg|\mathcal{G}\right\} \lep \dfrac{\sqrt{M}J\log (J\vee n)}{\sqrt{n}t_n^2} + \frac{(M\logpM)^{1.5}[\log (J\vee n)]^2}{\sqrt{n}t_n^3} + \frac{\log n}{n}.
\end{aligned}\]
Let $J \asymp M^{1.5}\log n$ and $t_n = (\log n)^{-1/2}$, we have 
\[\begin{aligned}
   \supg \Pr\left\{n^{-1/2}|\mathbb{G}^\varDelta-\mathbb{G}^{*\varDelta}|>\frac{16}{\sqrt{\log n}}\Bigg|\mathcal{G}\right\} \lep \dfrac{M^2(\log n)^2}{\sqrt{n}} + \frac{M^{1.5}[\logpM]^5}{\sqrt{n}} + \frac{\log n}{n} = o(1).
\end{aligned}\]
Then 
\[\begin{aligned}
    \supg \Pr\left\{n^{-1/2}|\mathbb{G}^\varDelta-\mathbb{G}^{*\varDelta}|>\frac{16}{\sqrt{\log n}}\right\} &\leq \E\left(\supg \Pr\left\{n^{-1/2}|\mathbb{G}^\varDelta-\mathbb{G}^{*\varDelta}|>\frac{16}{\sqrt{\log n}}\Bigg|\mathcal{G}\right\} \right)
    \\&\to 0
\end{aligned}\] 
where the limit uses the Bounded Convergence Theorem for the supreme of conditional probability. Thus 
\begin{equation}\label{eq: G bound 2}
    n^{-1/2}|\mathbb{G}^\varDelta-\mathbb{G}^{*\varDelta}| \lepg \frac{1}{\sqrt{\log n}}.
\end{equation}
 Also by (\ref{eq: diff G GDelta}) and (\ref{eq: diff G star GStarDelta}) we have \begin{equation}\label{eq: G bound 1}
    n^{-1/2}|\mathbb{G} - \mathbb{G}^\varDelta| + n^{-1/2}|\mathbb{G}^* - \mathbb{G}^{*\varDelta}| \lepg \dfrac{1}{\sqrt{\log n}}.
\end{equation} 
Thus, by (\ref{eq: G decompose}), (\ref{eq: G bound 1}) and (\ref{eq: G bound 2})
\begin{equation}\label{eq: diff G target}
    n^{-1/2}|\mathbb{G}-\mathbb{G}^*| \lepg \frac{1}{\sqrt{\log n}}.
\end{equation}
We complete the proof of Lemma \ref{lem: gaussian approx}.
\end{proof}

\section{Additional Simulation Results}
\label{app: additional simul}
\par This section includes the robustness simulation results omitted from the main text. The results in this section all use full-sample inference as described in the main text. 
\par As mentioned in Section \ref{sec: sim} of the main text, we first check the performance of our methodology under different choices of $M_D$. As shown in Tables \ref{tab:simul_M} and \ref{tab:simul_M_NL}, a larger $M_D$ does not benefit the inference in terms of coverage but produces a wider confidence band due to additional variances from a larger dimension of B-Spline functions. It shows that $M_D = 5$ is a reasonable choice. 
\par Next, we check the robustness of our method under different types of distributions that violate the compactness assumption. We set $U_{ji}\sim N(0,12^{-1/2})$ and $v_i\sim N(0,1)$ to maintain the same variances of the data from bounded distributions in the main text. Other settings are unchanged. Table \ref{tab:simul_UB} shows that the performance of the proposed method is robust to violations of compact supports. 
% \usepackage{multirow}
\begin{table}[H]
\begin{center} 
\caption{Simulation Results for Linear $g$ Functions with different $M_d$.}
\label{tab:simul_M}
\begin{tabular}{cc|cccc}
\hline\hline 
 {$M_d$}  &  {$n$} &    BiasInit & BiasDB & Coverage & Length  \\
                    \hline 
                     \multicolumn{6}{c}{$g(d) = 0$, $g^\prime(d)=0$}  \\    
                     \hline 
\multirow{4}{*}{5}  & 500  & 0.047 & 0.012 & 0.956 & 1.319 \\
                    & 1000 & 0.065 & 0.011 & 0.954 & 0.695 \\
                    & 2000 & 0.050 & 0.008 & 0.954 & 0.434 \\
                    & 3000 & 0.050 & 0.003 & 0.964 & 0.342 \\
                     \hline 
\multirow{4}{*}{10} & 500  & 0.047 & 0.025 & 0.950 & 3.740 \\
                    & 1000 & 0.067 & 0.014 & 0.966 & 1.952 \\
                    & 2000 & 0.050 & 0.011 & 0.956 & 1.227 \\
                    & 3000 & 0.049 & 0.009 & 0.958 & 0.971 \\
                     \hline 
\multirow{4}{*}{15} & 500  & 0.052 & 0.041 & 0.954 & 7.989 \\
                    & 1000 & 0.066 & 0.020 & 0.940 & 4.158 \\
                    & 2000 & 0.053 & 0.021 & 0.972 & 2.626 \\
                    & 3000 & 0.049 & 0.014 & 0.966 & 2.083 \\
                    \hline 
                     \multicolumn{6}{c}{$g(d) = d$, $g^\prime(d)=1$}  \\    
                     \hline 
\multirow{4}{*}{5}  & 500  & 0.052 & 0.014 & 0.956 & 1.337 \\
                    & 1000 & 0.056 & 0.006 & 0.956 & 0.689 \\
                    & 2000 & 0.053 & 0.012 & 0.938 & 0.435 \\
                    & 3000 & 0.044 & 0.010 & 0.948 & 0.341 \\
                     \hline 
\multirow{4}{*}{10} & 500  & 0.052 & 0.019 & 0.968 & 3.782 \\
                    & 1000 & 0.052 & 0.012 & 0.952 & 1.950 \\
                    & 2000 & 0.055 & 0.011 & 0.956 & 1.234 \\
                    & 3000 & 0.045 & 0.011 & 0.956 & 0.974 \\
                     \hline 
\multirow{4}{*}{15} & 500  & 0.056 & 0.042 & 0.960 & 8.076 \\
                    & 1000 & 0.066 & 0.031 & 0.952 & 4.154 \\
                    & 2000 & 0.056 & 0.013 & 0.964 & 2.642 \\
                    & 3000 & 0.044 & 0.014 & 0.956 & 2.086 \\
                      \hline  \hline 
\end{tabular}
\end{center}
{\footnotesize Note: ``BiasInit'' and ``BiasDB'' denote the average bias of the initial Lasso estimator $\hat g^\prime(\cdot)$ and the bias-corrected estimator $\tilde g^\prime(\cdot)$, respectively. ``Coverage'' shows the coverage probability of the 95\% confidence band defined as \eqref{eq: def confidence band} over 500 replications. ``Length'' stands for the point-wise average length of the confidence band.}
\end{table}
 % defined as $ |\mathcal{D}_\tau|^{-1}\sum_{d_0\in\mathcal{D}_\tau}\E\left|\tilde g^\prime(\cdot) - g^\prime(d_0)\right|$ and $ |\mathcal{D}_\tau|^{-1}\sum_{d_0\in\mathcal{D}_\tau}\E\left|\hat g^\prime(\cdot) - g^\prime(d_0)\right|$,

 %%%%%%%%%%%%%%%%%%%%%%%%%%%%%%%%%%%%%%%%%%%%%%%%%%%%%%%%%%%%%

% \usepackage{multirow}
\begin{table}[H]
\begin{center} 
\caption{Simulation Results for 
Nonlinear $g$ Functions with different $M_d$.}
\label{tab:simul_M_NL}
\begin{tabular}{cc|cccc}
\hline\hline 
 {$M_d$}  &  {$n$} &    BiasInit & BiasDB & Coverage & Length  \\
                    \hline 
                     \multicolumn{6}{c}{$g(d) = 0.05(d-3)^2$, $g^\prime(d)=0.1(d-3)$}  \\    
                     \hline  
                    \multirow{4}{*}{5} & 500  & 0.058 & 0.011 & 0.962 & 1.341 \\
                    & 1000 & 0.062 & 0.012 & 0.966 & 0.693 \\
                    & 2000 & 0.053 & 0.010 & 0.962 & 0.435 \\
                    & 3000 & 0.043 & 0.008 & 0.934 & 0.341 \\
                      \hline 
\multirow{4}{*}{10} & 500  & 0.058 & 0.018 & 0.968 & 3.775 \\
                    & 1000 & 0.061 & 0.019 & 0.956 & 1.951 \\
                    & 2000 & 0.057 & 0.008 & 0.966 & 1.232 \\
                    & 3000 & 0.044 & 0.010 & 0.932 & 0.969 \\
                      \hline 
\multirow{4}{*}{15} & 500  & 0.062 & 0.023 & 0.964 & 8.043 \\
                    & 1000 & 0.063 & 0.026 & 0.970 & 4.160 \\
                    & 2000 & 0.057 & 0.019 & 0.956 & 2.636 \\
                    & 3000 & 0.046 & 0.015 & 0.952 & 2.078 \\
                      \hline 
                     \multicolumn{6}{c}{$g(d) = 0.02(d-3)^3$, $g^\prime(d) = 0.06(d-3)^2$}  \\    
                     \hline
\multirow{4}{*}{5}  & 500  & 0.055 & 0.015 & 0.968 & 1.336 \\
                    & 1000 & 0.061 & 0.007 & 0.958 & 0.693 \\
                    & 2000 & 0.055 & 0.007 & 0.946 & 0.436 \\
                    & 3000 & 0.047 & 0.006 & 0.960 & 0.341 \\
                      \hline 
\multirow{4}{*}{10} & 500  & 0.052 & 0.016 & 0.956 & 3.783 \\
                    & 1000 & 0.062 & 0.007 & 0.968 & 1.952 \\
                    & 2000 & 0.057 & 0.009 & 0.966 & 1.233 \\
                    & 3000 & 0.048 & 0.005 & 0.950 & 0.970 \\
                      \hline 
\multirow{4}{*}{15} & 500  & 0.061 & 0.048 & 0.956 & 8.053 \\
                    & 1000 & 0.067 & 0.026 & 0.954 & 4.155 \\
                    & 2000 & 0.061 & 0.016 & 0.966 & 2.637 \\
                    & 3000 & 0.048 & 0.008 & 0.948 & 2.080 \\
                      \hline  \hline 
\end{tabular}
\end{center}
{\footnotesize Note: ``BiasInit'' and ``BiasDB'' denote the average bias of the initial Lasso estimator $\tilde g^\prime(\cdot)$ and the bias-corrected estimator $\hat g^\prime(\cdot)$, respectively. ``Coverage'' shows the coverage probability of the 95\% confidence band defined as \eqref{eq: def confidence band} over 500 replications. ``Length'' stands for the point-wise average length of the confidence band.}
\end{table}
 % defined as $ |\mathcal{D}_\tau|^{-1}\sum_{d_0\in\mathcal{D}_\tau}\E\left|\tilde g^\prime(\cdot) - g^\prime(d_0)\right|$ and $ |\mathcal{D}_\tau|^{-1}\sum_{d_0\in\mathcal{D}_\tau}\E\left|\hat g^\prime(\cdot) - g^\prime(d_0)\right|$,
\begin{table}[H]
\begin{center} 
\caption{Simulation Results with Unbounded Supports.}
\label{tab:simul_UB}
\begin{tabular}{c|cccc}
\hline \hline 
 {$n$} &    BiasInit & BiasDB & Coverage & Length  \\
  \hline 
                     \multicolumn{5}{c}{$g(d) = 0$, $g^\prime(d)=0$}  \\    
                     \hline  
500                & 0.048                                             & 0.011                                               & 0.956                                         & 1.170                                          \\

1000               & 0.055                                             & 0.007                                               & 0.948                                         & 0.528                                          \\

2000               & 0.054                                             & 0.004                                               & 0.942                                         & 0.320                                          \\

3000               & 0.046                                             & 0.001                                               & 0.932                                         & 0.245                                          \\
  \hline 
                     \multicolumn{5}{c}{$g(d) = d$, $g^\prime(d)=1$}  \\    
                     \hline 
500                & 0.063                                             & 0.005                                               & 0.946                                         & 1.160                                          \\

1000               & 0.061                                             & 0.004                                               & 0.958                                         & 0.531                                          \\

2000               & 0.053                                             & 0.002                                               & 0.950                                         & 0.317                                          \\

3000               & 0.046                                             & 0.002                                               & 0.950                                         & 0.243                                          \\
 \hline 
                     \multicolumn{5}{c}{$g(d) = 0.05(d-3)^2$, $g^\prime(d)=0.1(d-3)$}  \\    
                     \hline
500                & 0.053                                             & 0.013                                               & 0.952                                         & 1.143                                          \\

1000               & 0.061                                             & 0.002                                               & 0.962                                         & 0.532                                          \\

2000               & 0.050                                             & 0.003                                               & 0.960                                         & 0.316                                          \\

3000               & 0.045                                             & 0.003                                               & 0.934                                         & 0.243                                          \\
\hline 
                     \multicolumn{5}{c}{$g(d) = 0.02(d-3)^3$, $g^\prime(d) = 0.06(d-3)^2$}  \\    
                     \hline
500                & 0.053                                             & 0.005                                               & 0.966                                         & 1.167                                          \\

1000               & 0.055                                             & 0.004                                               & 0.944                                         & 0.533                                          \\

2000               & 0.052                                             & 0.002                                               & 0.942                                         & 0.318                                          \\

3000               & 0.044                                             & 0.002                                               & 0.956                                         & 0.247              \\
\hline \hline 
\end{tabular}
\end{center}
{\footnotesize Note: ``BiasInit'' and ``BiasDB'' denote the average bias of the initial Lasso estimator $\hat g^\prime(\cdot)$ and the bias-corrected estimator $\tilde g^\prime(\cdot)$, respectively. ``Coverage'' shows the coverage probability of the 95\% confidence band defined as \eqref{eq: def confidence band} over 500 replications. ``Length'' stands for the point-wise average length of the confidence band.}
% defined as $ |\mathcal{D}_\tau|^{-1}\sum_{d_0\in\mathcal{D}_\tau}\E\left|\tilde g^\prime(\cdot) - g^\prime(d_0)\right|$ and $ |\mathcal{D}_\tau|^{-1}\sum_{d_0\in\mathcal{D}_\tau}\E\left|\hat g^\prime(\cdot) - g^\prime(d_0)\right|$,
\end{table}

\end{document}